\pdfoutput=1
\RequirePackage{ifpdf}
\ifpdf 
\documentclass[pdftex]{sigma}
\else
\documentclass{sigma}
\fi

\numberwithin{equation}{section}

\newtheorem{Theorem}{Theorem}[section]
\newtheorem{Corollary}[Theorem]{Corollary}
\newtheorem{Lemma}[Theorem]{Lemma}
\newtheorem{Proposition}[Theorem]{Proposition}
 { \theoremstyle{definition}
\newtheorem{Definition}[Theorem]{Definition}
\newtheorem{Remark}[Theorem]{Remark} }

\newcommand{\Ja}{\mathcal{J}}

\usepackage{tikz-cd}
\usepackage{tikz}
\begin{document}
\allowdisplaybreaks

\newcommand{\arXivNumber}{1907.01436}

\renewcommand{\PaperNumber}{022}

\FirstPageHeading

\ShortArticleName{The Differential Geometry of the Orbit Space of Extended Affine Jacobi Group $A_1$}

\ArticleName{The Differential Geometry of the Orbit Space\\ of Extended Affine Jacobi Group $\boldsymbol{A_1}$}

\Author{Guilherme F.~ALMEIDA}

\AuthorNameForHeading{G.F.~Almeida}

\Address{SISSA, via Bonomea 265, Trieste, Italy}
\Email{\href{mailto:galmeida@sissa.it}{galmeida@sissa.it}}

\ArticleDates{Received May 30, 2020, in final form February 11, 2021; Published online March 09, 2021}

\Abstract{We define certain extensions of Jacobi groups of $A_1$, prove an analogue of Chevalley theorem for their invariants, and construct a Dubrovin--Frobenius structure on its orbit space.}


\Classification{53D45}

\begin{flushright}
\it Dedicated to the memory of Professor Boris Dubrovin
\end{flushright}


\section{Introduction}
 Dubrovin--Frobenius manifold is a geometric interpretation of a remarkable system of differential equations called WDVV equations~\cite{B.Dubrovin2}. Since early nineties, there has been a~continuous exchange of ideas from fields that are not trivially related to each other, such as topological quantum field theory, non-linear waves, singularity theory, random matrices theory, integrable systems, and Painlev\'e equations. Dubrovin--Frobenius manifolds theory is a bridge between them.

\subsection{Orbit space of reflection groups and its extensions}
In~\cite{B.Dubrovin2}, Dubrovin pointed out that WDVV solutions with certain good analytic properties are rela\-ted with partition functions of TFT. Afterwards, Dubrovin conjectured that WDVV solutions with certain good analytic properties are in one to one correspondence with discrete groups. This conjecture is supported by ideas which come from singularity theory, because in this setting there exists an integrable systems/discrete group correspondence. Furthermore, in minimal models, such as Gepner chiral rings, there exists a correspondence between physical models and discrete groups.
In~\cite{C.Hertling}, Hertling proved that a particular class of Dubrovin--Frobenius manifold, called polynomial Dubrovin--Frobenius manifold, is isomorphic to the orbit space of a~finite Coxeter group, which are spaces such that their geometric structure is invariant under the finite Coxeter group.
In~\cite{BertolaM.1,BertolaM.2, B.Dubrovin2, B.Dubrovin1, B.DubrovinI.A.B.StrachanY.ZhangD.Zuo, B.DubrovinandY.Zhang, D.Zuo}, there are many examples of WDVV solutions that are associated with orbit spaces of natural extensions of finite Coxeter groups, such as
extended affine Weyl groups, and Jacobi groups. Therefore, the construction of Dubrovin--Frobenius manifolds on~orbit space of reflection groups and its extensions is a prospective project of the classification of WDVV solutions. In addition, WDDV solutions arising from orbit spaces may also have some applications in TFT or some combinatorial problem, because previously these relationships were demonstrated in some examples, such as the orbit space of the finite Coxeter group $A_1$, and the extended affine Weyl group $A_1$~\cite{Dubrovin3,B.DubrovinandY.Zhang2}.

\subsection{Hurwtiz space/orbit space correspondence}
There are several other non-trivial connections that Dubrovin--Frobenius manifolds theory can make. For example, Hurwitz spaces is the one of the main sources of examples of Dubrovin--Frobenius manifolds. Hurwitz spaces $H_{g,n_0,n_1,\dots,n_m}$ are moduli space of covering over $\mathbb{CP}^1$ with a fixed ramification profile. More specifically, $H_{g,n_0,n_1,\dots,n_m}$ is moduli space of pairs
\begin{gather*}
\big\{C_g ,\lambda\colon C_g\mapsto \mathbb{CP}^1\big\},
\end{gather*}
where $C_g$ is a compact Riemann surface of genus $g$ and $\lambda$ is meromorphic function with poles in
\begin{gather*}
\lambda^{-1}(\infty)=\{\infty_0,\infty_1,\dots,\infty_{m} \}.
\end{gather*}
 Moreover, $\lambda$ has degree $n_i+1$ near $\infty_i$. Hurwitz space, with a choice of a specific Abelian differential, called quasi-momentum or primary differential, give rise to a Dubrovin--Frobenius manifold; see section~\cite{B.Dubrovin2, V.Shramchenko} for details. In some examples, the Dubrovin--Frobenius structure of Hurwitz spaces are isomorphic to Dubrovin--Frobenius manifolds associated with orbit spaces of suitable groups. For instance, the orbit space of the finite Coxeter group $A_1$ is isomorphic to the Hurwitz space $H_{0,1}$. Furthermore, orbit space of the extended affine Weyl group $\tilde A_1$ and of the Jacobi group $\Ja(A_1)$ are isomorphic to the Hurwitz spaces $H_{0,0,0}$ and $H_{1,1}$ respectively. Motivated by~these examples, we construct the following diagram
\[
 \begin{tikzcd}
 H_{0,1}\cong \text{orbit space of $A_1$} \arrow{r}{1} \arrow[swap]{d}{2} & H_{0,0,0}\cong \text{orbit space of $\tilde A_1$} \arrow{d}{4} \\%
H_{1,1}\cong \text{orbit space of $\Ja(A_1)$} \arrow{r}{3}& H_{1,0,0}\cong\ ?
 \end{tikzcd}
\]
From the Hurwitz space side, the vertical lines 2 and 4 mean that we increase the genus by $1$, and the horizontal lines mean that we split one pole of order $2$ into two simple poles. From the orbit space side, the vertical line 2 means that we are doing an extension from the finite Coxeter group $A_1$ to the Jacobi group $\Ja(A_1)$; the line horizontal line~1 means that we are extending the \text{orbit space of $A_1$} to the extended affine Weyl group $\tilde A_1$. Therefore, one might ask if the line 3 and 4 would imply an orbit space interpretation of the Hurwitz space $H_{1,0,0}$. The main goal of~this paper is to define a new class of groups such that its orbit space carries the Dubrovin--Frobenius structure of $H_{1,0,0}$. The new group is called extended affine Jacobi group $A_n$, and is denoted by~$\Ja\big(\tilde A_1\big)$. This group is an extension of the Jacobi group $\Ja(A_1)$ and of the extended affine Weyl group $\tilde A_1$.

\subsection{Results}

The main goal of this paper is to construct the Dubrovin--Frobenius structure of the Hurwitz space $H_{1,0,0}$ from the data of the group $\Ja\big(\tilde A_1\big)$. In other words, we derive the WDVV solution associated to the group $\Ja\big(\tilde A_1\big)$ without using the correspondent Hurwitz space construction. First of all, recall the definition of WDVV equation:
 \begin{Definition}
The function $F(t)$, $t=\big(t^1,t^2,\dots,t^n\big)$ is a solution of a WDVV equation if its third derivatives
\begin{gather}\label{freeenergy}
c_{\alpha\beta\gamma}=\frac{\partial^3F}{\partial t^{\alpha}\partial t^{\beta}\partial t^{\gamma}}
\end{gather}
satisfy the following conditions:

\begin{enumerate}\itemsep=0pt
\item[1)] $\eta_{\alpha\beta}=c_{1\alpha\beta}$ is constant nondegenerate matrix;
\item[2)] the function
\begin{gather*}
c_{\alpha\beta}^{\gamma}=\eta^{\gamma\delta}c_{\alpha\beta\delta}
\end{gather*}
is structure constant of associative algebra;
\item[3)] $F(t)$ must be quasi-homogeneous function
\begin{gather*}
F\big(c^{d_1}t^1,\dots,c^{d_n}t^n\big)=c^{d_F}F\big(t^1,\dots,t^n\big)
\end{gather*}
for any nonzero $c$ and for some numbers $d_1, \dots , d_n$, $d_F$.
\end{enumerate}
\end{Definition}
Our goal is to extract a WDVV equation from the data of a suitable group $\Ja\big(\tilde A_1\big)$. We define the group $\Ja\big(\tilde A_1\big)$. Recall that the group $A_1$ acts on $\mathbb{C}\ni v_0$ by reflections
 \begin{gather*}
 v_0\mapsto -v_0.
 \end{gather*}
The group $\Ja\big(\tilde A_1\big)$ is an extension of the group $A_1$ in the following sense:

 \begin{Proposition}
The group $\Ja\big(\tilde A_1\big)\ni (w,t,\gamma)$ acts on $\Omega:=\mathbb{C}\oplus \mathbb{C}^{2}\oplus \mathbb{H} \ni (u,v,\tau)=(u,v_0,v_2,\tau)$ as follows:
\begin{gather*}
w(u,v,\tau)=(u,wv,\tau),
\\
t(u,v,\tau)=\left(u- \langle \lambda,v \rangle_{\tilde A_1}-\frac{1}{2} \langle \lambda,\lambda \rangle_{\tilde A_1}\tau,v+\lambda\tau+\mu,\tau\right),
\\
\gamma(u,v,\tau)=\left(u+\frac{c \langle v,v \rangle_{\tilde A_1}}{2(c\tau+d)},\frac{v}{c\tau+d},\frac{a\tau+b}{c\tau+d}\right),
\end{gather*}
where $w \in A_1$ acts by reflection in the first $v_0$ variables of $ \mathbb{C}^{2} \ni v=(v_0,v_2)$,
\[
t=(\lambda,\mu) \in \mathbb{Z}^{2}, \qquad
\begin{pmatrix} a & b\\c & d\end{pmatrix} \in {\rm SL}_2(\mathbb{Z}),\qquad
\langle v,v \rangle_{\tilde A_1}=2v_0^2-2v_2^2.
\]
\end{Proposition}
See Section~\ref{group def} for details.

In order to define any geometric structure in an orbit space, first it is necessary to define a~notion of invariant $\Ja\big(\tilde A_1\big)$ sections. For this purpose, we generalise the ring of invariant functions used in~\cite{BertolaM.1, BertolaM.2} for the group $\Ja( A_1)$, which are called Jacobi forms. This notion was first defined in~\cite{M.EichlerandD.Zagier} by Eichler and Zagier for the group $\Ja(A_1)$, and it was further generalised for the group $\Ja( A_1)$ in~\cite{K.Wirthmuller} by Wirthmuller. Furthermore, an explicit base of generators were derived in~\cite{BertolaM.1, BertolaM.2} by Bertola. The Jacobi forms used in this thesis are defined by:

\begin{Definition}
The weak $\tilde A_1$-invariant Jacobi forms of weight $k$, order $l$, and index $m$ are functions on $\Omega=\mathbb{C}\oplus \mathbb{C}^{n+2}\oplus\mathbb{H}\ni (u,v_0,v_2,\tau)=(u,v,\tau)$ which satisfy
\begin{gather}
\varphi(w(u,v,\tau))=\varphi(u,v,\tau),\qquad \textrm{$A_1$-invariant condition},\nonumber
\\
\varphi(t(u,v,\tau))=\varphi(u,v,\tau),\nonumber
\\
\varphi(\gamma(u,v,\tau))=(c\tau+d)^{-k}\varphi(u,v,\tau),\nonumber
\\
E\varphi(u,v,\tau):=-\frac{1}{2\pi {\rm i}}\frac{\partial}{\partial u}\varphi(u,v,\tau)=m\varphi(u,v,\tau),\qquad \textrm{Euler vector field}.
\label{jacobiform in the introduction}
\end{gather}
\end{Definition}

Moreover, the weak $\tilde A_1$-invariant Jacobi forms are meromorphic in the variable $v_{2}$ on a fixed divisor, in contrast with the Jacobi forms of the group $\Ja(A_1)$ ,which are holomorphic in each variable; see details on Section~\ref{Jacobi forms Jtildea1}. The ring of weak $\tilde A_1$-invariant Jacobi forms gives the notion of the Euler vector field; indeed, the vector field defined in the last equation of~\eqref{jacobiform in the introduction} measures the degree of the Jacobi forms, which coincides with the index. The differential geometry of the orbit space of the group $\Ja\big(\tilde A_1\big)$ should be understood as the space such that its sections are written in terms of Jacobi forms. Then, in order for this statement to make sense, we must prove a Chevalley type theorem, which is:

\begin{Theorem}
The trigraded algebra of Jacobi forms $J_{\bullet,\bullet,\bullet}^{\Ja(\tilde A_{1})}=\bigoplus _{k,l,m}J_{k,l,m}^{\tilde A_{1}}$ is freely generated by $2$ fundamental Jacobi forms $(\varphi_0,\varphi_1)$ over the graded ring $E_{\bullet,\bullet}$
\begin{gather*}
J_{\bullet,\bullet,\bullet}^{\Ja(\tilde A_1)}=E_{\bullet,\bullet}\left[\varphi_0,\varphi_1\right],
\end{gather*}
where
\begin{gather*}
E_{\bullet,\bullet}=J_{\bullet,\bullet,0}
\end{gather*}
is the ring of coefficients.

More specifically, the ring of function $E_{\bullet,\bullet}$ is the space of functions $f(v_{2},\tau)$ such that, for fixed $\tau$, the functions $\tau\mapsto f(v_{2},\tau)$ is an elliptic function.
\end{Theorem}
Moreover, $(\varphi_0,\varphi_1)$ are given by

\begin{Corollary}
The function
\begin{gather*}
\left[{\rm e}^{z\frac{\partial}{\partial p}} \left({\rm e}^{2\pi {\rm i}u} \frac{\theta_1(v_0+v_2+p)\theta_1(-v_0+v_2+p)}{\theta_1(2v_2+p)\theta_1^{\prime}(0)}\right) \right]
\bigg|_{p=0}=\varphi_1^{\Ja(\tilde A_1)}+\varphi_0^{\Ja(\tilde A_1)}z+O\big(z^2\big),
\end{gather*}
generates the Jacobi forms $\varphi_0^{\Ja(\tilde A_1)}$ and $\varphi_1^{\Ja(A_1)}$,
where
\begin{gather*}
\varphi_0^{\Ja(\tilde A_1)}:=\frac{\partial }{\partial p}\big(\hat\varphi_1^{\Ja(\tilde A_1)} \big)\bigg|_{p=0}.
\end{gather*}
\end{Corollary}

This lemma realises the functions $(\varphi_0,\varphi_1,v_{2},\tau)$ as coordinates of the orbit space of $\Ja\big(\tilde A_1\big)$.
The unit vector field is chosen to be
\begin{gather}\label{unit vector field in the introduction}
e=\frac{\partial}{\partial \varphi_0},
\end{gather}
because $\varphi_0$ is the basic generator with maximum weight degree; see the Section~\ref{Jacobi forms Jtildea1} for details.

The last component we need to construct is the intersection form of the orbit space of $\Ja\big(\tilde A_1\big)$. The natural candidate to be such a~metric is the invariant metric of the group $\Ja\big(\tilde A_1\big)$, which given~by
\begin{gather}\label{intersection form in the introduction}
g=2{\rm d}v_0^2-2{\rm d}v_2^2+2{\rm d}u{\rm d}\tau.
\end{gather}
From the data of the intersection form~\eqref{intersection form in the introduction}, is possible to derive a second flat metric of the orbit space $\Ja\big(\tilde A_1\big)$.
 The second metric is given by
\begin{gather*}
\eta^{*}:={\rm Lie}_{e}g^{*},
\end{gather*}
and it is denoted by the Saito metric due to K.~Saito, who was the first to define this metric for the case of finite Coxeter group~\cite{Saito}. One of the main technical problems of this paper is to prove that the Saito metric $\eta^{*}$ is flat. At this point, we can state our main result.

\begin{Theorem}
A suitable covering of the orbit space $\left(\mathbb{C}\oplus\mathbb{C}^{2}\oplus\mathbb{H}\right)/\Ja\big(\tilde A_1\big)$ with the intersection form~\eqref{intersection form in the introduction}, unit vector field~\eqref{unit vector field in the introduction}, and Euler vector field given by the last equation of~\eqref{jacobiform in the introduction} has a~Dub\-rovin--Frobenius manifold structure. Moreover, a suitable covering of $\mathbb{C}\oplus\mathbb{C}^{n+1}\oplus\mathbb{H}/\Ja\big(\tilde A_1\big)$ is isomorphic as Dubrovin--Frobenius manifold to a suitable covering of the Hurwitz space $H_{1,0,0}$.
\end{Theorem}
See Section~\ref{Construction of WDVV solution} for details. In particular, we derive explicitly the WDVV solution associated with the orbit space of $\Ja\big(\tilde A_1\big)$, which is given by
\begin{gather*}
F\big(t^1,t^2,t^3,t^4\big)=\frac{\rm i}{4\pi}\big(t^1\big)^2t^4-2t^1t^2t^3
-\big(t^2\big)^2\log\left(t^2\frac{\theta_1^{\prime}\big(0,t^4\big)}{\theta_1\big(2t^3,t^4\big)}\right),
\end{gather*}
where
\begin{gather}
\label{theta def}
\theta_1(v,\tau)=2\sum_{n=0}^{\infty} (-1)^{n}{\rm e}^{\pi {\rm i}\tau(n+\frac{1}{2})^2 }\sin((2n+1)v).
\end{gather}

The results of this paper are important because of the following:
 \begin{enumerate}\itemsep=0pt
\item The Hurwitz spaces $H_{1,0,0}$ are classified by the group $\Ja\big(\tilde A_1\big)$, hence we increase the know\-ledge of the WDVV/discrete group correspondence. Recently, the case $\Ja\big(\tilde A_1\big)$ attracted the attention of experts, due to its application in integrable systems~\cite{M.CutimancoandV.Shramchenko, E.V.FerapontovM.V.PavlovL.Xue, Romano1}.

\item The orbit space construction of the group $\Ja\big(\tilde A_1\big)$ can be generalised to the group $\Ja(\tilde A_n)$; see the definition in~\cite{GAlmeida}. Further, the same can be done to the other classical finite Coxeter groups as $B_n$, $D_n$. Hence, these orbit spaces could give rise to a new class of Dubrovin--Frobenius manifolds. Furthermore, the associated integrable hierarchies of this new class of Dubrovin--Frobenius manifolds could have applications in Gromow--Witten theory and combinatorics.
\end{enumerate}

This paper is organised in the following way: In Section~\ref{invariant theory}, we define extended affine Jacobi group $\Ja\big(\tilde A_1\big)$ and we prove some results related with its ring of invariant functions. In Section~\ref{Frobenius structure A1}, we construct a~Dubrovin--Frobenius structure on the orbit spaces of $\Ja\big(\tilde A_1\big)$ and compute its free-energy. Furthermore, we show that the orbit space of the group $\Ja\big(\tilde A_1\big)$ is isomorphic, as a~Dubrovin--Frobenius manifold, to the Hurwitz--Frobenius manifold $\tilde H_{1,0,0}$~\cite{B.Dubrovin2, V.Shramchenko}. See Theorem~\ref{mainresult} for details.

\section[Invariant theory of J(A1)]
{Invariant theory of $\boldsymbol{\Ja\big(\tilde A_1\big)}$}\label{invariant theory}

 The focus of this section is to define a new extension of the finite Coxeter group $A_1$ such that it contains the affine Weyl group $\tilde A_1$ and the Jacobi group $\Ja(A_1)$. This new extension will be denoted by Extended affine Jacobi group $\Ja\big(\tilde A_1\big)$. Further, we prove that, from the data of the group $\Ja\big(\tilde A_1\big)$, we can reconstruct the Dubrovin--Frobenius structute of the Hurwitz space $H_{1,0,0}$ on the orbit space of $\Ja\big(\tilde A_1\big)$. The advantage of this orbit space construction is the Chevalley Theorem~\ref{chevalley Jtildea1}, which gives a global interpretation for orbit space of $\Ja\big(\tilde A_1\big)$. Furthermore, it~attaches the group $\Ja\big(\tilde A_1\big)$ to the Hurwitz space $H_{1,0,0}$, and this fact might be useful in the general understanding of WDVV/group correspondence. These results sre interesting because the Hurwitz space $H_{1,0,0}$ is well know to have a rich Dubrovin--Frobenius structure, called a~tri-Hamiltonian structure~\cite{Romano1} and~\cite{Pavlov2}. This fact realises the orbit space of $\Ja\big(\tilde A_1\big)$ as suitable ambient space for Dubrovin--Frobenius submanifolds. Furthermore, it shows an interesting relationship relation between the integrable systems of the ambient space and the integrable systems of its Dubrovin--Frobenius submanifolds.

\subsection[The group J(A1)]
{The group $\boldsymbol{\Ja\big(\tilde A_1\big)}$}\label{group def}

The main goal of this section is to motivate and to define the group $\Ja\big(\tilde A_1\big)$. In order to do that, it will be necessary to recall the definition of the group $A_1$, and some of its extensions. Moreover, its goal is to understand how to derive WDDV solution starting from these groups.

The group $A_n$ acts on the space $\Omega^{A_n}=\big\{(v_0,v_1,\dots,v_n)\in \mathbb{C}^{n+1}\colon\sum _{i=0}^nv_i=0 \big\}$ by permutations:
\begin{gather}\label{An action permutation}
(v_0,v_1,\dots,v_n)\mapsto(v_{i_0},v_{i_1},\dots,v_{i_n}).
\end{gather}
Let us concentrate on the simplest possible case, i.e., $n=1$. In this case, the action on $\mathbb{C}\cong\Omega^{A_1}$ is just:
\begin{gather*}
v_0\mapsto -v_0.
\end{gather*}
The understanding of the orbit space of $A_1$ requires a Chevalley theorem for the ring of invariants. The Chevalley theorem form the group $A_n$ says that

\begin{Theorem}[\cite{N.Bourbaki}]
Let the Coxeter group $A_n$ which acts on $\Omega^{A_n}\ni (v_0,v_1,\dots,v_n)$ as~\eqref{An action permutation}, then
\begin{gather*}
\mathbb{C}[v_0,v_1,\dots,v_n]^{A_n}\cong \mathbb{C}[a_2,a_3,\dots,a_{n+1}],
\end{gather*}
where $a_i$ are weighted homogeneous polynomials of degree $i$.
\end{Theorem}

 In the $A_1$ case, the ring of invariants is just
\begin{gather*}
\mathbb{C}\big[v_0^2\big]\cong\mathbb{C}[a_2] ,
\end{gather*}
 then the orbit space of $A_1$ is just the
 \begin{gather*}
 \mathop{\rm Spec} \left(\mathbb{C}\big[v_0^2\big]\right).
 \end{gather*}
 In the papers~\cite{B.Dubrovin2,B.Dubrovin1}, it was demonstrated that $\mathbb{C}/A_1$ has structure of Dubrovin--Frobenius manifold. Furthermore, it is isomorphic to the Hurwitz space $H_{0,1}$, i.e., the space of rational functions with a double pole. The isomorphism can be realized by the following map:
\begin{gather*}
[v_0]\mapsto \lambda^{A_1}(p,v_0)=(p-v_0)(p+v_0)=p^2+a_2.
\end{gather*}
Note that the isomorphism works, because $\lambda^{A_1}(p,v_0)$ is invariant under the $A_1$-action. Applying the methods developed in~\cite{B.Dubrovin2,B.Dubrovin1}, one can show that the WDVV solution associated with this orbit space is
\begin{gather*}
F\big(t^1\big)=\frac{\big(t^1\big)^3}{6},
\end{gather*}
where $t^1$ is the flat coordinate of the metric~$\eta$.

In~\cite{B.Dubrovin2, B.DubrovinandY.Zhang} it was also considered the extended affine $A_1$ that is denoted by $\tilde A_1$. The action on
\begin{gather*}
\left(L^{A_1}\otimes\mathbb{C}\right)\oplus \mathbb{C}
=\left\{(v_0,v_1,v_2)\in \mathbb{C}^3\colon \sum_{i=0}^1 v_i=0 \right\}
\end{gather*}
is{\samepage
\begin{gather*}
v_0\mapsto \pm v_0+\mu_0,\qquad
v_2\mapsto v_2+\mu_2,
\end{gather*}
where $\mu_0, \mu_2 \in \mathbb{Z}$.}

A notion of the invariant ring for the group extended affine $A_n$ was defined in~\cite{B.DubrovinandY.Zhang}, and Dubrovin and Zhang proved that this invariant ring for the case $\tilde A_1$ is isomorphic to
\begin{gather*}
\mathbb{C}\big[{\rm e}^{2\pi {\rm i} v_2}\cos(2\pi {\rm i} v_0),{\rm e}^{2\pi {\rm i} v_2}\big].
\end{gather*}
 Therefore, the orbit space of $\tilde A_1$ is the weight projective variety associated with
\begin{gather*}
 \mathop{\rm Spec}\left(\mathbb{C}\left[{\rm e}^{2\pi {\rm i} v_2}\cos(2\pi {\rm i} v_0),{\rm e}^{2\pi {\rm i} v_2}\right]\right).
\end{gather*}
 Further, a Dubrovin--Frobenius manifold structure was built on the orbit space of $\tilde A_1$ with the following WDVV solution:
\begin{gather}\label{WDVV Dub-Zhang}
F\big(t^1,t^2\big)=\frac{\big(t^1\big)^2t^2}{2}+{\rm e}^{t^2}.
\end{gather}
The orbit space of $\tilde A_1$ is also associated with a Hurwitz space, but the relation is slightly less straightforward. The first step is to consider the following map:
\begin{gather*}
[v_0,v_2]\mapsto \lambda^{\tilde A_1}(p,v_0,v_2)={\rm e}^p+{\rm e}^{2\pi {\rm i} v_2}\cos(2\pi {\rm i} v_0)+{\rm e}^{2\pi {\rm i} v_2}{\rm e}^{-p}.
\end{gather*}
The second is to consider the Legendre transformation of $S_2$ type \cite[Appendix~B and Chapter~5]{B.Dubrovin2}. Consider
\begin{gather*}
b={\rm e}^{2\pi {\rm i} v_2}\cos(2\pi {\rm i} v_0), \qquad a={\rm e}^{2\pi {\rm i}v_2},
\end{gather*}
and the following choice of primary differential ${\rm d}\tilde p$ implicity given by
\begin{gather*}
{\rm d}p=\frac{{\rm d}\tilde p}{\tilde p-b}.
\end{gather*}
Then, in these new coordinates $\lambda^{\tilde A_1}$, is given by
\begin{gather*}
\lambda(\tilde p, a,b)=\tilde p+\frac{a}{\tilde p-b}.
\end{gather*}
Hence, the orbit space of $\tilde A_1$ is isomorphic to the Hurwitz space $H_{0,0,0}$, i.e., space of fractional functions with two simple poles.

The next example of group to be considered is the Jacobi group $\Ja(A_1)$, which acts on
\begin{gather*}
\Omega^{\Ja(A_1)}:=\left(L^{A_1}\otimes\mathbb{C}\right)\oplus \mathbb{C}\oplus\mathbb{H}=\left\{(v_0,v_1,u,\tau)\in \mathbb{C}^3\oplus\mathbb{H}\colon\sum_{i=0}^1 v_i\in\mathbb{Z}+\tau\mathbb{Z} \right\}
\end{gather*}
 as follows:

$A_1$-action:
\begin{gather}
\label{jacobi A1 b}
v_0\mapsto -v_0,\qquad
u\mapsto u, \qquad
\tau\mapsto\tau.
\end{gather}

Translation:
\begin{gather}
\label{jacobi trans b}
v_0\mapsto v_0+\mu_0+\lambda_0\tau,\qquad
u\mapsto u-\lambda_0v_0-\frac{\lambda_0^2}{2}\tau,\qquad
\tau\mapsto\tau,
\end{gather}
where $\mu_0,\lambda_0\in \mathbb{Z}$.

${\rm SL}_2(\mathbb{Z})$-action:
\begin{gather}
\label{jacobi sl2z b}
v_0\mapsto \frac{v_0}{c\tau+d},\qquad
u\mapsto u-\frac{cv_0^2}{2(c\tau+d)},\qquad
\tau\mapsto\frac{a\tau+b}{c\tau+d},
\end{gather}
where $a,b,c,d\in \mathbb{Z}$, and $ad-bc=1$.

The notion of invariant ring of $\Ja(A_1)$ was first defined in~\cite{M.EichlerandD.Zagier}. However, the definitions stated in~\cite{BertolaM.1, BertolaM.2, K.Wirthmuller} are more suitable for this purpose.
\begin{Definition}
\label{jacobi forms bertola}
The weak $ A_1$-invariant, Jacobi forms of weight $k$, and index $m$ are holomorphic functions on $\Omega=\mathbb{C}\oplus\mathbb{C}\oplus\mathbb{H}\ni (u,v_0,\tau)$ which satisfy
\begin{gather*}
\varphi(u,-v_0,\tau)=\varphi\left(u,v_0,\tau\right),\qquad \textrm{$A_1$-invariant condition},
\\
\varphi\left(u-\lambda_0v_0-\frac{{\lambda_0}^2}{2}\tau,v_0+\lambda_0\tau+\mu,\tau\right)
=\varphi\left(u,v_0,\tau\right),
\\
\varphi\left(u+\frac{cv_0^2}{2(c\tau+d)},\frac{v_0}{c\tau+d},\frac{a\tau+b}{c\tau+d}\right)
=\left(c\tau+d\right)^{k}\varphi(u,v_0,\tau),
\\
E\varphi(u,v_0,\tau):=\frac{1}{2\pi {\rm i}}\frac{\partial}{\partial u}\varphi(u,v_0,\tau)=m\varphi(u,v_0,\tau).
\end{gather*}
Moreover,
$\varphi$ are locally bounded functions of $v_0$ as $\Im(\tau)\mapsto +\infty$ (weak condition).

The space of $\tilde A_1$-invariant Jacobi forms of weight $k$, and index $m$ is denoted by $J_{k,m}^{A_1}$, and $J_{\bullet,\bullet}^{\Ja( A_1)}=\bigoplus _{k,m}J_{k,m}^{ A_1}$ is the space of Jacobi forms $ A_1$-invariant.
\end{Definition}
In~\cite{M.EichlerandD.Zagier}, it was proved the following a version of the Chevalley theorem.
\begin{Theorem}
Let $J_{\bullet,\bullet}^{\Ja( A_1)}$ the ring of Jacobi forms $ A_1$-invariant, then
\begin{gather*}
J_{\bullet,\bullet}^{\Ja( A_1)}\cong M_{\bullet}[\varphi_0,\varphi_2],
\end{gather*}
where $M_{\bullet}$ is the ring of holomorphic modular forms, and
\begin{gather*}
\varphi_2={\rm e}^{2\pi {\rm i} u}{\left(\frac{\theta_1(v_0,\tau)}{\theta_1^{\prime}(0,\tau)}\right)}^2,\qquad
\varphi_0=\varphi_2\wp(v_0,\tau),
\end{gather*}
 $\theta_1$ is the Jacobi $\theta_1$-function~\eqref{theta def}, and $\wp$ is the Weierstrass $\rm P$-function, which is defined as
\begin{gather}\label{Weiestrass P 0}
\wp(v,\tau)=\frac{1}{v^2}+\sum_{m^2+n^2\neq 0}^{\infty}\frac{1}{(v-m-n\tau)^2}-\frac{1}{(m+n\tau)^2}.
\end{gather}
\end{Theorem}
Note that this Chevalley theorem is slightly different from the others. The ring of the coefficients is the ring of holomorphic of modular forms, instead of just $\mathbb{C}$. The geometric interpretation of this fact is that the orbit space of $\Ja(A_1)$ is a~line bundle, such that its base is family of elliptic curves $E_{\tau}$ quotient by the group $A_1$ parametrised by $\mathbb{H}/{\rm SL}_2(\mathbb{Z})$. In~\cite{BertolaM.1} and~\cite{BertolaM.2}, it was proved that orbit space of $\Ja(A_1)$ has a Dubrovin--Frobenius structure. Furthermore, the orbit space of $\Ja(A_1)$ is isomorphic to $H_{1,1}$, i.e., space of elliptic functions with one double pole. The explicit isomorphism is given by the map
 \begin{gather}\label{superpotential bertola}
[(u,v_0,\tau)]\mapsto \lambda^{\Ja(A_1)}(v,u,v_0,\tau)={\rm e}^{2\pi {\rm i} u}\frac{\theta_1(v-v_0,\tau)\theta_1(v+v_0,\tau)}{\theta_1^2(v,\tau)}.
\end{gather}
As in the $A_1$ case, the isomorphism is only possible, because the map~\eqref{superpotential bertola} is invariant under~\eqref{jacobi A1 b}--\eqref{jacobi sl2z b}. A WDVV solution for this case is the following:
 \begin{gather}\label{WDVV bertola}
F\big(t^1,t^2,\tau\big)=\frac{\big(t^1\big)^2\tau}{2}+\frac{t^1\big(t^2\big)^2}{2}-\frac{\pi {\rm i}\big(t^2\big)^2}{48}E_2(\tau),
\end{gather}
where
 \begin{gather*}
E_2(\tau)=1+\frac{3}{\pi^2}\sum_{m\neq 0}\sum_{n=-\infty}^{\infty}\frac{1}{(m+n\tau)^2}.
\end{gather*}
A remarkable fact in these orbit space constructions is its correspondence with Hurwitz spaces, which can be summarized by the following diagram:
\[
 \begin{tikzcd}
 H_{0,1}\cong\mathbb{C}/A_1 \arrow{r}{1} \arrow[swap]{d}{2} & H_{0,0,0}\cong\mathbb{C}^2/\tilde A_1 \arrow{d}{4} \\
H_{1,1}\cong\left(\mathbb{C}\oplus\mathbb{C}\oplus\mathbb{H}\right)/\Ja(A_1) \arrow{r}{3}& H_{1,0,0}\cong\ ?
 \end{tikzcd}
\]
The arrows of the diagram above have a double meaning. The first one is simply an extension of~the group, the arrow $2$ is ``Jacobi" extension, and the arrow $1$ is ``affine" extension. The~second meaning is related to the Hurwitz space side: the arrows $2$ and $4$ increase by one the genus, and the arrows $1$ and $3$ split a double pole in 2 simple poles. The missing part of the diagram is exactly the orbit space counter part of $H_{1,0,0}$. The diagram suggest that the new group should be an extension of the $A_1$ group, such that combine the groups $\tilde A_1$, and $\Ja(A_1)$, furthermore, it~should preserve $H_{1,0,0}$ in a similar way for what was done in~\eqref{superpotential bertola}. To construct the desired group, we start from the group $\Ja(A_1)$ and make an extension in order to incorporate the $\tilde A_1$ group. Concretely, we extend the domain $\Omega^{\Ja(A_1)}$ to
\begin{gather*}
\Omega^{\Ja(\tilde A_1)}:=\Omega^{A_1}\oplus\mathbb{C}\oplus\mathbb{C}\oplus\mathbb{H}=\big\{(v_0,v_1,v_2,u,\tau) \in \mathbb{C}^4\oplus\mathbb{H}\colon v_0+v_1\in \mathbb{Z}\oplus\tau\mathbb{Z} \big\},
\end{gather*}
and we extend the group action $\Ja(A_1)$ to the following action:

$A_1$-action:
\begin{gather}\label{mon 1}
v_0\mapsto -v_0,\qquad
v_2\mapsto v_2,\qquad
u\mapsto u,\qquad
\tau\mapsto\tau.
\end{gather}

Translation:
\begin{gather}
v_0\mapsto v_0+\mu_0+\lambda_0\tau,\qquad
v_2\mapsto v_2+\mu_2+\lambda_2\tau,\nonumber
\\
u\mapsto u-2\lambda_0v_0+2\lambda_2v_2-\lambda_0^2\tau+\lambda_2^2\tau+k,\qquad
\tau\mapsto\tau,\label{mon 2}
\end{gather}
where $(\lambda_0,\lambda_2), (\mu_0,\mu_2) \in \mathbb{Z}^2$, and $k \in \mathbb{Z}$.

${\rm SL}_2(\mathbb{Z})$-action:
\begin{gather}\label{mon 3}
v_0\mapsto \frac{v_0}{c\tau+d},\qquad
v_2\mapsto \frac{v_2}{c\tau+d},\qquad
u\mapsto u+\frac{c\big(v_0^2-v_2^2\big)}{(c\tau+d)},\qquad
\tau\mapsto\frac{a\tau+b}{c\tau+d},
\end{gather}
where $a,b,c,d\in \mathbb{Z}$, and $ad-bc=1$.

The group action~\eqref{mon 1}, \eqref{mon 2}, and \eqref{mon 3} is called extended affine Jacobi group $A_1$, and is denoted by $\Ja\big(\tilde A_1\big)$.

\begin{Remark}\sloppy
 The translations of the group $\tilde A_1$ are a subgroup of the translations of the group~$\Ja\big(\tilde A_1\big)$. Therefore, it is in that sense that $\Ja\big(\tilde A_1\big)$ is a combination of $\tilde A_1$ and $\Ja(A_1)$.
\end{Remark}

In order to rewrite the action of $\Ja\big(\tilde A_1\big)$ in an intrinsic way, consider the $A_1$ in the following extended space
\begin{gather*}
L^{\tilde A_1}=\left\{(z_0,z_1,z_2)\in \mathbb{Z}^{3}\colon \sum _{i=0}^3z_i=0 \right\}.
\end{gather*}
The action of $A_1$ on $L^{\tilde A_1}$ is given by
\begin{gather*}
w(z_0,z_1,z_2)=(z_{1},z_{0},z_2)
\end{gather*}
permutations in the first two variables. Moreover, $A_1$ also acts on the complexfication of $L^{\tilde A_1}\otimes\mathbb{C}$. Let us use the following identification $\mathbb{Z}^{2}\cong L^{\tilde A_1}$, $\mathbb{C}^{2}\cong L^{\tilde A_1}\otimes \mathbb{C}$, which is possible due to the maps
\begin{gather*}
(v_{0},v_2)\mapsto (v_{0},-v_0,v_2),\qquad
(v_{0},v_1,v_2)\mapsto (v_{0},v_2).
\end{gather*}
The action of $A_1$ on $\mathbb{C}^2 \ni v=(v_0,v_2)$ is:
\begin{gather*}
w(v)=w(v_0,v_2)=(-v_{0},v_2).
\end{gather*}
Let the quadratic form $\langle\,,\,\rangle_{\tilde A_1}$ be given by
\begin{gather}\label{invariant metric A1 t}
\langle v,v\rangle_{\tilde A_1}=v^{\rm T}M_{\tilde A_1}v
=v^{\rm T} \begin{pmatrix}2 & 0 \\0 & -2\\ \end{pmatrix}v=2v_0^2-2v_2^2.
\end{gather}
Consider the following group $L^{\tilde A_1}\times L^{\tilde A_1}\times \mathbb{Z}$ with the following group operation
\begin{gather*}
\forall (\lambda,\mu,k), \big(\tilde\lambda,\tilde\mu,\tilde k\big) \in L^{\tilde A_1}\times L^{\tilde A_1}\times \mathbb{Z},
\\
(\lambda,\mu,k)\bullet\big(\tilde\lambda,\tilde\mu,\tilde k\big)=\big(\lambda+\tilde\lambda,\mu+\tilde\mu,k+\tilde k+\langle \lambda,\tilde\lambda\rangle_{\tilde A_1}\big).
\end{gather*}
Note that $\langle ,\rangle_{\tilde A_1}$ is invariant under $A_1$ group, then $A_1$ acts on $L^{\tilde A_1}\times L^{\tilde A_1}\times \mathbb{Z}$. Hence, we can take the semidirect product $A_1\ltimes \big(L^{\tilde A_1}\times L^{\tilde A_1}\times \mathbb{Z}\big) $ given by the following product
\begin{gather*}
\forall (w,\lambda,\mu,k), \big(\tilde w,\tilde\lambda,\tilde\mu,\tilde k\big) \in A_1\times L^{\tilde A_1}\times L^{\tilde A_1}\times \mathbb{Z},
\\
(w,\lambda,\mu,k)\bullet\big(\tilde w,\tilde\lambda,\tilde\mu,\tilde k\big)=\big(w\tilde w,w\lambda+\tilde\lambda,w\mu+\tilde\mu,k+\tilde k+\langle\lambda,\tilde\lambda\rangle_{\tilde A_1}\big).
\end{gather*}
Denoting $W(\tilde A_1):=A_1\ltimes \big(L^{\tilde A_1}\times L^{\tilde A_1}\times \mathbb{Z}\big)$, we can define
\begin{Definition}
The Jacobi group $\Ja\big(\tilde A_1\big)$ is defined as a semidirect product $W\big(\tilde A_1\big)\rtimes {\rm SL}_2(\mathbb{Z})$. The group action of ${\rm SL}_2(\mathbb{Z})$ on $W(\tilde A_1)$ is defined as
\begin{gather*}
{\rm Ad}_{\gamma}(w)=w,
\\
{\rm Ad}_{\gamma}(\lambda,\mu,k)=\left(a\mu-b\lambda,-c\mu+d\lambda,k+\frac{ac}{2}\langle \mu,\mu \rangle_{\tilde A_1}-bc\langle \mu,\lambda\rangle_{\tilde A_1}+\frac{bd}{2}\langle\lambda,\lambda\rangle_{\tilde A_1}\right)
\end{gather*}
for $(w,t=(\lambda,\mu,k))\in W\big(\tilde A_1\big)$, $\gamma \in {\rm SL}_2(\mathbb{Z})$. Then the multiplication rule is given as follows
\begin{gather*}
(w,t,\gamma)\bullet\big(\tilde w,\tilde t,\tilde \gamma\big)=\big(w\tilde w,t{\rm Ad}_{\gamma}(w\tilde t),\gamma\tilde\gamma\big).
\end{gather*}
\end{Definition}

Then, the action of Jacobi group $\Ja\big(\tilde A_1\big)$ on $\Omega^{\Ja(\tilde A_1)}:=\mathbb{C}\oplus \mathbb{C}^{2}\oplus \mathbb{H} \in (u,v,\tau)$ is described by the main three generators
\begin{gather*}
\hat w=\big(w,0,I_{{\rm SL}_2(\mathbb{Z})}\big),\qquad
 t=\big(I_{A_1},\lambda,\mu,k,I_{{\rm SL}_2(\mathbb{Z})}\big),\qquad
\gamma=\left(I_{A_1},0,\begin{pmatrix}
a & b \\c & d\\\end{pmatrix}\right),
\end{gather*}
 which acts on $\Omega^{\Ja(\tilde A_1)}$ as follows
\begin{gather*}
\hat w(u,v=(v_0,v_2),\tau)= (u,-v_0,v_2,\tau ),
\\[-.5ex]
t(u,v=(v_0,v_2),\tau)=\left(u-\langle \lambda,v \rangle_{\tilde A_1}-\frac{1}{2}\langle \lambda,\lambda \rangle_{\tilde A_1}\tau+k,v_0+\lambda_0\tau+\mu_0,v_2+\lambda_2\tau+\mu_2,\tau\right),
\\[-.5ex]
\gamma(u,v=(v_0,v_2),\tau)=\left(u+\frac{c\langle v,v \rangle_{\tilde A_1}}{2(c\tau+d)},\frac{v_0}{c\tau+d},\frac{v_2}{c\tau+d},\frac{a\tau+b}{c\tau+d}\right),
\end{gather*}
where $\lambda,\mu, k \in L^{\tilde A_1}\times L^{\tilde A_1}\times \mathbb{Z}$,
\begin{gather*}
\lambda=(\lambda_0,\lambda_2), \qquad \mu=(\mu_0,\mu_2).
\end{gather*}
In a more condensed form we have the following proposition.
\begin{Proposition}
The group $\Ja\big(\tilde A_1\big)\ni (\hat w,t,\gamma)$ acts on $\Omega:=\mathbb{C}\oplus \mathbb{C}^{2}\oplus \mathbb{H} \ni (u,v,\tau)$ as follows:
\begin{gather}
\hat w(u,v,\tau)=(u,w(v),\tau),\nonumber
\\[-.5ex]
t(u,v,\tau)=\left(u-\langle \lambda,v \rangle_{\tilde A_1}-\frac{1}{2}\langle \lambda,\lambda \rangle_{\tilde A_1}\tau+k,v+\lambda\tau+\mu,\tau\right),\nonumber
\\[-.5ex]
\gamma(u,v,\tau)=\left(u+\frac{c\langle v,v \rangle_{\tilde A_1}}{2(c\tau+d)},\frac{v}{c\tau+d},\frac{a\tau+b}{c\tau+d}\right).
\label{jacobigroupA1tilde}
\end{gather}
\end{Proposition}

Substituting~\eqref{invariant metric A1 t} in~\eqref{jacobigroupA1tilde}, we get the transformation law~\eqref{mon 1}, \eqref{mon 2}, and~\eqref{mon 3}. The explanation of why~\eqref{jacobigroupA1tilde} is that a~group action for $\Ja\big(\tilde A_1\big)$ is just straightforward computations, but it is a~bit long, so this part of the proof will be omitted.

\subsection[Jacobi forms of J(A1)]{Jacobi forms of $\boldsymbol{\Ja\big(\tilde A_1\big)}$}\label{Jacobi forms Jtildea1}
In order to understand the differential geometry of orbit space, first we need to study the algebra of the invariant functions. Informally, every time that there is a group $W$ acting on a vector space $V$, one could think of the orbit spaces $V/W$ as $V$, but you should remember yourself one can only use the $W$-invariant sections of V. Hence, motivated by the definition of Jacobi forms of group~$A_n$ defined in~\cite{K.Wirthmuller}, and used in the context of Dubrovin--Frobenius manifold in~\cite{BertolaM.1,BertolaM.2}, we present the following:

\begin{Definition}\label{meromorphic jacobi forms}
The weak $\tilde A_1$-invariant Jacobi forms of weight $k$, order $l$, and index $m$ are functions on $\Omega=\mathbb{C}\oplus\mathbb{C}^2\oplus\mathbb{H}\ni (u,v_0,v_2,\tau)=(u,v,\tau)$ which satisfy
\begin{gather}
\varphi(w(u,v,\tau))=\varphi(u,v,\tau),\quad \textrm{$A_1$-invariant condition},\nonumber
\\[-.5ex]
\varphi(t(u,v,\tau))=\varphi(u,v,\tau),\nonumber
\\[-.5ex]
\varphi(\gamma(u,v,\tau))=(c\tau+d)^{-k}\varphi(u,v,\tau),\nonumber
\\[-.5ex]
E\varphi(u,v,\tau):=-\frac{1}{2\pi {\rm i}}\frac{\partial}{\partial u}\varphi(u,v_0,v_2,\tau)=m\varphi(u,v_0,v_2,\tau).
\label{jacobiform}
\end{gather}
Moreover,
\begin{enumerate}\itemsep=0pt
\item[1)] $\varphi$ is locally bounded functions of $v_0$ as $\Im(\tau)\mapsto +\infty$ (weak condition),
\item[2)] for fixed $u$, $v_0$, $\tau$ the function $v_{2}\mapsto \varphi(u,v_0,v_2,\tau)$ is meromorphic with poles of order at~most~$l+2m$ at in $v_{2}=0,\frac{1}{2},\frac{\tau}{2},\frac{1+\tau}{2} \mod\mathbb{Z}\oplus\tau\mathbb{Z}$,
\item[3)] for fixed $u,v_2\neq 0,\frac{1}{2},\frac{\tau}{2},\frac{1+\tau}{2} \mod \mathbb{Z}\oplus\tau\mathbb{Z}$, $\tau$ the function $v_{0}\mapsto \varphi(u,v_0,v_2,\tau)$ is holo\-morphic,
\item[4)] for fixed $u,v_0,v_2\neq 0,\frac{1}{2},\frac{\tau}{2},\frac{1+\tau}{2}\mod \mathbb{Z}\oplus\tau\mathbb{Z}$. The function $\tau\mapsto \varphi(u,v_0,v_2,\tau)$ is holo\-morphic.
\end{enumerate}
The space of $\tilde A_1$-invariant Jacobi forms of weight $k$, order $l$, and index $m$ are denoted by $J_{k,l,m}^{\tilde A_1}$, and $J_{\bullet,\bullet,\bullet}^{\Ja(\tilde A_1)}=\bigoplus _{k,l,m}J_{k,l,m}^{\tilde A_1}$ is the space of Jacobi forms $\tilde A_1$-invariant.
\end{Definition}

\begin{Remark}\label{Remark Zagier Ja1}
The condition $E\varphi(u,v_0,v_2,\tau)=m\varphi(u,v_0,v_2,\tau)$ implies that $\varphi(u,v_0,v_2,\tau)$ has the following form
\begin{gather*}
\varphi(u,v_0,v_2,\tau)=f(v_0,v_2,\tau){\rm e}^{2\pi {\rm i}m u}
\end{gather*}
and the function $f(v_0,v_2,\tau)$ has the following transformation law{\samepage
\begin{gather*}
f(v_0,v_2,\tau)=f(-v_0,v_2,\tau),
\\[-.5ex]
f(v_0,v_2,\tau)={\rm e}^{-2\pi {\rm i} m\big(\langle \lambda,v \rangle +\frac{\langle\lambda,\lambda \rangle}{2}\tau\big)}f(v_0+m_0+n_0\tau,v_2+m_2+n_2\tau,\tau),
\\[-.5ex]
f(v_0,v_2,\tau)=(c\tau+d)^{-k}{\rm e}^{2\pi {\rm i}m\big(\frac{c\langle v,v \rangle}{(c\tau+d)}\big)}f\left(\frac{v_0}{c\tau+d},\frac{v_2}{c\tau+d},\frac{a\tau+b}{c\tau+d}\right).
\end{gather*}

}

\noindent
The functions $f(v_0,v_2,\tau)$ are more closely related to the definition of Jacobi form of the Eichler--Zagier type~\cite{M.EichlerandD.Zagier}. The coordinate $u$ works as kind of automorphic correction in this functions $f(v_0,v_2,\tau)$. Further, the coordinate $u$ will be crucial to construct an equivariant metric on the orbit space of $\Ja\big(\tilde A_1\big)$; see Section~\ref{Frobenius structure A1}.
\end{Remark}

\begin{Remark}
Note that the Jacobi forms in the Definition~\ref{jacobi forms bertola} are holomorphic, and in the Definition \ref{meromorphic jacobi forms}, the Jacobi forms are meromorphic in the variable~$v_2$.
\end{Remark}

The main result of this section is the following:

 \textit{The ring of $\tilde A_1$-invariant Jacobi forms is polynomial over a suitable ring $E_{\bullet,\bullet}:=J_{\bullet,\bullet,0}^{\Ja(\tilde A_1)}$ on suitable generators $\varphi_0$, $\varphi_1$.} Before stating precisely the theorem, I will define the objects $E_{\bullet,\bullet},\varphi_0,\varphi_1$.

The ring $E_{\bullet,l}:=J_{\bullet,l,0}^{\Ja(\tilde A_1)}$ is the space of meromorphic Jacobi forms of index $0$ with poles of order at most $l$ at $0$, $\frac{1}{2}$, $\frac{\tau}{2}$, $\frac{1+\tau}{2}$ $\mod \mathbb{Z}\oplus\tau\mathbb{Z}$, by definition. The sub-ring $J_{\bullet,0,0}^{\Ja(\tilde A_1)}\subset E_{\bullet,\bullet}$ has a~nice structure, indeed:
\begin{Lemma}\label{lemmamodular Jtildea1}
The sub-ring $J_{\bullet,0,0}^{\Ja(\tilde A_1)}$ is equal to $M_{\bullet}:=\bigoplus M_{k}$, where $M_k$ is the space of modular forms of weight $k$ for the full group ${\rm SL}_2(\mathbb{Z})$.
\end{Lemma}

\begin{proof}
Using the Remark \ref{Remark Zagier Ja1}, we know that functions $\varphi(u,v_0,v_2,\tau)\in J_{\bullet,0,0}^{\Ja(\tilde A_1)}$ can not depend on $u$, so $\varphi(u,v_0,v_2,\tau)=\varphi(v_0,v_2,\tau)$. Moreover, for fixed $v_2$, $\tau$ the functions $v_0\mapsto \varphi(v_0,v_2,\tau) $ are holomorphic elliptic functions. Therefore, by Liouville theorem, these function are constant in $v_0$. Similar argument shows that these function do not depend on $v_2$, because $l+2m=0$, i.e., there is no pole. Then, $\varphi=\varphi(\tau)$ are standard holomorphic modular forms.
\end{proof}

\begin{Lemma}\label{ringcoef Jtildea1}
If $\varphi \in E_{\bullet,\bullet}=J_{\bullet,\bullet,0}^{\Ja(\tilde A_1)}$, then $\varphi$ depends only on the variables $v_2$, $\tau$. Moreover, if~$\varphi\in J_{0,l,0}^{\Ja(\tilde A_1)}$ for fixed $\tau$ the function $v_2\mapsto \varphi(v_2,\tau)$ is an elliptic function with poles of order at~most $l$ on $0,\frac{1}{2}, \frac{\tau}{2}, \frac{1+\tau}{2}\mod \mathbb{Z}\oplus\tau\mathbb{Z}$.
\end{Lemma}

\begin{proof}
The proof is essentially the same of the Lemma~\ref{lemmamodular Jtildea1}, the only difference is that now we have poles at $v_2=0,\frac{1}{2}, \frac{\tau}{2}, \frac{1+\tau}{2}\mod \mathbb{Z}\oplus\tau\mathbb{Z}$. Hence, we have dependence on $v_2$.
\end{proof}

As a consequence of Lemma~\ref{ringcoef Jtildea1}, the function $\varphi \in E_{k,l}=J_{k,l,0}^{\Ja(\tilde A_1)}$ has the following form
\begin{gather*}
\varphi(v_2,\tau)=f(\tau)g(v_2,\tau),
\end{gather*}
 where $f(\tau)$ is holomorphic modular form of weight $k$, and for fixed $\tau$, the function $v_2\mapsto g(v_2,\tau)$ is an elliptic function of order at most $l$ at the poles $0,\frac{1}{2}, \frac{\tau}{2}, \frac{1+\tau}{2}\mod \mathbb{Z}\oplus\tau\mathbb{Z}$.

At this stage, we are able to define $\varphi_0$, $\varphi_1$. Note that a natural way to produce meromorphic Jacobi forms is by using rational functions of holomorphic Jacobi forms. Starting here, we will denote the Jacobi forms related with the Jacobi group $\Ja(A_1)$ with the upper index $\Ja(A_1)$, for instance
\begin{gather*}
\varphi^{\Ja(A_1)},
\end{gather*}
and the Jacobi forms related with the Jacobi group $\Ja\big(\tilde A_1\big)$ with the upper index $\Ja\big(\tilde A_1\big)$
\begin{gather*}
\varphi^{\Ja(\tilde A_1)}.
\end{gather*}

In~\cite{BertolaM.1}, Bertola found a~basis of the generators of the Jacobi form algebra by producing a~holomorphic Jacobi form of type $A_n$ as product of $\theta$-functions
\begin{gather*}
\varphi^{\Ja(A_n)}={\rm e}^{2\pi {\rm i} u}\prod_{i=1}^{n+1} \frac{\theta_1(z_i,\tau)}{\theta_1^{\prime}(0,\tau)}.
\end{gather*}
Afterwards, Bertola defined a recursive operator to produce the remaining basic generators. In~order to recall the details, see~\cite{BertolaM.1}. Our strategy will follow the same logic of Bertola method; we~use theta functions to produce a basic generator and thereafter, we produce a recursive operator to produce the remaining part.

\begin{Lemma}
 Let be $\varphi_3^{\Ja(A_2)}(u_1,z_1,z_2,\tau)$ the holomorphic $A_2$-invariant Jacobi form, which corresponds to the algebra generator of maximal weight degree, in this case degree $3$. More exp\-li\-citly,
\begin{gather*}
\varphi_3^{\Ja(A_2)}={\rm e}^{-2\pi {\rm i} u_1}\frac{\theta_1(z_1,\tau)\theta_1(z_2,\tau)\theta_1(-z_1-z_2,\tau)}{{\theta_1^{\prime}(0,\tau)}^3}.
\end{gather*}
Let be $\varphi_2^{\Ja(A_1)}(u_2,z_3,\tau)$ the holomorphic $A_1$-invariant Jacobi form, which corresponds to the algebra generator of maximal weight degree, in this case degree $2{:}$
\begin{gather*}
\varphi_2^{\Ja(A_1)}={\rm e}^{-2\pi {\rm i} u_2}\frac{{\theta_1(z_3,\tau)}^2}{{\theta_1^{\prime}(0,\tau)}^2}.
\end{gather*}
Then, the function
\begin{gather}\label{varphi1 def Jtildea1}
\varphi_1^{\Ja(\tilde A_1)}=\frac{\varphi_3^{\Ja(A_2)}}{\varphi_2^{\Ja(A_1)}}
\end{gather}
is meromorphic Jacobi form of index $1$, weight $-1$, order $0$.
\end{Lemma}

\begin{proof}
 For our convenience, we change the labels $z_1$, $z_2$, $z_3$ to $v_0+v_2$, $-v_0+v_2$, $2v_2$, respectively. Then~\eqref{varphi1 def Jtildea1} has the following form
\begin{gather}
\label{varphi1 def2 Jtildea1}
\varphi_1^{\Ja(\tilde A_1)}(u,v_0,v_2,\tau)={\rm e}^{-2\pi {\rm i} u}\frac{\theta_1(v_0+v_2,\tau)\theta_1(-v_0+v_2,\tau)}{\theta_1^{\prime}(0,\tau) \theta_1(2v_2,\tau)}.
\end{gather}

Let us prove each item separately.

 $A_1$-\emph{invariant}.
The $A_1$ group acts on~\eqref{varphi1 def2 Jtildea1} by permuting its roots, thus~\eqref{varphi1 def2 Jtildea1} remains invariant under this operation.

 \emph{Translation invariant.}
 Recall that under the translation $v\mapsto v+m+n\tau$, the Jacobi theta function transforms as~\cite{BertolaM.1, E.T.WhittakerandG.N.Watson}:
\begin{gather}\label{transformtheta Jtildea1}
\theta_1(v_i+\mu_i+\lambda_i\tau,\tau)=(-1)^{\lambda_i+\mu_i}{\rm e}^{-2\pi {\rm i}\big(\lambda_iv_i+\frac{\lambda_i^2}{2}\tau\big)}\theta_1(v_i,\tau).
\end{gather}
Then, substituting the transformation~\eqref{transformtheta Jtildea1} into~\eqref{varphi1 def2 Jtildea1}, we conclude that~\eqref{varphi1 def2 Jtildea1} remains inva\-riant.

 ${\rm SL}_2(\mathbb{Z})$-\emph{invariant}.
Under ${\rm SL}_2(\mathbb{Z})$-action the following function transform as
\begin{gather}
\label{transformtheta2 Jtildea1}
\frac{\theta_1\big(\frac{v_i}{c\tau+d},\frac{a\tau+d}{c\tau+d}\big)}
{\theta_1^{\prime}\big(0,\frac{a\tau+d}{c\tau+d}\big)}
=(c\tau+d)^{-1}\exp\left(\frac{\pi {\rm i} cv_i^2 }{c\tau+d}\right)\frac{\theta_1(v_i,\tau)}{\theta_1^{\prime}(0,\tau)}.
\end{gather}
Then, substituting~\eqref{transformtheta2 Jtildea1} in~\eqref{varphi1 def2 Jtildea1}, we get
\begin{gather*}
\varphi_1^{\Ja(\tilde A_1)}\mapsto \frac{\varphi_1^{\Ja(\tilde A_1)}}{c\tau+d}.
\end{gather*}

 \emph{Index} $1$.
\begin{gather*}
\frac{1}{2\pi {\rm i}}\frac{\partial}{\partial u}\varphi_1{\Ja\big(\tilde A_1\big)}=\varphi_1^{\Ja(\tilde A_1)}.
\end{gather*}

 \emph{Analytic behavior.}
Note that $\varphi_1^{\Ja(\tilde A_1)}\theta_1^2(2v_2,\tau)$ is holomorphic function in all the variables~$v_i$. Therefore, $\varphi_1^{\Ja(\tilde A_1)}$ are holomorphic functions on the variables $v_0$, and meromorphic function in~the variable $v_{2}$ with poles on $\frac{j}{2}+\frac{l\tau}{2}$, $j,l=0,1$ of order~2, i.e.,~$l=0$, since $m=1$.
\end{proof}

In order to define the desired recursive operator, it is necessary to enlarge the domain of~the Jacobi forms from $\mathbb{C}\oplus\mathbb{C}^2\oplus\mathbb{H}\ni (u,v_0,v_2,\tau)$ to $\mathbb{C}\oplus\mathbb{C}^3\oplus\mathbb{H}\ni (u,v_0,v_2,p,\tau)$. In addition, we~define a lift of Jacobi forms defined in $ \mathbb{C}\oplus\mathbb{C}^2\oplus\mathbb{H}$ to $\mathbb{C}\oplus\mathbb{C}^3\oplus\mathbb{H}$ as
\begin{gather*}
\varphi(u,v_0+v_2,-v_0+v_2,\tau)\mapsto \hat\varphi(p):=\varphi(u,v_0+v_2+p,-v_0+v_2+p,\tau).
\end{gather*}
A convenient way to do computations in these extended Jacobi forms is to use the following coordinates
\begin{gather*}
s=u+g_1(\tau)p^2,\quad\
z_1=v_0+v_2+p,\quad\
z_2=-v_0+v_2+p,\quad\
z_3=2v_2+p,\quad\
\tau=\tau.
\end{gather*}
The bilinear form $\langle v,v \rangle_{\tilde A_1}$ is extended to
\begin{gather*}
\langle (z_1,z_2,z_3),(z_1,z_2,z_3) \rangle_E=z_1^2+z_2^2-z_3^2,
\end{gather*}
or equivalently,
\begin{gather*}
\langle (v_0,v_2,p),(v_0,v_2,p) \rangle_E=2v_0^2-2v_2^2+p^2.
\end{gather*}
The action of the Jacobi group $\tilde A_1$ in this extended space is
\begin{gather*}
\hat w_E(u,v,p,\tau)=(u,w(v),p,\tau),
\\
t_E(u,v,p,\tau)=\left(u-\langle \lambda,v \rangle_{E}-\frac{1}{2}\langle \lambda,\lambda \rangle_{E}\tau+k,v+p+\lambda\tau+\mu,\tau\right),
\\
\gamma_E(u,v,p,\tau)=\left(u+\frac{c\langle v,v \rangle_{E}}{2(c\tau+d)},\frac{v}{c\tau+d},\frac{p}{c\tau+d},\frac{a\tau+b}{c\tau+d}\right).
\end{gather*}

\begin{Proposition}
Let be $\varphi\in J_{k, m,\bullet}^{\Ja(\tilde A_1)}$, and $\hat\varphi$ the correspondent extended Jacobi form. Then,
\begin{gather*}
\frac{\partial}{\partial p}\big(\hat\varphi\big)\bigg|_{p=0} \in J_{k-1, m,\bullet}^{\Ja(\tilde A_1)}.
\end{gather*}
\end{Proposition}

\begin{proof}
$A_1$-\emph{invariant}.
The vector field $\frac{\partial}{\partial p}$ in coordinates $s$, $z_1$, $z_2$, $z_3$, $\tau$ reads
\begin{gather*}
\frac{\partial}{\partial p}=\frac{\partial}{\partial z_1}+\frac{\partial}{\partial z_2}+\frac{\partial}{\partial z_3}+2g_1(\tau)p\frac{\partial}{\partial u}.
\end{gather*}
Moreover, in the coordinates $s$, $z_1$, $z_2$, $z_3$, $\tau$ the $A_1$ group acts by permuting $z_1$ and $z_2$. Then
\begin{align*}
\frac{\partial}{\partial p}(\varphi(s,z_2,z_1,z_3,\tau))\bigg|_{p=0} &=\left(\frac{\partial}{\partial z_1}+\frac{\partial}{\partial z_2}+\frac{\partial}{\partial z_3}\right)( \varphi(s,z_2,z_1,z_3,\tau))\bigg|_{p=0}
\\
&=\left(\frac{\partial}{\partial z_1}+\frac{\partial}{\partial z_2}+\frac{\partial}{\partial z_3}\right) ( \varphi(s,z_1,z_2,z_3,\tau))\bigg|_{p=0}.
\end{align*}

\emph{Translation invariant:}
\begin{gather*}
\frac{\partial}{\partial p}\left(\varphi(u-\langle \lambda, v\rangle_E-\langle \lambda,\lambda \rangle_E,v+p+\lambda\tau+\mu,\tau)\right)\bigg |_{p=0}
\\ \qquad
{}=\frac{\partial}{\partial p} \langle \lambda, v\rangle_E \bigg |_{p=0}\varphi(u,v,\tau)+\frac{\partial \varphi}{\partial p}\left(u-\langle \lambda,v \rangle_{\tilde A_1}-\frac{1}{2}\langle \lambda,\lambda \rangle_{\tilde A_1}\tau+k,v+\lambda\tau+\mu,\tau\right)
\\ \qquad
{}=\frac{\partial \varphi}{\partial p}\left(u-\langle \lambda,v \rangle_{\tilde A_1}-\frac{1}{2}\langle \lambda,\lambda \rangle_{\tilde A_1}\tau+k,v+\lambda\tau+\mu,\tau\right)
=\frac{\partial \varphi}{\partial p}(u,v,\tau).
\end{gather*}

${\rm SL}_2(\mathbb{Z})$-\emph{equivariant of weight} $k$:
\begin{gather*}
\frac{\partial}{\partial p}\left(\varphi\left(u+\frac{c\langle v,v \rangle_{E}}{2(c\tau+d)},\frac{v}{c\tau+d},\frac{p}{c\tau+d},\frac{a\tau+b}{c\tau+d}\right) \right)\bigg |_{p=0}
\\ \qquad
=\frac{c}{2(c\tau\! +d)}\frac{\partial}{\partial p} \langle v, v\rangle_E \bigg |_{p=0}\!\varphi(u,v,\tau)\!+\!\frac{1}{c\tau\!+d}\frac{\partial \varphi}{\partial p}\left(\!u + \frac{c\langle v,v \rangle_{E}}{2(c\tau\!+d)},\frac{v}{c\tau\!+d},\frac{p}{c\tau\!+d},\frac{a\tau\!+b}{c\tau\!+d}\right)
\\ \qquad
=\frac{1}{c\tau+d}\frac{\partial \varphi}{\partial p}\left(u+\frac{c\langle v,v \rangle_{E}}{2(c\tau+d)},\frac{v}{c\tau+d},\frac{p}{c\tau+d},\frac{a\tau+b}{c\tau+d}\right)
=\frac{1}{(c\tau+d)^k}\frac{\partial \varphi}{\partial p}(u,v,\tau).
\end{gather*}
Then,
\begin{gather*}
\frac{\partial \varphi}{\partial p}\left(u+\frac{c\langle v,v \rangle_{E}}{2(c\tau+d)},\frac{v}{c\tau+d},\frac{p}{c\tau+d},\frac{a\tau+b}{c\tau+d}\right)
=\frac{1}{(c\tau+d)^{k-1}}\frac{\partial \varphi}{\partial p}(u,v,\tau).
\end{gather*}

\emph{Index} 1:
\begin{gather*}
\frac{1}{2\pi {\rm i}}\frac{\partial}{\partial u}\frac{\partial}{\partial p}\hat\varphi=\frac{1}{2\pi {\rm i}}\frac{\partial}{\partial p}\frac{\partial}{\partial u}\hat\varphi=\frac{\partial}{\partial p}\hat\varphi.
\tag*{\qed}
\end{gather*}
\renewcommand{\qed}{}
\end{proof}

\begin{Corollary}
The function
\begin{gather*}
\left[ {\rm e}^{z\frac{\partial}{\partial p} } \left({\rm e}^{2\pi {\rm i} u}\frac{\theta_1(v_0+v_2+p)\theta_1(-v_0+v_2+p)}{\theta_1(2v_2+p)\theta_1^{\prime}(0)}\right) \right]\bigg |_{p=0}=\varphi_1^{\Ja(\tilde A_1)}+\varphi_0^{\Ja(\tilde A_1)}z+O(z^2),
\end{gather*}
generates the Jacobi forms $\varphi_0^{\Ja(\tilde A_1)}$ and $\varphi_1^{\Ja(A_1)}$,
where
\begin{gather*}
\varphi_0^{\Ja(\tilde A_1)}:=\frac{\partial }{\partial p}\big(\hat\varphi_1^{\Ja(\tilde A_1)} \big)\bigg |_{p=0}.
\end{gather*}
\end{Corollary}

\begin{proof}
Acting $\frac{\partial}{\partial p}$ $k$ times in $\varphi_1^{\Ja(\tilde A_1)}$, we have
\begin{gather*}
\left[ \frac{\partial^k}{\partial ^k p}\left({\rm e}^{2\pi {\rm i} u}\frac{\theta_1(v_0+v_2+p)\theta_1(-v_0+v_2+p)}{\theta_1(2v_2+p)\theta_1^{\prime}(0)}\right) \right]
\bigg |_{p=0} \in J_{1-k,1,\bullet}^{\Ja(\tilde A_1)}.
\tag*{\qed}
\end{gather*}
\renewcommand{\qed}{}
\end{proof}

\begin{Corollary}
The generating function can be written as
\begin{gather*}
\left[ {\rm e}^{z\frac{\partial}{\partial p} } \left({\rm e}^{2\pi {\rm i} u}\frac{\theta_1(v_0+v_2+p)
\theta_1(-v_0+v_2+p)}{\theta_1(2v_2+p)\theta_1^{\prime}(0)}\right) \right]\bigg |_{p=0}\nonumber
\\ \qquad
{}={\rm e}^{-2\pi {\rm i}(u+ {\rm i}g_1(\tau)z^2)}\frac{\theta_1(z-v_0+v_2,\tau)\theta_1(z+v_0+v_2,\tau)}{\theta_1^{\prime}(0)\theta_1(z+2v_2)}.
\end{gather*}
\end{Corollary}

\begin{proof}
\begin{gather*}
\left[ {\rm e}^{z\frac{\partial}{\partial p} } \left({\rm e}^{2\pi {\rm i} u}\frac{\theta_1(v_0+v_2+p)
\theta_1(-v_0+v_2+p)}{\theta_1^{\prime}(0)\theta_1(2v_2+p)}\right) \right]\bigg |_{p=0}
\\ \qquad
{}=\left[ {\rm e}^{z\frac{\partial}{\partial p} } \left({\rm e}^{-2\pi {\rm i}(s+ {\rm i}g_1(\tau)p^2)}
\frac{\theta_1(v_0+v_2+p)\theta_1(-v_0+v_2+p)}{\theta_1(2v_2+p)\theta_1^{\prime}(0)}\right) \right]\bigg |_{p=0}
\\ \qquad
{}={\rm e}^{-2\pi {\rm i}(u+ {\rm i}g_1(\tau)z^2)}\frac{\theta_1(z-v_0+v_2,\tau) \theta_1(z+v_0+v_2,\tau)}{\theta_1^{\prime}(0)\theta_1(z+2v_2)}.\tag*{\qed}
\end{gather*}
\renewcommand{\qed}{}
\end{proof}

The next lemma is one of the main points of inquiry in this section, because this lemma identify the~orbit space of the group $\Ja\big(\tilde A_1\big)$ with the Hurwitz space $H_{1,0,0}$. This relationship is possible due~to the construction of the generating function of the Jacobi forms of type $\tilde A_1$, which can be~completed to be the Landau--Ginzburg superpotential of $H_{1,0,0}$ as follows:
\begin{gather*}
{\rm e}^{-2\pi {\rm i}(u+ {\rm i}g_1(\tau)z^2)}\frac{\theta_1(z-v_0+v_2,\tau)\theta_1(z+v_0+v_2,\tau)}
{\theta_1^{\prime}(0)\theta_1(z+2v_2)}
\\ \qquad{}
\mapsto {\rm e}^{-2\pi {\rm i} u}\frac{\theta_1(v-v_0+v_2,\tau) \theta_1(v+v_0+v_2,\tau)}{\theta_1(v\tau)\theta_1(v+2v_2,\tau)}.
\end{gather*}

\begin{Lemma}\label{jacobiform1 Jtildea1}
There exists a local isomorphism between $\Omega/\Ja\big(\tilde A_1\big)$ and $H_{1,0,0}$.
\end{Lemma}
\begin{proof}
The correspondence is realized by the map
\begin{gather}\label{lambda0 Jtildea1}
[(u,v_0,v_2,\tau)]\longleftrightarrow \lambda(v)={\rm e}^{-2\pi {\rm i} u}\frac{\theta_1(v-v_0,\tau)\theta_1(v+v_0,\tau)}{\theta_1(v-v_2,\tau)\theta_1(v+v_2,\tau)},
\end{gather}
where $\theta_1(v,\tau)$ is the Jacobi $\theta_1$-function defined on~\eqref{theta def}.

It is necessary to prove that the map is well defined and one to one.

\textit{Well defined.} Note that the map does not depend on the choice of the representative of $[(u,v_0,v_2,\tau)]$ if the function~\eqref{lambda0 Jtildea1} is invariant under the action of $\Ja\big(\tilde A_1\big)$. Therefore, let us prove the invariance of the map~\eqref{lambda0 Jtildea1}.

\textit{$A_1$-invariant.}
The $A_1$ group acts on~\eqref{lambda0 Jtildea1} by permuting its roots, thus~\eqref{lambda0 Jtildea1} remains invariant under this operation.

\textit{Translation invariant.}
 Recall that under the translation $v\mapsto v+m+n\tau$, the Jacobi $\theta$-function transforms as~\cite{E.T.WhittakerandG.N.Watson}:
\begin{gather}
\label{transformtheta Jtildea1 0}
\theta_1(v_i+\mu_i+\lambda_i\tau,\tau)=(-1)^{\lambda_i+\mu_i}{\rm e}^{-2\pi {\rm i}\big(\lambda_iv_i+\frac{\lambda_i^2}{2}\tau\big)}\theta_1(v_i,\tau).
\end{gather}
Then, substituting the transformation~\eqref{transformtheta Jtildea1 0} into~\eqref{lambda0 Jtildea1}, we conclude that~\eqref{lambda0 Jtildea1} remains inva\-riant.

\textit{${\rm SL}_2(\mathbb{Z})$-invariant.}
Under ${\rm SL}_2(\mathbb{Z})$-action the following function transforms to
\begin{gather}\label{transformtheta2 Jtildea1 0}
\frac{\theta_1\big(\frac{v_i}{c\tau+d},\frac{a\tau+d}{c\tau+d}\big)}
{\theta_1^{\prime}\big(0,\frac{a\tau+d}{c\tau+d}\big)}
=(c\tau+d)^{-1}\exp\left(\frac{\pi {\rm i} cv_i^2 }{c\tau+d}\right)\frac{\theta_1(v_i,\tau)}{\theta_1^{\prime}(0,\tau)}.
\end{gather}
Then, substituting the transformation~\eqref{transformtheta2 Jtildea1 0} into~\eqref{lambda0 Jtildea1}, we conclude that~\eqref{lambda0 Jtildea1} remains inva\-riant.

\textit{Injectivity.} Note that for fixed $v$, $v_0$, $v_2$, $u$, the function $\tau\mapsto f(\tau):=\lambda(v,v_0,v_2,u,\tau)$ is a~modu\-lar form with character~\cite{M.EichlerandD.Zagier}. This is clear because $\lambda(v,v_0,v_2,u,\tau)$ is rational function of~$\theta_1(z,\tau)$, which is modular form with character for special values of $z$~\cite{M.EichlerandD.Zagier}.
If $\lambda(v,v_0,v_2,u,\tau) =\lambda(v,\hat v_0,\hat v_2, \hat u, \hat \tau)$, then for fixed $v$, $v_0$, $v_2$, $u$, $\hat v_0$, $\hat v_2$, $\hat u$, we have $f(\tau)=f(\hat \tau)$; in particular, $f(\tau)$,~$f(\hat \tau)$~have the same vanishing order, and this implies that $\tau$, $\hat \tau$ belongs to the same ${\rm SL}_2(\mathbb{Z})$ orbit.

{\samepage
Two elliptic functions are equal if they have the same zeros and poles with multiplicity $\mod \mathbb{Z}\oplus \tau\mathbb{Z}$. So, for a fixed $\tau$ in the ${\rm SL}_2(\mathbb{Z})$ orbit
\begin{gather*}
\hat v_0=v_0+\lambda_0\tau+\mu_0,\qquad
\hat v_2=v_2+\lambda_2\tau+\mu_2,\qquad
(\lambda_i,\mu_i)\in \mathbb{Z}^2.
\end{gather*}
Furthermore, for two different representations of the same ${\rm SL}_2(\mathbb{Z})$ orbit, but considering fixed cells, we have
\begin{gather*}
\hat v_0=\frac{v_0}{c\tau+d},\qquad
\hat v_2=\frac{v_2}{c\tau+d},\qquad
\hat \tau=\frac{a\tau+b}{c\tau+d},
\end{gather*}
where $\left(\begin{smallmatrix}
a & b \\
c & d
\end{smallmatrix}\right)\in {\rm SL}_2(\mathbb{Z})$.

}

Since, $\lambda(v,v_0,v_2,u,\tau)$ is invariant under translations, and ${\rm SL}_2(\mathbb{Z})$, for $\hat\tau=\tau$, we have
\begin{gather*}
\hat u=u-\langle \lambda,v\rangle_{\tilde A_1}-\langle \lambda, \lambda \rangle_{\tilde A_1}\frac{\tau}{2}+k.
\end{gather*}
For $\hat \tau=\frac{a\tau+b}{c\tau+d}$,
\begin{gather*}
\hat u=u-\frac{c\langle v, v \rangle_{\tilde A_1}}{2(c\tau+d)}+k,
\end{gather*}
where $k\in \mathbb{Z}$.

\textit{Surjectivity.}
 Any elliptic function can be written as rational functions of Weierstrass $\sigma$-function up to a multiplication factor~\cite{E.T.WhittakerandG.N.Watson}; by using the formula
\begin{gather*}
\sigma(v_i,\tau)=\frac{\theta_1(v_i,\tau)}{\theta_1^{\prime}(0,\tau)}\exp\big({-}2\pi {\rm i}g_1(\tau)v_i^2\big),\qquad
g_1(\tau)=\frac{\eta^{\prime}(\tau)}{\eta(\tau)},
\end{gather*}
where $\eta(\tau)$ is the Dedekind $\eta$-function, we get the desire result.
\end{proof}

\begin{Remark}
Lemma \ref{jacobiform1 Jtildea1} is a local biholomorphism of manifolds, but this does not neces\-sa\-rily means isomorphism of Dubrovin--Frobenius structure. On a Hurwitz space, there may exist in several inequivalent Dubrovin--Frobenius structures. For instance, in~\cite{Romano2} Romano constructed two gene\-ra\-li\-sed WDDV solution on the Hurwitz space $H_{1,0,0}$. Furthermore, in~\cite{BertolaM.1} and~\cite{BertolaM.2}, Bertola constructed two different Dubrovin--Frobenius structures on the orbit space of the Jacobi group~$G_2$. The Dubrovin--Frobenius structure of this orbit space will be constructed only in Section~\ref{Frobenius structure A1}.
\end{Remark}

\begin{Remark}
Lemma \ref{jacobiform1 Jtildea1} associates a group to $H_{1,0,0}$, and this could be useful for the general understanding of the WDDV solutions/discrete group correspondence~\cite{B.Dubrovin2}.
\end{Remark}

\begin{Corollary}
The functions $\big(\varphi_0^{\tilde A_1},\varphi_1^{\tilde A_1}\big)$ obtained by the formula
\begin{gather}
\lambda^{\tilde A_1}={\rm e}^{-2\pi {\rm i} u}\frac{\theta_1(v-v_0,\tau)\theta_1(v+v_0,\tau)}
{\theta_1(v-v_2,\tau)\theta_1(v+v_2,\tau)}\nonumber
\\ \hphantom{\lambda^{\tilde A_1}}
{}=\varphi_{1}^{\tilde A_1}\left[\zeta(v-v_2,\tau)-\zeta(v+v_{2},\tau)+2\zeta(v_2,\tau)\right]+\varphi_0^{\tilde A_1}
\label{superpotentialA1}
\end{gather}
 are Jacobi forms of weight $0$, $-1$ respectively, index~$1$, and order $0$. More explicitly,
\begin{gather}
\varphi_{1}^{\tilde A_1}=\frac{\theta_1(v_0+v_2,\tau)\theta_1(-v_0+v_2,\tau)}{\theta^{\prime}_1(0,\tau)\theta_1(2v_2,\tau)}{\rm e}^{-2\pi {\rm i} u},\nonumber
\\
\varphi_{0}^{\tilde A_1}=-\varphi_{1}^{\tilde A_1}\left[\zeta(v_0-v_2,\tau)-\zeta(v_0+v_{2},\tau)+2\zeta(v_2,\tau)\right],
\label{thetaformulainvariant}
\end{gather}
where $\zeta(v,\tau)$ is the Weierstrass $\zeta$-function for the lattice $(1,\tau)$, i.e.,
\begin{gather*}
\zeta(v,\tau)=\frac{1}{v}+\sum_{m^2+n^2\neq 0}^{\infty} \frac{1}{v-m-n\tau}+\frac{1}{m+n\tau}+\frac{v}{(m+n\tau)^2}.
\end{gather*}
\end{Corollary}
\begin{proof}
Let us prove each item separately.

\textit{$A_1$-invariant, translation invariant.}
The first line of~\eqref{superpotentialA1} are $A_1$-invariant, and translation invariant by the lemma~\eqref{jacobiform1 Jtildea1}. Then, by the Laurent expansion of $\lambda^{\tilde A_1}$, we have that $\varphi_i^{\tilde A_1}$ are $A_1$-invariant, and translation invariant.

\textit{${\rm SL}_2(\mathbb{Z})$-equivariant.}
The first line of~\eqref{superpotentialA1} are ${\rm SL}_2(\mathbb{Z})$-invariant, but the Weierstrass $\zeta$-functions of the second line of~\eqref{superpotentialA1} have the following transformation law
\begin{gather*}
\zeta\left(\frac{z}{c\tau+d},\frac{a\tau+b}{c\tau+d}\right)=(c\tau+d)\zeta(z,\tau).
\end{gather*}
Then, $\varphi_i^{\tilde A_1}$ must have the following transformation law:
\begin{gather*}
\varphi_0^{\tilde A_1}\left(u+\frac{c\langle v,v \rangle_{\tilde A_1}}{2(c\tau+d)},\frac{v}{c\tau+d},\frac{a\tau+b}{c\tau+d}\right)=\varphi_0^{\tilde A_1}(u,v,\tau),
\\
\varphi_1^{\tilde A_1}\left(u+\frac{c\langle v,v\rangle_{\tilde A_1}}{2(c\tau+d)},\frac{v}{c\tau+d},\frac{a\tau+b}{c\tau+d}\right)=(c\tau+d)^{-k}\varphi_1^{\tilde A_1}(u,v,\tau).
\end{gather*}

\textit{Index $1$:}
\begin{gather*}
\frac{1}{2\pi {\rm i}}\frac{\partial}{\partial u}\lambda^{\tilde A_1}=\lambda^{\tilde A_1}.
\end{gather*}
Then
\begin{gather*}
\frac{1}{2\pi {\rm i}}\frac{\partial}{\partial u}\varphi_i^{\tilde A_1}=\varphi_i^{\tilde A_1}.
\end{gather*}

\textit{Analytic behavior.}
Note that $\lambda^{\tilde A_1}\theta_1^2(2v_2,\tau)$ is holomorphic function in all the variables $v_i$. Therefore, $\varphi_i^{\tilde A_1}$ are holomorphic functions on the variables $v_0$, and meromorphic function in the variable $v_{2}$ with poles on $\frac{j}{2}+\frac{l\tau}{2}$, $j,l=0,1$ of order 2, i.e.,~$l=0$, since $m=1$ for all~$\varphi_i^{\tilde A_1}$.

To prove the formula~\eqref{thetaformulainvariant} let us compute the following limit:
\begin{gather*}
\lim_{z\to v_2}\lambda^{\tilde A_1}v_2=\varphi_1^{\tilde A_1}
={\rm e}^{-2\pi {\rm i} u}\frac{\theta_1(v_0+v_2,\tau)\theta_1(-v_0+v_2,\tau)}
{\theta^{\prime}_1(0,\tau)\theta_1(2v_2,\tau)}.
\end{gather*}
Let us also compute the zeros of $\lambda^{\tilde A_1}$
\begin{gather*}
\lambda^{\tilde A_1}(v_0)=0=\varphi_{1}^{\tilde A_1}\left[\zeta(v_0-v_2,\tau)-\zeta(v_0+v_{2},\tau)+2\zeta(v_2,\tau)\right]+\varphi_0^{\tilde A_1}.
\tag*{\qed}
\end{gather*}
\renewcommand{\qed}{}
\end{proof}

\begin{Lemma}The functions $\varphi_0^{\tilde A_1}$, $\varphi_1^{\tilde A_1}$ are algebraically independent over the ring $E_{\bullet,\bullet}$.
\end{Lemma}
\begin{proof}
If $P(X,Y)$ is any polynomial in $E_{\bullet,\bullet}(X,Y)$, such that $P\big(\varphi_0^{\tilde A_1}$, $\varphi_1^{\tilde A_1}\big)=0$, then, the fact $\varphi_0^{\tilde A_1}$, $\varphi_1^{\tilde A_1}$ have an index that implies that each homogeneous component $P_d\big(\varphi_0^{\tilde A_1}, \varphi_1^{\tilde A_1}\big)$ has to vanish identically. Defining $p_d\left(\frac{\varphi_0^{\tilde A_1}}{\varphi_1^{\tilde A_1}}\right):=\frac{P_d\big(\varphi_0^{\tilde A_1}, \varphi_1^{\tilde A_1}\big)}{{\big(\varphi_1^{\left(\tilde A_1\right)}\big)}^d}$, we have that $p_d\left(\frac{\varphi_0^{\tilde A_1}}{\varphi_1^{\tilde A_1}}\right)$ is identically $0$ iff $\frac{\varphi_0^{\tilde A_1}}{\varphi_1^{\tilde A_1}}$ is constant (belongs to $E_{\bullet,\bullet}$), but
\begin{gather*}
\frac{\varphi_0^{\tilde A_1}}{\varphi_1^{\tilde A_1}}=\frac{\wp^{\prime}(v_2,\tau)}{\wp(v_0,\tau)-\wp(v_2,\tau)}\neq a(v_2,\tau),
\end{gather*}
where $a(v_2,\tau)$ is any function belongs to $E_{\bullet,\bullet}$.
Then, $\varphi_0^{\tilde A_1}$, $\varphi_1^{\tilde A_1}$ are algebraically independent over the ring $E_{\bullet,\bullet}$.

Recall that $\wp(v,\tau)$ is the Weierstrass $\rm P$-function~\eqref{Weiestrass P 0}.
\end{proof}

Consider the Landau--Ginzburg superpotential~\eqref{generatorsAn+10cor} for the $\Ja(A_2)$ case below.

\begin{Theorem}[\cite{BertolaM.1}]
The ring of $A_{2}$-invariant Jacobi forms is free module of rank $3$ over the ring of modular forms, moreover there exist a formula for its generators given by
\begin{gather}
\lambda^{A_{2}}={\rm e}^{-2\pi {\rm i} u_2}\frac{ \theta_1(z+v_0+v_{2},\tau)\theta_1(z-v_0+v_{2},\tau) \theta_1(z-2v_{2})}{\theta_1^{3}(z,\tau)}\nonumber
\\ \hphantom{\lambda^{A_{2}}}
=-\frac{\varphi_{3}^{A_{2}}}{2}\wp^{\prime}(z,\tau)+\varphi_{2}^{A_{2}}\wp(z,\tau)+\varphi_{0}^{A_{2}}.
\label{generatorsAn+10cor}
\end{gather}
\end{Theorem}

\begin{Lemma}
{\sloppy
Let $\big\{\varphi_0^{\tilde A_1 },\varphi_1^{\tilde A_1 }\big\}$ be set of functions given by the formula~\eqref{superpotentialA1}, and $\big\{\varphi_0^{A_{2}},\varphi_2^{A_{2}},\varphi_3^{A_{2}}\big\}$ given by~\eqref{generatorsAn+10cor}, then
\begin{gather*}
\varphi_{3}^{A_{2}}=\varphi_{1}^{\tilde A_{1}}\varphi_{2}^{A_{1}},\\
\varphi_{2}^{A_{2}}=\varphi_{0}^{\tilde A_{n}}\varphi_{2}^{A_{1}}+ a_2(v_2,\tau)\varphi_{j}^{\tilde A_{n}}\varphi_{2}^{A_{1}},\\
\varphi_{0}^{A_{2}}= a_0(v_2,\tau)\varphi_{0}^{\tilde A_{1}}\varphi_{2}^{A_{1}}+b_0(v_2,\tau)\varphi_{2}^{\tilde A_{1}}\varphi_{2}^{A_{1}},
\end{gather*}}
where
\begin{gather*}
\varphi_2^{A_1}:=\frac{\theta_1^2(2v_{2},\tau)}{{\theta_1^{\prime}(0,\tau)}^2}{\rm e}^{2\pi {\rm i} (-u_2+u_1)}
\end{gather*}
and $a_i$, $b_i$ are elliptic functions on $v_{2}$.

\end{Lemma}
\begin{proof}
Note the following relation
\begin{gather*}
\frac{\lambda^{A_{2}}}{\lambda^{\tilde A_1}}=\frac{\theta_1(z-2v_{2},\tau)\theta_1(z+2v_{2}),\tau} {\theta_1^2(z,\tau)}{\rm e}^{2\pi {\rm i}(-u_2+u_1)}
=\varphi_2^{A_1}\wp(z,\tau)-\varphi_2^{A_1}\wp(2v_2,\tau).
\end{gather*}
Hence,
\begin{gather*}
-\frac{\varphi_{3}^{A_{2}}}{2}\wp^{\prime}(z,\tau)+\varphi_{2}^{A_{2}}\wp(z,\tau)+\varphi_{0}^{A_{2}}
=\big(\varphi_{1}^{\tilde A_1}[\zeta(z,\tau)-\zeta(z+2v_{2},\tau)+2\zeta(v_2,\tau)]+\varphi_0^{\tilde A_n}\big)
\\ \hphantom{-\frac{\varphi_{3}^{A_{2}}}{2}\wp^{\prime}(z,\tau)+ \varphi_{2}^{A_{2}}\wp(z,\tau)+\varphi_{0}^{A_{2}}=}
{}\times\big(\varphi_2^{A_1}\wp(z,\tau)-\varphi_2^{A_1}\wp(2v_{2},\tau)\big).
\end{gather*}
Then, the desired result is obtained by doing a Laurent expansion in the variable $z$ on both sides of the equality.
\end{proof}

\begin{Corollary}
\begin{gather*}
E_{\bullet,\bullet}\big[\varphi_0^{\tilde A_1 },\varphi_1^{\tilde A_1 }\big]=E_{\bullet,\bullet}\left[\frac{\varphi_0^{A_{2}}}{\varphi_2^{A_1}}, \frac{\varphi_2^{A_{2}}}{\varphi_2^{A_1}},\frac{\varphi_3^{A_{2}}}{\varphi_2^{A_1}}\right].
\end{gather*}
\end{Corollary}

Moreover, we have the following lemma:
\begin{Lemma}
Let be $\varphi\in J_{\bullet,\bullet,m}^{\tilde A_1}$, then $\varphi \in E_{\bullet,\bullet}\left[\frac{\varphi_0^{A_{2}}}{\varphi_2^{A_1}}, \frac{\varphi_2^{A_{2}}}{\varphi_2^{A_1}},\frac{\varphi_3^{A_{2}}}{\varphi_2^{A_1}}\right]$.
\end{Lemma}

\begin{proof}
Let be $\varphi\in J_{\bullet,\bullet,m}^{\tilde A_1}$, then the function $\frac{\varphi}{\big(\varphi_1^{\tilde A_1}\big)^m}$ is an elliptic function on the variables $(v_0,v_2)$ with poles on $v_0-v_2$, $v_0+v_2$, $2v_2$ due to the zeros of $\varphi_1^{\tilde A_1}$ and the poles of $\varphi$, which are by definition in $2v_2$. Expanding the function $\frac{\varphi}{\big(\varphi_1^{\tilde A_1}\big)^m}$ in the variables $v_0$, $v_2$ we get{\samepage
\begin{gather}
\label{eq lemma 1 for chevalley}
\frac{\varphi}{\big(\varphi_1^{\tilde A_1}\big)^m}=\sum_{i=-1}^{m} a^i\wp^{(i)}(v_0+v_{2})+\sum_{i=-1}^{m} b^i\wp^{(i)}(-v_0+v_{2})+c(v_2,\tau),
\end{gather}
where $\wp^{-1}(v):=\zeta(v)$, and $c(v_2,\tau)$ is an elliptic function in the variable $v_2$.}

However, the function $\frac{\varphi}{\big(\varphi_1^{\tilde A_1}\big)^m}$ is invariant under the permutations of the variables $v_0$, so the equation~\eqref{eq lemma 1 for chevalley} is
\begin{gather}
\label{eq lemma 2 for chevalley}
\frac{\varphi}{\big(\varphi_1^{\tilde A_1}\big)^m}=\sum_{i=-1}^{m} a^i\big(\wp^{(i)}(v_0+v_{2})+\wp^{(i)}(-v_0+v_{2})\big)+c(v_2,\tau).
\end{gather}

Now we complete this function to $A_{2}$-invariant function by summing and subtracting the following function in equation~\eqref{eq lemma 2 for chevalley}
\begin{gather*}
f(v_{2},\tau)=\sum_{i=-1}^{m} a^i\wp^{(i)}(2v_2).
\end{gather*}
Hence,
\begin{gather}
\label{eq lemma 3 for chevalley}
\frac{\varphi}{\big(\varphi_1^{\tilde A_1}\big)^m}=\sum_{i=-1}^{m} a^i\big(\wp^{(i)}(v_0+v_{2})+\wp^{(i)}(-v_0+v_{2})+\wp^{(i)}(2v_2)\big)+g(v_2,\tau).
\end{gather}

Multiplying both side of the equation~\eqref{eq lemma 3 for chevalley} by $\varphi_1^{A_1}$, we get
\begin{gather*}
\varphi= \left(\sum_{i=-1}^{m}\! a^i\big(\wp^{(i)}(v_0+v_{2})+\wp^{(i)}(-v_0+v_{2})+\wp^{(i)}(2v_2)\big)\! \right) \big(\varphi_3^{A_2}\big)^m +g(v_2,\tau) \big(\varphi_3^{A_2}\big)^m\!.\!\!
\end{gather*}
To finish the proof, we will show that
\begin{gather*}
 \left(\sum_{i=-1}^{m} a^i\big(\wp^{(i)}(v_0+v_{2})+\wp^{(i)}(-v_0+v_{2})+\wp^{(i)}(2v_2)\big) \right) \big(\varphi_3^{A_2}\big)^m
\end{gather*}
is a weak holomorphic Jacobi form of type $A_2$. To finish the proof, note the following:
\begin{enumerate}\itemsep=0pt
\item The functions
\begin{gather}\label{gather of last chevalley proof}
\big(\varphi_{3}^{A_2}\big)^m\big(\wp^{(i)}(v_0+v_{2})+\wp^{(i)}(-v_0+v_{2})+\wp^{(i)}(2v_2)\big)
\end{gather}
 are $A_{2}$-invariant by construction.
\item The functions~\eqref{gather of last chevalley proof} are invariant under the action of $(\mathbb{Z}\oplus\tau\mathbb{Z})^{2}$, because $\varphi_{3}^{A_{2}}$ is invariant, and
\begin{gather}\label{gather of last chevalley proof 1}
\wp^{(i)}(v_0+v_{2})+\wp^{(i)}(-v_0+v_{2})+\wp^{(i)}(2v_2)
\end{gather}
 are elliptic functions.
\item The functions~\eqref{gather of last chevalley proof} are equivariant under the action of ${\rm SL}_2(\mathbb{Z})$, because $\varphi_{3}^{A_{2}}$ is equivariant, and
\eqref{gather of last chevalley proof 1} are elliptic functions.
\item The function $\varphi_{3}^{A_{2}}$ has zeros on~$v_0-v_2$, $v_0+v_2$, $2v_2$ of order $m$, and
\eqref{gather of last chevalley proof 1} has poles on~$v_0-v_2$, $v_0+v_2$, $2v_2$ of order $i+2\leq m$. So, the functions~\eqref{gather of last chevalley proof} are holomorphic.
\end{enumerate}
Hence,
\begin{gather*}
\varphi \in E_{\bullet,\bullet}\left[\frac{\varphi_0^{A_{2}}} {\varphi_2^{A_1}},\frac{\varphi_2^{A_{2}}}{\varphi_2^{A_1}},\frac{\varphi_3^{A_{2}}}{\varphi_2^{A_1}}\right].
\tag*{\qed}
\end{gather*}\renewcommand{\qed}{}
\end{proof}

At this stage, the principal theorem can be stated in a precise way as follows.

\begin{Theorem}\label{chevalley Jtildea1}
The trigraded algebra of Jacobi forms $J_{\bullet,\bullet,\bullet}^{\Ja(\tilde A_1)}=\bigoplus _{k,l,m}J_{k,l,m}^{\tilde A_1}$ is freely generated by $2$ fundamental Jacobi forms $\big(\varphi_0^{\tilde A_1},\varphi_1^{\tilde A_1}\big)$ over the graded ring $E_{\bullet,\bullet}$
\begin{gather*}
J_{\bullet,\bullet,\bullet}^{\Ja(\tilde A_1)}=E_{\bullet,\bullet}\left [\varphi_0^{\tilde A_1},\varphi_1^{\tilde A_1}\right ].
\end{gather*}
\end{Theorem}
\begin{proof}
\begin{gather*}
J_{\bullet,\bullet,\bullet}^{\tilde A_1}\subset E_{\bullet,\bullet}\left [\frac{\varphi_0^{A_{2}}}{\varphi_2^{A_1}},\frac{\varphi_2^{A_{2}}}{\varphi_2^{A_1}},\frac{\varphi_3^{A_{2}}}{\varphi_2^{A_1}}\right ]=E_{\bullet,\bullet}\left [\varphi_0^{\tilde A_1},\varphi_1^{\tilde A_1}\right ]\subset J_{\bullet,\bullet,\bullet}^{\tilde A_1}.
\tag*{\qed}
\end{gather*}
\renewcommand{\qed}{}
\end{proof}

\begin{Remark}
The structural difference between the Chevalley theorems of the groups $J(A_1)$ and $\Ja\big(\tilde A_1\big)$ lies in the ring of coefficients. The ring of coefficients of Jacobi forms with respect to $J(A_1)$ are modular forms, and the ring of coefficients of Jacobi forms with respect to~$\Ja\big(\tilde A_1\big)$
are the ring of elliptic functions with poles on $0, \frac{1}{2}, \frac{\tau}{2}, \frac{1+\tau}{2} \mod \mathbb{Z}\oplus\tau\mathbb{Z}$, for fixed~$\tau$.
See Lemma~\ref{ringcoef Jtildea1}.
\end{Remark}

\begin{Remark}The geometry of $\Omega^{\Ja(\tilde A_1)}/\Ja\big(\tilde A_1\big)$ is similar to $\Omega^{\Ja( A_1)}/\Ja( A_1)$. Indeed, the orbit space of $\Ja\big(\tilde A_1\big)$ is locally a line bundle over a family of two elliptic curves, $E_{\tau}/A_1\otimes E_{\tau}$, where the first one is quotient by $A_1$, and both are parametrised by $\mathbb{H}/{\rm SL}_2(\mathbb{Z})$.
\end{Remark}

\section[Frobenius structure on the orbit space of Ja(A1)]{Frobenius structure on the orbit space of $\boldsymbol{\Ja\big(\tilde A_1\big)}$}\label{Frobenius structure A1}

In this section, a Dubrovin--Frobenius manifold structure will be constructed on the orbit space of~$\Ja\big(\tilde A_1\big)$. More precisely, it will define the data $\big(\Omega^{\Ja(\tilde A_1}/\Ja\big(\tilde A_1\big),g^{*},e,E\big)$, with the intersection form $g^*$, unit vector field~$e$, and Euler vector field~$E$. This data will be written naturally in terms of the invariant functions of $\Ja\big(\tilde A_1\big)$. Thereafter, it will be proved that this data is enough to the construction of the Dubrovin--Frobenius structure.

\subsection{Intersection form}

The first step to be done is to construct the intersection form. It will be shown that such metric can be constructed by using just the data of the group $\Ja\big(\tilde A_1\big)$. The strategy is to~combine the intersection form of the group $\tilde A_1$ and $\Ja(A_1)$. Recall that the intersection form of the group~$\tilde A_1$~\cite{B.Dubrovin2, B.DubrovinandY.Zhang} is
\begin{gather*}
{\rm d}s^2=2{\rm d}v_0^2-2{\rm d}v_2^2,
\end{gather*}
and the intersection form of $\Ja(A_1)$~\cite{BertolaM.1, BertolaM.2, B.Dubrovin2} is
\begin{gather*}
{\rm d}s^2={\rm d}v_0^2+2{\rm d}u {\rm d}\tau.
\end{gather*}
Therefore, the natural candidate to be the intersection form of $\Ja\big(\tilde A_1\big)$ is
\begin{gather*}
{\rm d}s^2=2{\rm d}v_0^2-2{\rm d}v_2^2+2{\rm d}u {\rm d}\tau.
\end{gather*}
The following lemma proves that this metric is invariant metric of the group $\Ja\big(\tilde A_1\big)$. To be precise, the metric will be invariant under the action of $A_1$, and translations, and equivariant under the action of ${\rm SL}_2(\mathbb{Z})$.

\begin{Lemma}
The metric
\begin{gather}\label{metrich100}
{\rm d}s^2=2{\rm d}v_0^2-2{\rm d}v_2^2+2{\rm d}u {\rm d}\tau
\end{gather}
is invariant under the transformations~\eqref{mon 1}, \eqref{mon 2}.
Moreover, the transformations~\eqref{mon 3} deter\-mine a conformal transformation of the metric ${\rm d}s^2$, i.e.,
\begin{gather*}
2{\rm d}v_0^2-2{\rm d}v_2^2+2{\rm d}u {\rm d}\tau\mapsto \frac{2{\rm d}v_0^2-2{\rm d}v_2^2+2{\rm d}u {\rm d}\tau}{(c\tau+d)^2}.
\end{gather*}
\end{Lemma}

\begin{proof}Under~\eqref{mon 1}, \eqref{mon 2}, the differentials transform to
\begin{gather*}
{\rm d}v_0\mapsto-{\rm d}v_0,\qquad
{\rm d}v_0\mapsto {\rm d}v_0+\lambda_0{\rm d}\tau,\qquad
{\rm d}v_2\mapsto {\rm d}v_2+\lambda_2{\rm d}\tau,\\
{\rm d}u\mapsto {\rm d}u-\lambda_0^2{\rm d}\tau-2\lambda_0{\rm d}v_0+\lambda_2^2{\rm d}\tau+2\lambda_2{\rm d}v_2,\qquad
{\rm d}\tau\mapsto {\rm d}\tau.
\end{gather*}
Hence,
\begin{gather*}
{\rm d}v_0^2\mapsto dv^2_0,\qquad
{\rm d}v_0^2\mapsto {\rm d}v_0^2+2\lambda_0{\rm d}v_0{\rm d}\tau+\lambda_0^2{\rm d}\tau^2,\qquad
{\rm d}v_2^2\mapsto {\rm d}v_2^2+2\lambda_2{\rm d}v_2{\rm d}\tau+\lambda_2^2{\rm d}\tau^2,\\
2{\rm d}u {\rm d}\tau\mapsto 2{\rm d}u {\rm d}\tau-2\lambda_0^2{\rm d}\tau^2-4\lambda_0{\rm d}v_0{\rm d}\tau+2\lambda_2^2{\rm d}\tau^2+4\lambda_2{\rm d}v_2{\rm d}\tau.
\end{gather*}
So,
\begin{gather*}
2{\rm d}v_0^2-2{\rm d}v_2^2+2{\rm d}u {\rm d}\tau\mapsto 2{\rm d}v_0^2-2{\rm d}v_2^2+2{\rm d}u {\rm d}\tau.
\end{gather*}
Let us show that the metric has conformal transformation under the transformations~\eqref{mon 3}
\begin{gather*}
{\rm d}v_0\mapsto \frac{{\rm d}v_0}{c\tau+d}-\frac{v_0{\rm d}\tau}{(c\tau+d)^2},\qquad
{\rm d}v_2\mapsto \frac{{\rm d}v_2}{c\tau+d}-\frac{v_2{\rm d}\tau}{(c\tau+d)^2},\qquad
{\rm d}\tau\mapsto \frac{{\rm d}\tau}{(c\tau+d)^2},\\
{\rm d}u\mapsto {\rm d}u+\frac{c(2v_0{\rm d}v_0-2v_2{\rm d}v_2^2)}{c\tau+d}-\frac{c(v_0^2-v_2^2){\rm d}\tau}{(c\tau+d)^2}.
\end{gather*}
Then,
\begin{gather*}
{\rm d}v_0^2\mapsto \frac{{\rm d}v_0^2}{(c\tau+d)^2}-\frac{2v_0{\rm d}v_0{\rm d}\tau}{(c\tau+d)^3}+\frac{v_0^2{\rm d}\tau^2}{(c\tau+d)^4},\qquad
{\rm d}v_2^2\mapsto \frac{{\rm d}v_2^2}{(c\tau+d)^2}-\frac{2v_2{\rm d}v_2{\rm d}\tau}{(c\tau+d)^3}+\frac{v_2^2{\rm d}\tau^2}{(c\tau+d)^4},\\
2{\rm d}u {\rm d}\tau \mapsto \frac{2{\rm d}u {\rm d}\tau}{(c\tau+d)^2}+\frac{c(4v_0{\rm d}v_0-4v_2{\rm d}v_2){\rm d}\tau}{(c\tau+d)^3}-\frac{c(2v_0^2-2v_2^2){\rm d}\tau^2}{(c\tau+d)^4}.
\end{gather*}
Then,
\begin{gather*}
2{\rm d}v_0^2-2{\rm d}v_2^2+2{\rm d}u {\rm d}\tau\mapsto \frac{2{\rm d}v_0^2-2{\rm d}v_2^2+2{\rm d}u {\rm d}\tau}{(c\tau+d)^2}.
\tag*{\qed}
\end{gather*}
\renewcommand{\qed}{}
\end{proof}

\subsection{Euler and unit vector field}

The next step is to construct a two vector field, which is a intrinsic object of the orbit spa\-ce~$\Ja\big(\tilde A_1\big)$. The first one is the Euler vector
\begin{gather}\label{eulerh100}
E=-\frac{1}{2\pi {\rm i}}\frac{\partial}{\partial u},
\end{gather}
which was already defined in the last equation of~\eqref{jacobiform}. Therefore, it is an already intrinsic object, since it comes from the definition of meromorphic Jacobi forms associated with $\Ja\big(\tilde A_1\big)$. In the invariant coordinates, the vector field~\eqref{eulerh100} reads as
\begin{gather*}
E=\varphi_0\frac{\partial}{\partial \varphi_0}+\varphi_1\frac{\partial}{\partial \varphi_1}.
\end{gather*}
The second one is given by the coordinates $(\varphi_0,\varphi_1,v_2,\tau)$ as
\begin{gather}\label{unith100}
e=\frac{\partial}{\partial \varphi_0},
\end{gather}
and it is denoted by the unit vector field. This object is intrinsic to the orbit space of $\Ja\big(\tilde A_1\big)$, because it is written in terms of a meromorphic Jacobi forms associated with $\Ja\big(\tilde A_1\big)$.

\subsection{Flat coordinates of the Saito metric}

In order to construct the Dubrovin--Frobenius structure, it will be necessary to introduce the coordinates $\big(t^1,t^2,t^3,t^4\big)$.

\begin{Lemma}
There is a change of coordinates in $\Omega^{\Ja(\tilde A_1)}/\Ja\big(\tilde A_1\big)$, given by
\begin{gather*}
t^1=\varphi_0+2t^2\frac{\theta_1^{\prime}(v_2|\tau)}{\theta_1(v_2|\tau)},\qquad
t^2=\varphi_1,\qquad
t^3=v_2,\qquad
t^4=\tau.
\end{gather*}
\end{Lemma}
\begin{proof}
Note that the function~\eqref{lambda0 Jtildea1} can be parametrised by $\big(t^1,t^2,t^3,t^4\big)$ as follows:
\begin{gather*}
\lambda=\varphi_0+\varphi_1[\zeta(v-v_2|\tau)-\zeta(v+v_2|\tau)+2\zeta(v_2)]
\\ \hphantom{\lambda}
{}=\varphi_0+\varphi_1\left[\frac{\theta_1^{\prime}(v-v_2|\tau)}{\theta_1(v-v_2|\tau)} -\frac{\theta_1^{\prime}(v+v_2|\tau)}{\theta_1(v+v_2|\tau)} +2\frac{\theta_1^{\prime}(v_2|\tau)}{\theta_1(v_2|\tau)}\right]
\\ \hphantom{\lambda}
{}=\varphi_0+2\frac{\theta_1^{\prime}(v_2|\tau)}{\theta_1(v_2|\tau)} +\varphi_1\left[\frac{\theta_1^{\prime}(v-v_2|\tau)}{\theta_1(v-v_2|\tau)} -\frac{\theta_1^{\prime}(v+v_2|\tau)}{\theta_1(v+v_2|\tau)}\right]
\\ \hphantom{\lambda}
{}=t^1+t^2\left[\frac{\theta_1^{\prime}\big(v-t^3|t^4\big)}{\theta_1\big(v-t^3|t^4\big)} -\frac{\theta_1^{\prime}\big(v+t^3|t^4\big)}{\theta_1\big(v+t^3|t^4\big)}\right]
\end{gather*}
from the first line to the second line, the following equation was used:
\begin{gather*}
\zeta(v-v_2,\tau)=\frac{\theta_1^{\prime}(v-v_2|\tau)}{\theta_1(v-v_2|\tau)}+4\pi {\rm i}g_1(\tau)(v-v_2).
\end{gather*}
 In this way, $\big(t^1,t^2,t^3,t^4\big)$ are local coordinates of $\Omega^{\Ja(\tilde A_1)}/\Ja\big(\tilde A_1\big)$ due to Lemma~\ref{jacobiform1 Jtildea1}.
\end{proof}

The side back effect of the coordinates $\big(t^1,t^2,t^3,t^4\big)$ is the fact that they are not globally single valued functions on the quotient.

\begin{Lemma}
The coordinates $\big(t^1,t^2,t^3,t^4\big)$ have the following transformation laws under the action of the group $\Ja\big(\tilde A_1\big){:}$ they are invariant under~\eqref{mon 1}.
They transform as follows under~\eqref{mon 2}$:$
\begin{gather*}
t^1\mapsto t^1-\lambda_2t^2,\qquad
t^2\mapsto t^2,\qquad
t^3\mapsto t^3+\mu_2+\lambda_2t^4,\qquad
t^4\mapsto t^4.
\end{gather*}

Moreover, they transform as follows under~\eqref{mon 3}
\begin{gather*}
t^1\mapsto t^1+\frac{2ct^2t^3}{ct^4+d},\qquad
t^2\mapsto \frac{t^2}{ct^4+d},\qquad
t^3\mapsto \frac{t^3}{ct^4+d},\qquad
t^4\mapsto \frac{at^4+b}{ct^4+d}.
\end{gather*}
\end{Lemma}

\begin{proof}
The invariance under~\eqref{mon 1} is clear, since only $t^1$ depends on $v_0$, and its dependence is given by $\varphi_0$, which is invariant under~\eqref{mon 1}. Let us check how $t^{\alpha}$ transforms under~\eqref{mon 2}, \eqref{mon 3}: Since $t^3=v_2$, $t^4=\tau$, we have the desired transformations law defined as~$\Ja\big(\tilde A_1\big)$. The coordinate $t^2=\varphi_1$ is a invariant under~\eqref{mon 2} and transforms as modular form of~weight~$-1$ under~\eqref{mon 2}. The only non-trivial term is $t^1$, because it contains the term $\frac{\theta_1^{\prime}(v_2|\tau)}{\theta_1(v_2|\tau)}$, which transforms as follows under~\eqref{mon 2}, \eqref{mon 3}~\cite{E.T.WhittakerandG.N.Watson}
\begin{gather*}
\frac{\theta_1^{\prime}(v_2|\tau)}{\theta_1(v_2|\tau)}\mapsto \frac{\theta_1^{\prime}(v_2|\tau)}{\theta_1(v_2|\tau)}-2\pi {\rm i} n_2, \qquad
\frac{\theta_1^{\prime}(v_2|\tau)}{\theta_1(v_2|\tau)}\mapsto (c\tau+d)\frac{\theta_1^{\prime}(v_2|\tau)}{\theta_1(v_2|\tau)} +2\pi {\rm i} ct^3.
\end{gather*}
The proof is completed when we do the rescaling from $t^1$ to $\frac{t^1}{2\pi {\rm i}}$.
\end{proof}

In order to make the coordinates $\big(t^1,t^2,t^3,t^4\big)$ being well defined, it will be necessary to define them in a suitable covering over $\Omega^{\Ja(\tilde A_1)}/\Ja\big(\tilde A_1\big)$. It is clear that the multivaluedness comes from the coordinates $t^3$, $t^4$ essentially. Therefore, the problem is solved by defining a suitable covering over the orbit space of $\Ja\big(\tilde A_1\big)$. This can be done by fixing a lattice $\big(1,t^4\big)$ and a~representative of orbit given by the action
\begin{gather}\label{translation action of v2}
t^3\mapsto t^3+\mu_2+\lambda_2t^4.
\end{gather}
In order to also realise the coordinates $(u,v_0,v_2,\tau)$ as globally well-behaviour in the covering of~the~orbit space of $\Ja\big(\tilde A_1\big)$, we also forget the $A_1$-action by fixing a representative of each orbit. Therefore, in the following covering the problem
\begin{gather}
\label{covering Jtildea1}
\widetilde{\Omega^{\Ja(\tilde A_1)}/\Ja\big(\tilde A_1\big)}:=\Omega^{\Ja(\tilde A_1)}/\mathbb{Z}\oplus\tau\mathbb{Z},
\end{gather}
where $\mathbb{Z}\oplus\tau\mathbb{Z}$ acts on $\Omega^{\Ja(\tilde A_1)}$ as
\begin{gather*}
v_0 \mapsto v_0+\lambda_0\tau+\mu_0,\qquad
u\mapsto u-2\lambda_0v_0-n_0^2\tau,\qquad
v_2\mapsto v_2,\qquad
\tau\mapsto \tau.
\end{gather*}
 In the covering~\eqref{covering Jtildea1} the coordinates $t^{\alpha}$, and the intersection form $g^*$ are globally single valued. Hence, we have the necessary conditions to have Dubrovin--Frobenius manifold, since its geometry structure should be globally well defined. Note that, $\Omega^{\Ja(\tilde A_1)}/\Ja\big(\tilde A_1\big)$ has the structure of Twisted Frobenius manifold~\cite{B.Dubrovin2}.

\begin{Remark}
 $\big(t^1,t^2\big)$ lives in an enlargement of the algebra of $E_{\bullet,\bullet}[\varphi_0,\varphi_1]$. The extended algebra is the same as $E_{\bullet,\bullet}[\varphi_0,\varphi_1]$, but it is necessary to add the function $\frac{\theta_1^{\prime}(v_2,\tau)}{\theta_1(v_2,\tau)}$ in the ring of~coefficients $E_{\bullet,\bullet}$.
\end{Remark}

\begin{Remark}\label{covering Hurwitz space orbit space jtildea1}
Note that a covering in the orbit space corresponds to a covering in the Hurwitz space. The fixation of a lattice in the orbit space of $\Ja\big(\tilde A_1\big)$ is equivalent to a choice of homology basis in the Hurwitz space $H_{1,0,0}$. Moreover, a choice of the representative of the action~\eqref{translation action of v2} in~the variable $v_2$ is a choice of logarithm root in the Hurwitz space $H_{1,0,0}$. Furthermore, fixing a~representative of the $A_1$-action is to choice a pole or equivalently to choice a sheet in the Hurwitz space $H_{1,0,0}$.
\end{Remark}

\begin{Remark}
The Dubrovin--Frobenius structure in a Hurwitz space is based on an open dense domain of a solution of a Darboux--Egoroff system~\cite{B.Dubrovin2, V.Shramchenko}. Hence, it is a local construction. Indeed, the canonical coordinates associated to the Hurwitz spaces are local coordinates even in the covering space described in Remark~\ref{covering Hurwitz space orbit space jtildea1}. The construction of the orbit space of $\Ja\big(\tilde A_1\big)$ complements the construction of the Hurwitz space $H_{1,0,0}$, because now, there exists global object where the local Dubrovin--Frobenius structure of $H_{1,0,0}$ lives. Indeed, the coordinates $(\varphi_0,\varphi_1,v_2,\tau)$ are global coordinates for the covering space~\eqref{covering Jtildea1}, and from this fact, we derive that the Dubrovin--Frobenius structure is globally well defined in the covering. This is possible because we realise the group $\Ja\big(\tilde A_1\big)$ as a monodromy of orbit space $\Ja\big(\tilde A_1\big)$, and we know how the group $\Ja\big(\tilde A_1\big)$ acts on the Dubrovin--Frobenius structure.
\end{Remark}

\subsection{Construction of WDVV solution}\label{Construction of WDVV solution}

\begin{Theorem}\label{mainresult}
 There exists Dubrovin--Frobenius structure on the manifold $\widetilde{\Omega/\Ja\big(\tilde A_1\big)}$ with the intersection form~\eqref{metrich100}, the Euler vector field~\eqref{eulerh100}, and the unity vector field~\eqref{unith100}. Moreover, $\widetilde{\Omega/\Ja\big(\tilde A_1\big)}$ is isomorphic as Dubrovin--Frobenius manifold to $\widetilde H_{1,0,0}$.
\end{Theorem}

\begin{proof}
The first step to be done is the computation of the intersection form in coordinates $\big(t^1,t^2,t^3,t^4\big)$. Hence, consider the transformation formula of~${\rm d}s^2$:
\begin{gather}\label{metriceta}
g^{\alpha\beta}(t)=\frac{\partial t^\alpha}{\partial x^i}\frac{\partial t^\beta}{\partial x^j}g^{ij},
\end{gather}
where $x^1=u$, $x^2=v_0$, $x^3=v_2$, $x^4=\tau$.

 From the expression:
\begin{gather*}
{\rm d}s^2=2{\rm d}v_0^2-2{\rm d}v_2^2+2{\rm d}u {\rm d}\tau=g_{ij}{\rm d}x^i{\rm d}x^j,
\end{gather*}
we have
\[
(g_{ij})=
\begin{bmatrix}
 0 & 0 & 0 &1 \\
 0 & 2 &0 &0 \\
 0& 0 & -2 & 0 \\
 1 & 0 & 0 &0
\end{bmatrix}.
\]
Therefore,
\[
\big(g^{ij}\big)=(g_{ij})^{-1}=
\begin{bmatrix}
 0 & 0 & 0 &1 \\
 0 & \frac{1}{2} &0 &0 \\
 0& 0 & -\frac{1}{2} & 0 \\
 1 & 0 & 0 &0
\end{bmatrix}.
\]
To compute $g^{\alpha\beta}(t)$, let us write $t^{\alpha}$ in terms of $x^i$:
\begin{gather*}
t^4=\tau,\qquad t^3=v_2,\qquad
t^2=-\frac{\theta_1(v_0+v_2,\tau)\theta_1(v_0-v_2,\tau)} {\theta_1(2v_2,\tau)\theta_1^{\prime}(0,\tau)}{\rm e}^{-2\pi {\rm i} u},
\end{gather*}
using the following formulae~\cite{E.T.WhittakerandG.N.Watson}:
\begin{gather*}
\frac{\wp^{\prime}(v_2)}{\wp(v_0)-\wp(v_0)}=\zeta(v_0-v_2,\tau)-\zeta(v_0+v_{2},\tau)+2\zeta(v_2,\tau),
\\
\wp(v_0,\tau)-\wp(v_2,\tau)=-\frac{\sigma(v_0+v_2,\tau)\sigma(v_0-v_2,\tau)} {\sigma^2(v_0,\tau)\sigma^2(v_2,\tau)},
\qquad
\frac{\sigma(2v_2,\tau)}{\sigma^4(v_2,\tau)}=-\wp^{\prime}(v_2,\tau),
\end{gather*}
it is possible to rewrite $t^1$ in a more suitable way
\begin{align*}
t^1&=-t^2[\zeta(v_0-v_2,\tau)-\zeta(v_0+v_{2},\tau)+2\zeta(v_2,\tau)] +2t^2\frac{\theta_1^{\prime}(v_2,\tau)}{\theta_1(v_2,\tau)}
\\
&=-t^2\frac{\wp^{\prime}(v_2,\tau)}{\wp(v_0,\tau)-\wp(v_2,\tau)} +2t^2\frac{\theta_1^{\prime}(v_2,\tau)}{\theta_1(v_2,\tau)}
\\
&=-t^2\frac{\wp^{\prime}(v_2,\tau)\theta_1^2(v_2,\tau)\theta_1^2(v_0,\tau)} {\theta_1(v_0+v_2,\tau)\theta_1(v_0-v_2,\tau){\theta_1^{\prime}(0,\tau)}^2} +2t^2\frac{\theta_1^{\prime}(v_2,\tau)}{\theta_1(v_2,\tau)}
\\
&=-\frac{\wp^{\prime}(v_2,\tau)\theta_1^2(v_2,\tau)\theta_1^2(v_0,\tau)} {\theta_1(2v_2,\tau){\theta_1^{\prime}(0,\tau)}^3}{\rm e}^{-2\pi {\rm i} u}+2t^2\frac{\theta_1^{\prime}(v_2,\tau)}{\theta_1(v_2,\tau)}
=\frac{\theta_1^2(v_0,\tau)}{\theta_1^2(v_2,\tau)}{\rm e}^{-2\pi {\rm i} u}+2t^2\frac{\theta_1^{\prime}(v_2,\tau)}{\theta_1(v_2,\tau)}.
\end{align*}

To summarize
\begin{gather*}
t^1=\frac{\theta_1^2(v_0,\tau)}{\theta_1^2(v_2,\tau)}{\rm e}^{-2\pi {\rm i} u}+2t^2\frac{\theta_1^{\prime}(v_2,\tau)}{\theta_1(v_2,\tau)},\nonumber
\\
t^2=-\frac{\theta_1(v_0+v_2,\tau)\theta_1(v_0-v_2,\tau)} {\theta_1(2v_2,\tau)\theta_1^{\prime}(0,\tau)}{\rm e}^{-2\pi {\rm i} u},
\qquad
t^3=v_2,
\qquad
t^4=\tau.
\end{gather*}
Computing $g^{\alpha\beta}$ according to~\eqref{metriceta}
\begin{gather*}
g^{\alpha\beta}=\frac{1}{2}\frac{\partial t^{\alpha}}{\partial v_0}\frac{\partial t^{\beta}}{\partial v_0}-\frac{1}{2}\frac{\partial t^{\alpha}}{\partial v_2}\frac{\partial t^{\beta}}{\partial v_2}+\frac{\partial t^{\alpha}}{\partial u}\frac{\partial t^{\beta}}{\partial \tau}+\frac{\partial t^{\alpha}}{\partial \tau}\frac{\partial t^{\beta}}{\partial u}.
\end{gather*}
Trivially, we get
\begin{gather*}
g^{44}=g^{34}=0,\qquad
g^{33}=-\frac{1}{2},\qquad
g^{24}=-2\pi {\rm i}t^2,\qquad
g^{14}=-2\pi {\rm i}t^1.
\end{gather*}

The following non-trivial terms are computed in Appendix~\ref{appendixA}:
\begin{gather}
g^{23}=-\frac{t^1}{2}+t^2\frac{\theta_1^{\prime}\big(2t^3,\tau\big)}{\theta_1\big(2t^3,\tau\big)},
\qquad
g^{13}=-2\pi {\rm i} t^2 \frac{\partial}{\partial \tau}\left(\log\frac{\theta_1^{\prime}(0,\tau)}{\theta_1\big(2t^3,\tau\big)}\right),\nonumber
\\
g^{22}=2\big(t^2\big)^2\left[\frac{\theta_1^{\prime\prime}\big(2t^3,\tau\big)} {\theta_1(2t^3,\tau)}-\frac{\theta_1^{\prime^2}\big(2t^3,\tau\big)} {\theta_1^2\big(2t^3,\tau\big)}\right],\nonumber
\\
g^{12}=-2\pi {\rm i}\big(t^2\big)^2\left[\frac{\partial^2}{\partial t^3\partial \tau}\left( \log\left(\frac{\theta_1^{\prime}(0,\tau)}{\theta_1\big(2t^3,\tau\big)}\right)\right)\right],
\label{g23}
\\
g^{11}=-4\big(t^2\big)^2\frac{\theta_1^{\prime}\big(t^3,\tau\big)}{\theta_1\big(t^3,\tau\big)}\frac{\partial} {\partial t^3}\left(\frac{\theta_1^{\prime}\big(t^3,\tau\big)}{\theta_1\big(t^3,\tau\big)}\right)\left[ 2\frac{\theta_1^{\prime}\big(t^3,\tau\big)}{\theta_1\big(t^3,\tau\big)} -2\frac{\theta_1^{\prime}\big(2t^3,\tau\big)}{\theta_1\big(2t^3,\tau\big)}\right]\nonumber
\\ \hphantom{g^{11}=}
{}+8\frac{\theta_1^{\prime^2}\big(t^3,\tau\big)}{\theta_1^2\big(t^3,\tau\big)}\big(t^2\big)^2 \left[\frac{\theta_1^{\prime\prime}\big(2t^3,\tau\big)}{\theta_1\big(2t^3,\tau\big)} -\frac{\theta_1^{\prime^2}\big(2t^3,\tau\big)}{\theta_1^2\big(2t^3,\tau\big)}\right]
-2\big(t^2\big)^2\left[\frac{\partial}{\partial t^3}\left(\frac{\theta_1^{\prime}\big(t^3,\tau\big)}{\theta_1\big(t^3,\tau\big)}\right)\right]^2\nonumber
\\ \hphantom{g^{11}=}
{}-16\pi {\rm i}\big(t^2\big)^2\frac{\theta_1^{\prime}\big(t^3,\tau\big)}{\theta_1\big(t^3,\tau\big)} \frac{\partial}{\partial \tau}\left(\frac{\theta_1^{\prime}\big(t^3,\tau\big)}{\theta_1\big(t^3,\tau\big)}\right).\label{g11}
\end{gather}

Differentiating $g^{\alpha\beta}$ w.r.t.\ $t^1$ we obtain a constant matrix $\eta^{*}$:
\[
\big(\eta^{\alpha\beta}\big)=\frac{\partial}{\partial t^1}\big(g^{\alpha\beta}\big)=
\begin{bmatrix}
 0 & 0 & 0 &-2\pi {\rm i} \\
 0 & 0 &-\frac{1}{2} &0 \\
 0& -\frac{1}{2} &0 & 0 \\
 -2\pi {\rm i} & 0 & 0 &0
\end{bmatrix}.
\]
So $t^1$, $t^2$, $t^3$, $t^4$ are the flat coordinates.

The next step is to calculate the matrix $F^{\alpha\beta}$ using the formula
\begin{gather}\label{Hessian formula}
F^{\alpha\beta}=\frac{g^{\alpha\beta}}{\deg \big( g^{\alpha\beta}\big)}.
\end{gather}
We can compute $\deg \big(g^{\alpha\beta}\big)$ using the fact that we compute $\deg (t^{\alpha})$. Indeed,
\begin{gather*}
E=-\frac{1}{2\pi {\rm i}}\frac{\partial}{\partial u}.
\end{gather*}
This implies that
\begin{gather*}
\deg \big(t^1\big)=\deg \big(t^2\big)=1,
\qquad
\deg \big(t^3\big)=\deg \big(t^4\big)=0.
\end{gather*}
Then, the function $F$ is obtained from the equation
\begin{gather*}
\frac{\partial^2 F}{\partial t^{\alpha}\partial t^{\beta}}=\eta_{\alpha\alpha^{\prime}}\eta_{\beta\beta^{\prime}}F^{\alpha^{\prime}\beta^{\prime}}.
\end{gather*}
Computing
\begin{gather*}
F^{\alpha4}=\frac{g^{\alpha 4}}{\deg \big(g^{\alpha 4}\big)},
\end{gather*}
we derive
\begin{gather*}\label{normalization}
\frac{\partial^3 F}{\partial t^1\partial t^{\alpha}\partial t^{\beta}}=\eta_{\alpha\beta}.
\end{gather*}
Hence,
\begin{gather}\label{free energy incomplete}
F=\frac{\rm i}{4\pi}\big(t^1\big)^2t^4-2t^1t^2t^3+f\big(t^2,t^3,t^4\big).
\end{gather}
Substituting $F^{23}$ and $F^{13}$ in~\eqref{free energy incomplete}
\begin{gather*}
F=\frac{\rm i}{4\pi}\big(t^1\big)^2t^4-2t^1t^2t^3-\big(t^2\big)^2\log \left(\frac{\theta_1^{\prime}\big(0,t^4\big)}{\theta_1\big(2t^3,t^4\big)}\right)+h\big(t^2\big) +A_{\alpha\beta}t^{\alpha}t^{\beta}+C_{\alpha}t^{\alpha}+D,
\end{gather*}
where $A_{\alpha\beta}$, $C_{\alpha}$, $C_{\alpha}$ are constants. Note that $F^{22}$, $F^{12}$ contains the same information, furthermore, there is no information in $F^{33}$, $F^{34}$, $F^{44}$ because
\begin{gather*}
\deg \big(g^{33}\big)=\deg \big(g^{34}\big)=\deg \big(g^{44}\big)=0.
\end{gather*}
However, $h\big(t^2\big)$ can be computed by using $g^{33}$
\begin{gather*}
g^{33}=-\frac{1}{2}=E^{\epsilon}\eta^{3\mu}\eta^{3\lambda}c_{\epsilon\mu\lambda}=\frac{t^2}{4}c_{222}.
\end{gather*}
Using the formula~\eqref{freeenergy}, we have
\begin{gather}\label{finalfreeenergy}
F\big(t^1,t^2,t^3,t^4\big)=\frac{\rm i}{4\pi}\big(t^1\big)^2t^4-2t^1t^2t^3-\big(t^2\big)^2 \log\left(t^2\frac{\theta_1^{\prime}\big(0,t^4\big)}{\theta_1\big(2t^3,t^4\big)}\right).
\end{gather}
The remaining part of proof is to show that the equation~\eqref{finalfreeenergy} satisfies WDDV equations. Let us prove it step by step.

1.\ \textit{Commutative of the algebra.}
Defining the structure constant of the algebra as
\begin{gather*}
c_{\alpha\beta\gamma}(t)=\frac{\partial^3F}{\partial t^{\alpha}\partial t^{\beta}\partial t^{\gamma} },
\end{gather*}
commutative is straightforward.

2. \textit{Normalization.}
Using equation~\eqref{finalfreeenergy}, we obtain
\begin{gather*}
c_{1\alpha\beta}(t)=\frac{\partial^3F}{\partial t^{1}\partial t^{\beta}\partial t^{\gamma} }=\eta_{\alpha\beta}.
\end{gather*}

3.\ \textit{Quasi homogeneity.}
Applying the Euler vector field in the function~\eqref{finalfreeenergy}, we have
\begin{gather*}
E(F)=2F-2t^2.
\end{gather*}

4.\ \textit{Associativity.}
In order to prove that the algebra is associativity, we will first shown that the algebra is semisimple. First of all, note that the multiplication by the Euler vector field is equivalent to the intersection form. Indeed,
\begin{gather}
\label{relation multiplication Euler and intersection form}
E\bullet\partial_{\alpha}=t^{\sigma}c_{\sigma\alpha}^{\beta}\partial_{\beta} =t^{\sigma}\partial_{\sigma}\big(\eta^{\beta\mu}\partial_{\alpha}\partial_{\mu}F \big) \partial_{\beta}
=({\rm d}_{\alpha}-{\rm d}_{\beta})\eta^{\beta\mu}\partial_{\alpha}\partial_{\mu}F \partial_{\beta}=\eta_{\alpha\mu}g^{\mu\beta} \partial_{\beta}.
\end{gather}
Therefore, the multiplication by the Euler vector field is semisimple if the following polynomial
\begin{gather}\label{gather det}
\det\big(\eta_{\alpha\mu}g^{\mu\beta}-u\delta_{\alpha}^{\beta}\big)=0,
\end{gather}
has only simple roots; since $\det(\eta_{\alpha\mu})\neq 0$, the equation~\eqref{gather det} is equivalent to
\begin{gather*}
\det\big(g^{\alpha\beta}-u\eta^{\alpha\beta}\big)=0.
\end{gather*}
Using that $\eta^{\alpha\beta}=\partial_1g^{\alpha\beta}$, we have that
\begin{gather*}
\det\big(g^{\alpha\beta}-u\eta^{\alpha\beta}\big)=\det \big(g^{\alpha\beta}\big(t^1-u,t^2,t^3,t^4\big)\big)=0.
\end{gather*}
So, this is enough to compute $\det g^{\alpha\beta}$. In particular, computing $\det g$ in the coordinates $(\varphi_0,\varphi_1,v_2, \tau)$. Recall that
\begin{gather*}
g^{\alpha\beta}=g\big({\rm d}t^{\alpha},{\rm d}t^{\beta}\big),\qquad\, \text{in coordinates } \big(t^1,t^2,t^3,t^4\big),
\\
g^{lm}=g({{\rm d}v_l,{\rm d}v_m}), \qquad \text{in coordinates } (u,v_0,v_2,\tau),
\\
g^{ij}=g({{\rm d}\varphi_i,{\rm d}\varphi_j}), \qquad\, \text{in coordinates } (\varphi_0,\varphi_1,v_2,\tau).
\end{gather*}
Then,
\begin{gather*}
\det g^{ij}= \det\left(\frac{\partial\varphi_i}{\partial v_l} \right) \det\left(\frac{\partial\varphi_j}{\partial v_m} \right)\det \big(g^{lm}\big).
\end{gather*}

\begin{Remark}The coordinates $(u,v_0,v_2,\tau)$ are defined away from the submanifold defined by $\det g=0$. Therefore, we have to change coordinates to compute the roots of $\det g=0$.
\end{Remark}

Hence, it is enough to compute the $\det\big(\frac{\partial\varphi_i}{\partial v_l}\big)$
\begin{gather}
\det\left(\frac{\partial\varphi_i}{\partial v_l} \right)=
\begin{bmatrix}
 \dfrac{\partial\varphi_0}{\partial v_0} & \dfrac{\partial\varphi_0}{\partial v_2} & \dfrac{\partial\varphi_0}{\partial \tau}&-2\pi {\rm i}\varphi_0 \vspace{1mm}\\
 \dfrac{\partial\varphi_1}{\partial v_0} & \dfrac{\partial\varphi_1}{\partial v_2} & \dfrac{\partial\varphi_1}{\partial \tau} &-2\pi {\rm i}\varphi_1 \\
 0& 1 &0 & 0 \\
 0& 0 & 1&0
\end{bmatrix}\nonumber
\\ \hphantom{\det\left(\frac{\partial\varphi_i}{\partial v_l} \right)}
{}=-2\pi {\rm i} \varphi_0\varphi_1\left[2\frac{\theta_1^{\prime}(v_0)}{\theta_1(v_0)} -\frac{\theta_1^{\prime}(-v_0+v_2)}{\theta_1(-v_0+v_2)} +\frac{\theta_1^{\prime}(v_0+v_2)}{\theta_1(v_0+v_2)} \right]\nonumber
\\ \hphantom{\det\left(\frac{\partial\varphi_i}{\partial v_l} \right)}
= -2\pi {\rm i} {\rm e}^{-4\pi {\rm i} u}\frac{\theta_1(2v_0)}{\theta_1(2v_2){\theta_1^{\prime}(0)}^2}.\label{roots of det g}
\end{gather}
Then, equation~\eqref{roots of det g} has four distinct roots $v_0=0,\frac{1}{2},\frac{\tau}{2},\frac{1+\tau}{2}$. Hence, the following system of equations
\begin{gather}\label{system of canonical coordinates}
\det \big( g^{\alpha\beta}\big(t^1,t^2,t^3,t^4\big)\big)=0,\qquad
\det \big(\eta^{\alpha\beta}\big(t^1,t^2,t^3,t^4\big)\big)\neq 0,
\end{gather}
implies in existence of four functions $y_i\big(t^2,t^3,t^4\big)$ such that
\begin{gather*}
t^1=y_i\big(t^2,t^3,t^4\big), \qquad i=1,2,3,4.
\end{gather*}
Sending $t^1\mapsto t^1-u$ in~\eqref{system of canonical coordinates}, we obtain
\begin{gather*}
u^i=t^1-y^i\big(t^2,t^3,t^4\big), \qquad i=1,2,3,4.
\end{gather*}
The multiplication by the Euler vector field
\begin{gather*}
g^{i}_j=\eta_{jk}g^{ki},\qquad \text{in canonical coordinates } \big(u^1,u^2,u^3,u^4\big)
\end{gather*}
 is diagonal, so
\begin{gather*}
g^{ij}=u^i\eta^{ij}\delta_{ij},
\end{gather*}
where $\eta^{ij}$ is the canonical coordinates $\big(u^1,u^2,u^3,u^4\big)$,
 and the unit vector field have the following form:
\begin{gather*}
\frac{\partial}{\partial t^1}=\sum_{i=1}^4 \frac{\partial u_i}{\partial t^1}\frac{\partial}{\partial u_i}=\sum_{i=1}^4 \frac{\partial}{\partial u_i}.
\end{gather*}
Moreover, since
\begin{gather*}
[E,e]=\left[t^1\frac{\partial}{\partial t^1} +t^2\frac{\partial}{\partial t^2},\frac{\partial}{\partial t^1}\right]=-e,
\end{gather*}
the Euler vector field in the coordinates $\big(u^1,u^2,u^3,u^4\big)$ takes the following form:
\begin{gather*}
E=\sum_{i=1}^4 u^i\frac{\partial}{\partial u_i}.
\end{gather*}
 Using the relationship~\eqref{relation multiplication Euler and intersection form} between the coordinates $\big(u^1,u^2,u^3,u^4\big)$, we have
 \begin{gather}\label{final equation associativite}
u^i\eta^{ij}\delta_{ij}=u^l\eta^{im}\eta^{jn}c_{lmn},
\end{gather}
differentiating both side of the equation~\eqref{final equation associativite} with respect $t^1$
 \begin{gather*}
 c_{ij}^k=\delta_{ij},
\end{gather*}
which proves that the algebra is associative and semisimple.

Therefore, we proved that the equation~\eqref{finalfreeenergy} satisfies the WDVV equation. Moreover, the function $F$~\eqref{finalfreeenergy} is exactly the Free energy of the Dubrovin--Frobenius manifold of the Hurwitz space $\widetilde H_{1,0,0}$. Hence, the covering of orbit space of $\Ja\big(\tilde A_1\big)$ and the covering over the Hurwtiz space $H_{1,0,0}$ are isomorphic as a Dubrovin--Frobenius manifold, because they have the same WDVV solution.
\end{proof}

\begin{Remark}
Even thought the Dubrovin--Frobenius structure constructed in a suitable covering of the orbit space of $\Ja\big(\tilde A_1\big)$ is isomorphic as a Dubrovin--Frobenius manifold to a suitable covering of the Hurwitz space $H_{1,0,0}$, this fact does not mean that the construction presented in~this paper is equivalent to the Hurwitz space construction, because:
\begin{enumerate}\itemsep=0pt
\item The constructions start with different hypotheses. Indeed, we derive a WDVV solution in~the Hurwitz space framework from the data of the Hurwitz space itself, and through the choice of a~suitable primary differential; see~\cite{B.Dubrovin2} and~\cite{V.Shramchenko} for the definition. On another hand, the~orbit space construction is derived from the data of the group $\Ja\big(\tilde A_1\big)$.
\item The Hurwitz space construction is based on domain of a solution of a Darboux--Egorrof system. The coordinate system associated with this system of equation is called canonical coordinates. Therefore, the Hurwitz space construction is a local construction, since it is based in a local solution of a system of equations. The orbit space construction, in the other hand, is built based on the invariant coordinates $(\varphi_0,\varphi_1, v_2, \tau)$, which some how have a global meaning. Furthermore, note that the existence of the invariant coordinates $(\varphi_0,\varphi_1, v_2, \tau)$ is not guaranteed in the Hurwitz space construction.
\item The intersection form~\eqref{metrich100}, the Euler vector field~\eqref{eulerh100}, and the unity vector field~\eqref{unith100} are intrinsic objects of the orbit space $\Ja\big(\tilde A_1\big)$. Therefore, the Theorem \ref{mainresult} derives the WDVV solution~\eqref{finalfreeenergy} by using the equation~\eqref{Hessian formula} without using the correspondence with the Hurwtiz space $H_{1,0,0}$. This argument was already used in the introduction of~\cite{B.Dubrovin1} to demonstrate the difference between the Hurwitz space construction on the $H_{0,n}$ and the orbit space construction of the orbit space of $A_n$.
\end{enumerate}
\end{Remark}

\begin{Remark}
The WDVV solution~\eqref{finalfreeenergy} was presented on p.~28 of~\cite{Dubrovin3}. However, there is a~typo in the last term of the WDVV solution in the paper~\cite{Dubrovin3}. The WDVV solution in a~correct form can also be found in~\cite{M.CutimancoandV.Shramchenko} and~\cite{E.V.FerapontovM.V.PavlovL.Xue}.
\end{Remark}

\section{Conclusion}
The WDVV solution of $H_{1,0,0}$, which is~\eqref{finalfreeenergy}, contains the term $\log\big(\frac{\theta_1^{\prime}(0,t^4)}{\theta_1(2t^3,t^4)}\big)$ on the two exceptional variables $\big(t^3,t^4\big)$. This is a reflection of how the ring of invariants affects the WDVV solution. The same pattern is obtained in $\Ja(A_1)$, and $\tilde A_1$. The equation~\eqref{WDVV bertola} contains $E_2(\tau)$ which is a quasi modular form, and the equation~\eqref{WDVV Dub-Zhang} contains ${\rm e}^{t^2}$. These facts could be useful in regards to the understanding of the WDVV/groups correspondence.

The arrows of the diagram of in Section~\ref{group def} may have a third meaning, which is an embedding of Dubrovin--Frobenius submanifolds~\cite{Strachan1, Strachan2} in to the ambient space $H_{1,0,0}$. The fact that $H_{1,0,0}$ contains three Dubrovin--Frobenius submanifolds is not an accident. This comes from the tri-Hamiltonian structure that $H_{1,0,0}$ has~\cite{Pavlov2, Romano1}. In a subsequent publication, we will study the Dubrovin--Frobenius manifolds of~$H_{1,0,0}$, and its associated integrable systems.

\appendix

\section{Appendix}\label{appendixA}
\textbf{Computing $\boldsymbol {g^{12}}$:}
\begin{gather*}
g^{23}=-\frac{1}{2}\frac{\partial t^{2}}{\partial v_2}=-\frac{t^2}{2}\left[-\frac{\theta_1^{\prime}(v_0-v_2,\tau)}{\theta_1(v_0-v_2,\tau)} +\frac{\theta_1^{\prime}(v_0+v_2,\tau)}{\theta_1(v_0+v_2,\tau)} -2\frac{\theta_1^{\prime}(2v_2,\tau)}{\theta_1(2v_2,\tau)}\right]
\\ \hphantom{g^{23}}
{}=-\frac{t^2}{2}\left[-\frac{\theta_1^{\prime}(v_0-v_2,\tau)}{\theta_1(v_0-v_2,\tau)} +\frac{\theta_1^{\prime}(v_0+v_2,\tau)}{\theta_1(v_0+v_2,\tau)} -2\frac{\theta_1^{\prime}(v_2,\tau)}{\theta_1(v_2,\tau)}\right] -t^2\frac{\theta_1^{\prime}(v_2,\tau)}{\theta_1(v_2,\tau)}
+t^2\frac{\theta_1^{\prime}(2v_2,\tau)}{\theta_1(2v_2,\tau)}
\\ \hphantom{g^{23}}
{}=-\frac{1}{2\wp^{\prime}(v_2,\tau)}\left[-\zeta(v_0-v_2,\tau)+\zeta(v_0+v_2,\tau) -2\zeta(v_2,\tau)\right]-t^2\frac{\theta_1^{\prime}(v_2,\tau)}{\theta_1(v_2,\tau)}
+t^2\frac{\theta_1^{\prime}(2v_2,\tau)}{\theta_1(2v_2,\tau)}
\\ \hphantom{g^{23}}
{}=\frac{1}{2}\frac{1}{\wp(z_0,\tau)-\wp(z_2,\tau)} -t^2\frac{\theta_1^{\prime}(v_2,\tau)}{\theta_1(v_2,\tau)} +t^2\frac{\theta_1^{\prime}(2v_2,\tau)}{\theta_1(2v_2,\tau)}
=-\frac{t^1}{2}+t^2\frac{\theta_1^{\prime}(2v_2,\tau)}{\theta_1(2v_2,\tau)}.
\end{gather*}
\textbf{Computing $\boldsymbol {g^{13}}$:}
\begin{gather}
g^{13}=-\frac{1}{2}\frac{\partial t^{1}}{\partial v_2}=-\frac{\theta_1^{\prime}(v_2,\tau)} {\theta_1(v_2,\tau)}\frac{1}{\wp(z_0) -\wp(z_2)}-\frac{\partial t^2}{\partial v_2} \frac{\theta_1^{\prime}(v_2,\tau)}{\theta_1(v_2,\tau)} -t^2\left[\frac{\theta_1^{\prime\prime}(v,\tau)}{\theta_1(v,\tau)} -\frac{\theta_1^{\prime^2}(v,\tau)}{\theta_1^2(v,\tau)}\right]\nonumber
\\ \hphantom{g^{13}}
{}=-\frac{\theta_1^{\prime}(v_2,\tau)}{\theta_1(v_2,\tau)}\frac{1}{\wp(z_0)-\wp(z_2)} -t^1\frac{\theta_1^{\prime}(v_2,\tau)}{\theta_1(v_2,\tau)} -2t^2\frac{\theta_1^{\prime}(2v_2,\tau)}{\theta_1(2v_2,\tau)} \frac{\theta_1^{\prime}(v_2,\tau)}{\theta_1(v_2,\tau)}\nonumber
\\ \hphantom{g^{13}=}
-t^2\left[\frac{\theta_1^{\prime\prime}(v,\tau)}{\theta_1(v,\tau)} -\frac{\theta_1^{\prime^2}(v,\tau)}{\theta_1^2(v,\tau)}\right]\nonumber
\\ \hphantom{g^{13}}
=-2t^2\frac{\theta_1^{\prime^2}(v_2,\tau)}{\theta_1^2(v_2,\tau)}-2t^2\frac{\theta_1^
{\prime}(2v_2,\tau)}{\theta_1(2v_2,\tau)}\frac{\theta_1^{\prime}(v_2,\tau)}{\theta_1(v_2,\tau)} -t^2\left[\frac{\theta_1^{\prime\prime}(v,\tau)}{\theta_1(v,\tau)} -\frac{\theta_1^{\prime^2}(v,\tau)}{\theta_1^2(v,\tau)}\right]\nonumber
\\ \hphantom{g^{13}}
=-t^2\frac{\theta_1^{\prime^2}(v_2,\tau)}{\theta_1^2(v_2,\tau)}
-2t^2\frac{\theta_1^{\prime}(2v_2,\tau)}{\theta_1(2v_2,\tau)}
\frac{\theta_1^{\prime}(v_2,\tau)}{\theta_1(v_2,\tau)}
-t^2\frac{\theta_1^{\prime\prime}(v,\tau)}{\theta_1(v,\tau)}.\label{g131}
\end{gather}

To simplify this expression we need the following lemma.

\begin{Lemma}[\cite{BertolaM.1}]\label{lemmaa1l}
When $x+y+z=0$ holds
\begin{gather}
\frac{\theta_1^{\prime\prime}(x,\tau)}{\theta_1(x,\tau)}
+\frac{\theta_1^{\prime\prime}(y,\tau)}{\theta_1(y,\tau)}
-2\frac{\theta_1^{\prime}(x,\tau)}{\theta_1(x,\tau)}
\frac{\theta_1^{\prime}(y,\tau)}{\theta_1(y,\tau)}\nonumber
\\ \qquad
=4\pi {\rm i} \frac{\partial}{\partial \tau}\left(\log\left(\frac{\theta_1^{\prime}(0,\tau)}{\theta(x-y,\tau)}\right)\right) +2\frac{\theta_1^{\prime}(x-y,\tau)}{\theta_1(x-y,\tau)} \left[\frac{\theta_1^{\prime}(x,\tau)}{\theta_1(x,\tau)} -\frac{\theta_1^{\prime}(y,\tau)}{\theta_1(y,\tau)}\right].
\label{lemmaa1}
\end{gather}
\end{Lemma}

\begin{proof}
Applying the formulas\vspace{-.5ex}
\begin{gather*}
\zeta(v,\tau)=\frac{\theta_1^{\prime}(v,\tau)}{\theta_1(v,\tau)}+4\pi {\rm i}g_1(\tau)v,\qquad
\wp(v,\tau)=-\frac{\theta_1^{\prime\prime}(v,\tau)}{\theta_1(v,\tau)} +{\left(\frac{\theta_1^{\prime}(v,\tau)}{\theta_1(v,\tau)}\right)}^2-4\pi {\rm i}g_1(\tau),
\end{gather*}
in the identity~\cite{E.T.WhittakerandG.N.Watson}\vspace{-.5ex}
\begin{gather*}
\left[\zeta(x)+\zeta(y)+\zeta(z)\right]^2=\wp(x)+\wp(y)+\wp(z),
\end{gather*}
we get
\begin{gather*}
\left(\frac{\theta_1^{\prime}(x,\tau)}{\theta_1(x,\tau)}+\frac{\theta_1^{\prime}(y,\tau)}{\theta_1(y,\tau)} +\frac{\theta_1^{\prime}(z,\tau)}{\theta_1(z,\tau)}\right)^2
\\[-.5ex] \qquad
{}=-12\pi {\rm i}g_1(\tau)-\frac{\theta_1^{\prime\prime}(x,\tau)}{\theta_1(x,\tau)} +\frac{\theta_1^{\prime^2}(x,\tau)}{\theta_1^2(x,\tau)} -\frac{\theta_1^{\prime\prime}(y,\tau)}{\theta_1(y,\tau)}
+\frac{\theta_1^{\prime^2}(y,\tau)}{\theta_1^2(y,\tau)} -\frac{\theta_1^{\prime\prime}(z,\tau)}{\theta_1(z,\tau)} +\frac{\theta_1^{\prime^2}(z,\tau)}{\theta_1^2(z,\tau)}.
\end{gather*}
Simplifying
\begin{gather*}
2\frac{\theta_1^{\prime}(x-y,\tau)}{\theta_1(x-y,\tau)} \left[\frac{\theta_1^{\prime}(x,\tau)}{\theta_1(x,\tau)} -\frac{\theta_1^{\prime}(y,\tau)}{\theta_1(y,\tau)}\right] +2\frac{\theta_1^{\prime}(x,\tau)}{\theta_1(x,\tau)} \frac{\theta_1^{\prime}(y,\tau)}{\theta_1(y,\tau)}
\\[-.5ex] \qquad
{}=3\frac{\eta}{\omega}-\frac{\theta_1^{\prime\prime}(x,\tau)}{\theta_1(x,\tau)} -\frac{\theta_1^{\prime\prime}(y,\tau)}{\theta_1(y,\tau)} -\frac{\theta_1^{\prime\prime}(z,\tau)}{\theta_1(z,\tau)},
\end{gather*}
using the fact that
\begin{gather}\label{difeqtheta}
4\pi {\rm i}\frac{\partial_{\tau}\theta_1^{\prime}(0,\tau)}{\theta_1^{\prime}(0,\tau)}=-12\pi {\rm i}g_1(\tau),
\qquad
\frac{\partial^2}{\partial^2 v}\theta_1(v,\tau)=4\pi {\rm i}\frac{\partial}{\partial \tau}\theta_1(v,\tau),
\end{gather}
and doing the substitution $y\mapsto -y$, $z\mapsto x-y$, we get the desired identity.
\end{proof}

Substituting in the lemma $x=v_2$, $y=-v_2$ we get
\begin{gather}
\label{idtheta1a}
2\frac{\theta_1^{\prime\prime}(v_2,\tau)}{\theta_1(v_2,\tau)} +2\frac{\theta_1^{\prime^2}(v_2,\tau)}{\theta_1^2(v_2,\tau)}
=4\pi {\rm i} \frac{\partial}{\partial \tau}\left(\log\left(\frac{\theta_1^{\prime}(0,\tau)}{\theta_1(2v_2,\tau)}\right)\right) +4\frac{\theta_1^{\prime}(2v_2,\tau)}{\theta_1(2v_2,\tau)} \frac{\theta_1^{\prime}(v_2,\tau)}{\theta_1(v_2,\tau)}.
\end{gather}
Substituting~\eqref{idtheta1a} in~\eqref{g131}
\begin{gather*}
g^{13}=-2\pi {\rm i} t^2 \frac{\partial}{\partial \tau}\left(\log\left(\frac{\theta_1^{\prime}(0,\tau)}{\theta_1(2v_2,\tau)}\right)\right).
\end{gather*}
\textbf{Computing $\boldsymbol {g^{22}}$:}
\begin{gather*}
g^{22}=\frac{1}{2}\left(\frac{\partial t^2}{\partial v_0}\right)^2-\frac{1}{2}\left(\frac{\partial t^2}{\partial v_2}\right)^2+2\frac{\partial t^2}{\partial u}\frac{\partial t^2}{\partial \tau}=\frac{1}{2}\left(\frac{\partial t^2}{\partial v_0}\right)^2-\frac{1}{2}\left(\frac{\partial t^2}{\partial v_2}\right)^2-4\pi {\rm i}t^2\frac{\partial t^2}{\partial \tau}.
\end{gather*}
First, we separately compute $\frac{\partial t^2}{\partial v_2}$, $\frac{\partial t^2}{\partial v_0}$, $\frac{\partial t^2}{\partial \tau}$
\begin{gather*}
\frac{1}{2}\left(\frac{\partial t^2}{\partial v_0}\right)^2=\frac{\big(t^2\big)^2}{2}\left[\frac{\theta_1^{\prime}(v_0+v_2,\tau)}{\theta_1(v_0+v_2,\tau)} +\frac{\theta_1^{\prime}(v_0-v_2,\tau)}{\theta_1(v_0-v_2,\tau)}\right]^2,
\\[-.5ex]
-\frac{1}{2}\left(\frac{\partial t^2}{\partial v_2}\right)^2 =-\frac{\big(t^2\big)^2}{2}\left[-\frac{\theta_1^{\prime}(v_0-v_2,\tau)}{\theta_1(v_0-v_2,\tau)} +\frac{\theta_1^{\prime}(v_0+v_2,\tau)}{\theta_1(v_0+v_2,\tau)} -2\frac{\theta_1^{\prime}(2v_2,\tau)}{\theta_1(2v_2,\tau)}\right]^2,
\\[-.5ex]
-4\pi {\rm i} t^2\frac{\partial t^2}{\partial \tau}=-4\pi {\rm i} \frac{\big(t^2\big)^2}{2}\left[\frac{\partial_{\tau}\theta_1(v_0+v_2,\tau)}{\theta_1(v_0+v_2,\tau)} +\frac{\partial_{\tau}\theta_1(v_0-v_2,\tau)}{\theta_1(v_0-v_2,\tau)} -\frac{\partial_{\tau}\theta_1(2v_2,\tau)}{\theta_1(2v_2,\tau)}\right]
\\ \hphantom{-4\pi {\rm i} t^2\frac{\partial t^2}{\partial \tau}=}
{}-4\pi {\rm i} \frac{\big(t^2\big)^2}{2}\left[-\frac{\partial_{\tau}\theta_1^{\prime}(0,\tau)}{\theta_1^{\prime}(0,\tau)}\right].
\end{gather*}
Summing the equations we get
\begin{gather*}
g^{22}=\frac{\big(t^2\big)^2}{2}\left[4\frac{\theta_1^{\prime}(v_0+v_2,\tau)}{\theta_1(v_0+v_2,\tau)} \frac{\theta_1^{\prime}(v_0-v_2,\tau)}{\theta_1(v_0-v_2,\tau)}\right]
\\[.3ex] \hphantom{g^{22}=}
{}+\frac{\big(t^2\big)^2}{2}\left[4\frac{\theta_1^{\prime}(2v_2,\tau)}{\theta_1(2v_2,\tau)} \left[-\frac{\theta_1^{\prime}(v_0-v_2,\tau)}{\theta_1(v_0-v_2,\tau)} +\frac{\theta_1^{\prime}(v_0+v_2,\tau)}{\theta_1(v_0+v_2,\tau)}\right] -4\frac{\theta_1^{\prime^2}(2v_2,\tau)}{\theta_1^2(2v_2,\tau)}\right]
\\[.3ex] \hphantom{g^{22}=}
{}+\frac{\big(t^2\big)^2}{2}\left[-2\frac{\theta_1^{\prime\prime}(v_0+v_2,\tau)}{\theta_1(v_0+v_2,\tau)} -2\frac{\theta_1^{\prime\prime}(v_0-v_2,\tau)}{\theta_1(v_0-v_2,\tau)}-8\pi {\rm i}\left[-\frac{\partial_{\tau}\theta_1(2v_2,\tau)}{\theta_1(2v_2,\tau)} -\frac{\partial_{\tau}\theta_1^{\prime}(0,\tau)}{\theta_1^{\prime}(0,\tau)}\right]\right],
\end{gather*}
where was used~\eqref{difeqtheta}. Substituting in Lemma~\ref{lemmaa1l} $x=v_0+v_2$, $y=v_0-v_2$ we get
\begin{gather*}
\frac{\theta_1^{\prime\prime}(v_0-v_2,\tau)}{\theta_1(v_0-v_2,\tau)} +\frac{\theta_1^{\prime\prime}(v_0+v_2,\tau)}{\theta_1(v_0+v_2,\tau)} -2\frac{\theta_1^{\prime}(v_0-v_2,\tau)}{\theta_1(v_0-v_2,\tau)} \frac{\theta_1^{\prime}(v_0+v_2,\tau)}{\theta_1(v_0+v_2,\tau)}
\\[.3ex] \qquad
{}=4\pi {\rm i} \frac{\partial}{\partial \tau}\left(\log\left(\frac{\theta_1^{\prime}(0,\tau)}{\theta(2v_2,\tau)}\right)\right) +2\frac{\theta_1^{\prime}(2v_2,\tau)}{\theta_1(2v_2,\tau)} \left[\frac{\theta_1^{\prime}(v_0+v_2,\tau)}{\theta_1
(v_0+v_2,\tau)}-\frac{\theta_1^{\prime}(v_0-v_2,\tau)}{\theta_1(v_0-v_2,\tau)}\right].
\end{gather*}
Substituting the last identity in $g^{22}$ we get
\begin{gather*}
g^{22}=2\big(t^2\big)^2\left[\frac{\theta_1^{\prime\prime}(2v_2,\tau)}{\theta_1(2v_2,\tau)} -\frac{\theta_1^{\prime^2}(2v_2,\tau)}{\theta_1^2(2v_2,\tau)}\right].
\end{gather*}
\textbf{Computing $\boldsymbol {g^{12}}$:}
\begin{gather*}
g^{12}=\frac{1}{2}\frac{\partial t^1}{\partial v_0}\frac{\partial t^2}{\partial v_0}-\frac{1}{2}\frac{\partial t^1}{\partial v_2}\frac{\partial t^2}{\partial v_2}+\frac{\partial t^1}{\partial u}\frac{\partial t^2}{\partial \tau}+\frac{\partial t^2}{\partial u}\frac{\partial t^1}{\partial \tau}
\\ \hphantom{g^{12}}
{}=\frac{1}{2}\frac{\partial t^1}{\partial v_0}\frac{\partial t^2}{\partial v_0}-\frac{1}{2}\frac{\partial t^1}{\partial v_2}\frac{\partial t^2}{\partial v_2}-2\pi {\rm i}t^2\frac{\partial t^1}{\partial \tau}-2\pi {\rm i}t^1\frac{\partial t^2}{\partial \tau}.
\end{gather*}
We have that
\begin{gather*}
\frac{\partial t^1}{\partial v_0}=2\frac{\theta_1^{\prime}(v_0,\tau)}{\theta_1(v_0,\tau)} \frac{\theta_1^2(v_0,\tau)}{\theta_1^2(v_2,\tau)}{\rm e}^{-2\pi {\rm i} u} +2\frac{\partial t^2}{\partial v_0}\frac{\theta_1^{\prime}(v_2,\tau)}{\theta_1(v_2,\tau)},
\\
\frac{\partial t^1}{\partial v_2}=-2\frac{\theta_1^{\prime}(v_2,\tau)}{\theta_1(v_2,\tau)} \frac{\theta_1^2(v_0,\tau)}{\theta_1^2(v_2,\tau)}{\rm e}^{-2\pi {\rm i} u} +2\frac{\partial t^2}{\partial v_2}\frac{\theta_1^{\prime}(v_2,\tau)}{\theta_1(v_2,\tau)} +2t^2\left[\frac{\theta_1^{\prime\prime}(v_2,\tau)}{\theta_1(v_2,\tau)} -\frac{\theta_1^{\prime^2}(v_2,\tau)}{\theta_1^2(v_2,\tau)}\right],
\\
\frac{\partial t^1}{\partial \tau}=2\left[\frac{\partial_{\tau}\theta_1(v_0,\tau)}{\theta_1(v_0,\tau)} -\frac{\partial_{\tau}\theta_1(v_2,\tau)}{\theta_1(v_2,\tau)}\right] \frac{\theta_1^2(v_0,\tau)}{\theta_1^2(v_2,\tau)}{\rm e}^{-2\pi {\rm i} u}\!+2\frac{\partial t^2}{\partial \tau}\frac{\theta_1^{\prime}(v_2,\tau)}{\theta_1(v_2,\tau)}
+2t^2\frac{\partial}{\partial \tau}\left(\frac{\theta_1^{\prime}(v_2,\tau)}{\theta_1(v_2,\tau)}\right).
\end{gather*}
Therefore
\begin{gather*}
\frac{1}{2}\frac{\partial t^1}{\partial v_0}\frac{\partial t^2}{\partial v_0}=t^2\left[\frac{\theta_1^{\prime}(v_0\!+\!v_2,\tau)}{\theta_1(v_0\!+\!v_2,\tau)} +\frac{\theta_1^{\prime}(v_0\!-\!v_2,\tau)}{\theta_1(v_0\!-\!v_2,\tau)}\right] \frac{\theta_1^{\prime}(v_0,\tau)}{\theta_1(v_0,\tau)} \frac{\theta_1^2(v_0,\tau)}{\theta_1^2(v_2,\tau)}{\rm e}^{-2\pi {\rm i} u}\!
+\!\left(\frac{\partial t^2}{\partial v_0}\right)^2\!\frac{\theta_1^{\prime}(v_2,\tau)}{\theta_1(v_2,\tau)},
\\
-\frac{1}{2}\frac{\partial t^1}{\partial v_2}\frac{\partial t^2}{\partial v_2}=-t^2\left[ - \frac{\theta_1^{\prime}(v_0-v_2,\tau)}{\theta_1(v_0-v_2,\tau)}
+\frac{\theta_1^{\prime}(v_0+v_2,\tau)}{\theta_1(v_0+v_2,\tau)}\right] \frac{\theta_1^{\prime}(v_2,\tau)}{\theta_1(v_2,\tau)} \frac{\theta_1^2(v_0,\tau)}{\theta_1^2(v_2,\tau)}{\rm e}^{-2\pi {\rm i} u}
\\ \hphantom{-\frac{1}{2}\frac{\partial t^1}{\partial v_2}\frac{\partial t^2}{\partial v_2}=}
{}-t^2\left[-2\frac{\theta_1^{\prime}(2v_2,\tau)}{\theta_1(2v_2,\tau)}\right] \frac{\theta_1^{\prime}(v_2,\tau)}{\theta_1(v_2,\tau)} \frac{\theta_1^2(v_0,\tau)}{\theta_1^2(v_2,\tau)}{\rm e}^{-2\pi {\rm i} u}
-\left(\frac{\partial t^2}{\partial v_2}\right)^2\frac{\theta_1^{\prime}(v_2,\tau)}{\theta_1(v_2,\tau)}
\\ \hphantom{-\frac{1}{2}\frac{\partial t^1}{\partial v_2}\frac{\partial t^2}{\partial v_2}=-t^2\bigg[}
{}-t^2\frac{\partial t^2}{\partial v_2}\left[\frac{\theta_1^{\prime\prime}(v_2,\tau)}{\theta_1(v_2,\tau)} -\frac{\theta_1^{\prime^2}(v_2,\tau)}{\theta_1^2(v_2,\tau)}\right],
\\
-2\pi {\rm i}t^1\frac{\partial t^2}{\partial \tau}=-2\pi {\rm i}\left[\frac{\partial_{\tau}\theta_1(v_0+v_2,\tau)}{\theta_1(v_0+v_2,\tau)} +\frac{\partial_{\tau}\theta_1(v_0-v_2,\tau)}{\theta_1(v_0-v_2,\tau)}\right]t^2 \frac{\theta_1^2(v_0,\tau)}{\theta_1^2(v_2,\tau)}{\rm e}^{-2\pi {\rm i} u}
\\ \hphantom{-2\pi {\rm i}t^1\frac{\partial t^2}{\partial \tau}=}
{}-2\pi {\rm i}\left[-\frac{\partial_{\tau}\theta_1(2v_2,\tau)}{\theta_1(2v_2,\tau)} -\frac{\partial_{\tau}\theta_1^{\prime}(0,\tau)}{\theta_1^{\prime}(0,\tau)}\right]t^2 \frac{\theta_1^2(v_0,\tau)}{\theta_1^2(v_2,\tau)}{\rm e}^{-2\pi {\rm i} u}
-4\pi {\rm i}t^2\frac{\partial t^2}{\partial \tau}\frac{\theta_1^{\prime}(v_2,\tau)}{\theta_1(v_2,\tau)},
\\
-2\pi {\rm i}t^2\frac{\partial t^1}{\partial \tau}=-4\pi {\rm i}t^2\left[\frac{\partial_{\tau}\theta_1(v_0,\tau)}{\theta_1(v_0,\tau)} -\frac{\partial_{\tau}\theta_1(v_2,\tau)}{\theta_1(v_2,\tau)}\right] \frac{\theta_1^2(v_0,\tau)}{\theta_1^2(v_2,\tau)}{\rm e}^{-2\pi {\rm i} u}
-4\pi {\rm i}t^2\frac{\partial t^2}{\partial \tau}\frac{\theta_1^{\prime}(v_2,\tau)}{\theta_1(v_2,\tau)}
\\ \hphantom{-2\pi {\rm i}t^2\frac{\partial t^1}{\partial \tau}=}
-4\pi {\rm i}\big(t^2\big)^2\frac{\partial}{\partial \tau}\left(\frac{\theta_1^{\prime}(v_2,\tau)}{\theta_1(v_2,\tau)}\right).
\end{gather*}
Let us separate $g^{12}$ in three terms
\begin{gather*}
g^{12}=(1)+(2)+(3),
\end{gather*}
where
\begin{gather*}
(1)=t^2\frac{\theta_1^2(v_0,\tau)}{\theta_1^2(v_2,\tau)}{\rm e}^{-2\pi {\rm i} u}\left[\frac{\theta_1^{\prime}(v_0,\tau)}{\theta_1(v_0,\tau)} \left(\frac{\theta_1^{\prime}(v_0+v_2,\tau)}{\theta_1(v_0+v_2,\tau)} +\frac{\theta_1^{\prime}(v_0-v_2,\tau)}{\theta_1(v_0-v_2,\tau)}\right)\right]
\\ \hphantom{(1)=}
{}+t^2\frac{\theta_1^2(v_0,\tau)}{\theta_1^2(v_2,\tau)}{\rm e}^{-2\pi {\rm i} u}\left[\frac{\theta_1^{\prime}(v_2,\tau)}{\theta_1(v_2,\tau)} \left( - \frac{\theta_1^{\prime}(v_0-v_2,\tau)}{\theta_1(v_0-v_2,\tau)} +\frac{\theta_1^{\prime}(v_0+v_2,\tau)}{\theta_1(v_0+v_2,\tau)} -2\frac{\theta_1^{\prime}(2v_2,\tau)}{\theta_1(2v_2,\tau)}\right)\right]
\\ \hphantom{(1)=}
{}+t^2\frac{\theta_1^2(v_0,\tau)}{\theta_1^2(v_2,\tau)}{\rm e}^{-2\pi {\rm i} u}\left[ - 2\pi {\rm i}\left(\frac{\partial_{\tau}\theta_1(v_0+v_2,\tau)}{\theta_1(v_0+v_2,\tau)} +\frac{\partial_{\tau}\theta_1(v_0-v_2,\tau)}{\theta_1(v_0-v_2,\tau)}\right)\right]
\\ \hphantom{(1)=}
{}+t^2\frac{\theta_1^2(v_0,\tau)}{\theta_1^2(v_2,\tau)}{\rm e}^{-2\pi {\rm i} u}\left[-2\pi {\rm i}\left(-\frac{\partial_{\tau}\theta_1(2v_2,\tau)}{\theta_1(2v_2,\tau)} -\frac{\partial_{\tau}\theta_1^{\prime}(0,\tau)}{\theta_1^{\prime}(0,\tau)}\right)\right]
\\ \hphantom{(1)=}
{}+t^2\frac{\theta_1^2(v_0,\tau)}{\theta_1^2(v_2,\tau)}{\rm e}^{-2\pi {\rm i} u}\left[ - 4\pi {\rm i}\left(\frac{\partial_{\tau}\theta_1(v_0,\tau)}{\theta_1(v_0,\tau)} -\frac{\partial_{\tau}\theta_1(v_2,\tau)}{\theta_1(v_2,\tau)}\right)\right],
\\
(2)=\frac{\theta_1^{\prime}(v_2,\tau)}{\theta_1(v_2,\tau)}
\left[\left(\frac{\partial t^2}{\partial v_0}\right)^2
 - \left(\frac{\partial t^2}{\partial v_2}\right)^2 \!\!\!-8\pi {\rm i}t^2\frac{\partial t^2}{\partial \tau}\right]\\
 \hphantom{(2)}{}
=4\frac{\theta_1^{\prime}(v_2,\tau)}{\theta_1(v_2,\tau)}\big(t^2\big)^2
\left[\frac{\theta_1^{\prime\prime}(2v_2,\tau)}{\theta_1(2v_2,\tau)} -\frac{\theta_1^{\prime^2}(2v_2,\tau)}{\theta_1^2(2v_2,\tau)}\right],
\end{gather*}
where was used the previous computation of $g^{22}$
\begin{gather*}
(3)=-4\pi {\rm i}\big(t^2\big)^2\frac{\partial}{\partial \tau}\left(\frac{\theta_1^{\prime}(v_2,\tau)}{\theta_1(v_2,\tau)}\right)
-t^2\frac{\partial t^2}{\partial v_2}\left[\frac{\theta_1^{\prime\prime}(v_2,\tau)}{\theta_1(v_2,\tau)} -\frac{\theta_1^{\prime^2}(v_2,\tau)}{\theta_1^2(v_2,\tau)}\right].
\end{gather*}
To simplify the expression $(1)$ we need to use the Lemma~\ref{lemmaa1l} with the following substitutions $x=v_0$, $y=v_2$
\begin{gather}
\frac{\theta_1^{\prime\prime}(v_0,\tau)}{\theta_1(v_0,\tau)} +\frac{\theta_1^{\prime\prime}(v_2,\tau)}{\theta_1(v_2,\tau)} -2\frac{\theta_1^{\prime}(v_0,\tau)}{\theta_1(v_0,\tau)} \frac{\theta_1^{\prime}(v_2,\tau)}{\theta_1(v_2,\tau)}\nonumber
\\ \qquad
{}=4\pi {\rm i} \frac{\partial}{\partial \tau}\left(\log\left(\frac{\theta_1^{\prime}(0,\tau)}{\theta(v_0-v_2,\tau)}\right)\right)
+2\frac{\theta_1^{\prime}(v_0-v_2,\tau)}{\theta_1(v_0-v_2,\tau)} \left[\frac{\theta_1^{\prime}(v_0,\tau)}{\theta_1(v_0,\tau)} -\frac{\theta_1^{\prime}(v_2,\tau)}{\theta_1(v_2,\tau)}\right].
\label{eq theta id g12 1}
\end{gather}
Using the substitutions $x=v_0$, $y=-v_2$
\begin{gather}
\frac{\theta_1^{\prime\prime}(v_0,\tau)}{\theta_1(v_0,\tau)} +\frac{\theta_1^{\prime\prime}(v_2,\tau)}{\theta_1(v_2,\tau)} +2\frac{\theta_1^{\prime}(v_0,\tau)}{\theta_1(v_0,\tau)} \frac{\theta_1^{\prime}(v_2,\tau)}{\theta_1(v_2,\tau)}\nonumber
\\ \qquad
{}=4\pi {\rm i} \frac{\partial}{\partial \tau}\left(\log\left(\frac{\theta_1^{\prime}(0,\tau)}{\theta(v_0+v_2,\tau)}\right)\right) +2\frac{\theta_1^{\prime}(v_0+v_2,\tau)}{\theta_1(v_0+v_2,\tau)} \left[\frac{\theta_1^{\prime}(v_0,\tau)}{\theta_1(v_0,\tau)} +\frac{\theta_1^{\prime}(v_2,\tau)}{\theta_1(v_2,\tau)}\right].
\label{eq theta id g12 2}
\end{gather}
Summing~\eqref{eq theta id g12 1} with~\eqref{eq theta id g12 2}
\begin{gather*}
2\frac{\theta_1^{\prime\prime}(v_0,\tau)}{\theta_1(v_0,\tau)} +2\frac{\theta_1^{\prime\prime}(v_2,\tau)}{\theta_1(v_2,\tau)} -4\pi {\rm i} \frac{\partial}{\partial \tau}\left( \log\left(\frac{\theta_1^{\prime}(0,\tau)}{\theta(v_0-v_2,\tau)}\right)\right)
-4\pi {\rm i} \frac{\partial}{\partial \tau}\left( \log\left(\frac{\theta_1^{\prime}(0,\tau)}{\theta(v_0+v_2,\tau)}\right)\right)
\\ \qquad
{}=2\frac{\theta_1^{\prime}(v_0,\tau)}{\theta_1(v_0,\tau)} \left(\frac{\theta_1^{\prime}(v_0+v_2,\tau)}{\theta_1(v_0+v_2,\tau)} +\frac{\theta_1^{\prime}(v_0-v_2,\tau)}{\theta_1(v_0-v_2,\tau)}\right)
\\ \qquad \phantom{=}
+2\frac{\theta_1^{\prime}(v_2,\tau)}{\theta_1(v_2,\tau)} \left(-\frac{\theta_1^{\prime}(v_0-v_2,\tau)}{\theta_1(v_0-v_2,\tau)} +\frac{\theta_1^{\prime}(v_0+v_2,\tau)}{\theta_1(v_0+v_2,\tau)}\right).
\end{gather*}
Substituting in $(1)$ we get
\begin{gather*}
(1)=t^2\frac{\theta_1^2(v_0,\tau)}{\theta_1^2(v_2,\tau)}{\rm e}^{-2\pi {\rm i} u}\left[ - 2\frac{\theta_1^{\prime}(2v_2,\tau)}{\theta_1(2v_2,\tau)} \frac{\theta_1^{\prime}(v_2,\tau)}{\theta_1(v_2,\tau)}\right]
\\ \hphantom{(1)=}
+t^2\frac{\theta_1^2(v_0,\tau)}{\theta_1^2(v_2,\tau)}{\rm e}^{-2\pi {\rm i} u}\left[ - 2\pi {\rm i}\left( - \frac{\partial_{\tau}\theta_1(2v_2,\tau)}{\theta_1(2v_2,\tau)} +\frac{\partial_{\tau}\theta_1^{\prime}(0,\tau)}{\theta_1^{\prime}(0,\tau)}\right)+8\pi {\rm i}\frac{\partial_{\tau}\theta_1(v_2,\tau)}{\theta_1(v_2,\tau)}\right]
\\ \hphantom{(1)}
{}=t^2\frac{\theta_1^2(v_0,\tau)}{\theta_1^2(v_2,\tau)}{\rm e}^{-2\pi {\rm i} u}\left[ - 2\frac{\theta_1^{\prime}(2v_2,\tau)}{\theta_1(2v_2,\tau)} \frac{\theta_1^{\prime}(v_2,\tau)}{\theta_1(v_2,\tau)}-2\pi {\rm i}\frac{\partial}{\partial \tau}\left( \log\frac{\theta_1^{\prime}(0,\tau)}{\theta_1(2v_2,\tau)}\right) +2\frac{\theta_1^{\prime\prime}(v_2,\tau)}{\theta_1(v_2,\tau)}\right].
\end{gather*}
Using the identity~\eqref{lemmaa1}, we get
\begin{gather*}
(1)=t^2\frac{\theta_1^2(v_0,\tau)}{\theta_1^2(v_2,\tau)}{\rm e}^{-2\pi {\rm i} u}\left[\frac{\theta_1^{\prime\prime}(v_2,\tau)}{\theta_1(v_2,\tau)} -\frac{\theta_1^{\prime^2}(v_2,\tau)}{\theta_1^2(v_2,\tau)}\right].
\end{gather*}
We compute $(3)$
\begin{gather*}
(3)=-4\pi {\rm i}\big(t^2\big)^2\frac{\partial}{\partial \tau}\left(\frac{\theta_1^{\prime}(v_2,\tau)}{\theta_1(v_2,\tau)}\right)
-t^2\frac{\partial t^2}{\partial v_2}\left[\frac{\theta_1^{\prime\prime}(v_2,\tau)}{\theta_1(v_2,\tau)} -\frac{\theta_1^{\prime^2}(v_2,\tau)}{\theta_1^2(v_2,\tau)}\right]
\\ \hphantom{(3)}
{}=-4\pi {\rm i}\big(t^2\big)^2\frac{\partial}{\partial \tau}\left(\frac{\theta_1^{\prime}(v_2,\tau)}{\theta_1(v_2,\tau)}\right) -t^2\left(t^1-2t^2\frac{\theta_1^{\prime}(2v_2,\tau)}{\theta_1(2v_2,\tau)}\right) \left[\frac{\theta_1^{\prime\prime}(v_2,\tau)}{\theta_1(v_2,\tau)} -\frac{\theta_1^{\prime^2}(v_2,\tau)}{\theta_1^2(v_2,\tau)}\right]
\\ \hphantom{(3)}
{}=-4\pi {\rm i}\big(t^2\big)^2\frac{\partial}{\partial \tau}\left(\frac{\theta_1^{\prime}(v_2,\tau)}{\theta_1(v_2,\tau)}\right)
+2\big(t^2\big)^2\frac{\theta_1^{\prime}(2v_2,\tau)}{\theta_1(2v_2,\tau)} \left[\frac{\theta_1^{\prime\prime}(v_2,\tau)}{\theta_1(v_2,\tau)} -\frac{\theta_1^{\prime^2}(v_2,\tau)}{\theta_1^2(v_2,\tau)}\right]
\\ \hphantom{(3)=}
{}-t^2\frac{\theta_1^2(v_0,\tau)}{\theta_1^2(v_2,\tau)}{\rm e}^{-2\pi {\rm i} u}\left[\frac{\theta_1^{\prime\prime}(v_2,\tau)}{\theta_1(v_2,\tau)} \!-\!\frac{\theta_1^{\prime^2}(v_2,\tau)}{\theta_1^2(v_2,\tau)}\right]
\!-\!2\big(t^2\big)^2\frac{\theta_1^{\prime}(v_2,\tau)}{\theta_1(v_2,\tau)} \left[\frac{\theta_1^{\prime\prime}(v_2,\tau)}{\theta_1(v_2,\tau)} \!-\!\frac{\theta_1^{\prime^2}(v_2,\tau)}{\theta_1^2(v_2,\tau)}\right].
\end{gather*}
The result implies
\begin{gather*}
(1) + (3) =-4\pi {\rm i}\big(t^2\big)^2\frac{\partial}{\partial \tau}\left(\frac{\theta_1^{\prime}(v_2,\tau)}{\theta_1(v_2,\tau)}\right)\\
\hphantom{(1) + (3) =}{}
 - 2\big(t^2\big)^2\left[\frac{\theta_1^{\prime}(v_2,\tau)}{\theta_1(v_2,\tau)} - \frac{\theta_1^{\prime}(2v_2,\tau)}{\theta_1(2v_2,\tau)}\right] \left[\frac{\theta_1^{\prime\prime}(v_2,\tau)}{\theta_1(v_2,\tau)} - \frac{\theta_1^{\prime^2}(v_2,\tau)}{\theta_1^2(v_2,\tau)}\right] .
\end{gather*}
\textbf{Computing $\boldsymbol {g^{12}}$:}
\begin{gather*}
g^{12}=-4\pi {\rm i}\big(t^2\big)^2\frac{\partial}{\partial \tau}\left(\frac{\theta_1^{\prime}(v_2,\tau)}{\theta_1(v_2,\tau)}\right)
-2\big(t^2\big)^2\left[\frac{\theta_1^{\prime}(v_2,\tau)}{\theta_1(v_2,\tau)} -\frac{\theta_1^{\prime}(2v_2,\tau)}{\theta_1(2v_2,\tau)}\right] \left[\frac{\theta_1^{\prime\prime}(v_2,\tau)}{\theta_1(v_2,\tau)} -\frac{\theta_1^{\prime^2}(v_2,\tau)}{\theta_1^2(v_2,\tau)}\right]
\\ \hphantom{g^{12}=}
{}+4\frac{\theta_1^{\prime}(v_2,\tau)}{\theta_1(v_2,\tau)}\big(t^2\big)^2 \left[\frac{\theta_1^{\prime\prime}(2v_2,\tau)}{\theta_1(2v_2,\tau)} -\frac{\theta_1^{\prime^2}(2v_2,\tau)}{\theta_1^2(2v_2,\tau)}\right].
\end{gather*}
To simplify this expression we need to prove one more lemma.

\begin{Lemma}
\begin{gather}
2\frac{\theta_1^{\prime\prime\prime}(v_2,\tau)}{\theta_1(v_2,\tau)} +2\frac{\theta_1^{\prime\prime}(v_2,\tau)\theta_1^{\prime}(v_2,\tau)}{\theta_1(v_2,\tau)} -4\frac{\theta_1^{\prime^3}(v_2,\tau)}{\theta_1^3(v_2,\tau)}\nonumber
\\ \qquad
{}=4\pi {\rm i} \frac{\partial^2}{\partial v_2\partial \tau}\left( \log\left(\frac{\theta_1^{\prime}(0,\tau)}{\theta_1(2v_2,\tau)}\right)\right)
+8\frac{\theta_1^{\prime\prime}(2v_2,\tau)}{\theta_1(2v_2,\tau)} \frac{\theta_1^{\prime}(v_2,\tau)}{\theta_1(v_2,\tau)} -8\frac{\theta_1^{\prime^2}(2v_2,\tau)}{\theta_1^2(2v_2,\tau)} \frac{\theta_1^{\prime}(v_2,\tau)}{\theta_1(v_2,\tau)}\nonumber
\\ \qquad\hphantom{=}
{}+4\frac{\theta_1^{\prime}(2v_2,\tau)}{\theta_1(2v_2,\tau)} \frac{\theta_1^{\prime\prime}(v_2,\tau)}{\theta_1(v_2,\tau)}
-4\frac{\theta_1^{\prime}(2v_2,\tau)}{\theta_1(2v_2,\tau)} \frac{\theta_1^{\prime^2}(v_2,\tau)}{\theta_1^2(v_2,\tau)}.
\label{lemmaa2}
\end{gather}
\end{Lemma}

\begin{proof}
Differentiating the identity \label{idtheta1} with respect to $v_2$ we obtain~\eqref{lemmaa2}.
\end{proof}
\textbf{Computing $\boldsymbol {g^{12}}$:}
\begin{gather*}
g^{12}=\big(t^2\big)^2\left[-\frac{\theta_1^{\prime\prime\prime}(v_2,\tau)}{\theta_1(v_2,\tau)} +\frac{\theta_1^{\prime}(v_2,\tau)\theta_1^{\prime\prime}(v_2,\tau)}{\theta_1^2(v_2,\tau)}
-2\frac{\theta_1^{\prime}(v_2,\tau)}{\theta_1(v_2,\tau)} \frac{\theta_1^{\prime\prime}(v_2,\tau)}{\theta_1(v_2,\tau)} +2\frac{\theta_1^{\prime^3}(v_2,\tau)}{\theta_1^3(v_2,\tau)}\right]
\\ \hphantom{g^{12}=}
{}+\big(t^2\big)^2\left[2\frac{\theta_1^{\prime}(2v_2,\tau)}{\theta_1(2v_2,\tau)} \frac{\theta_1^{\prime\prime}(v_2,\tau)}{\theta_1(v_2,\tau)} -2\frac{\theta_1^{\prime}(2v_2,\tau)}{\theta_1(2v_2,\tau)} \frac{\theta_1^{\prime^2}(v_2,\tau)}{\theta_1^2(v_2,\tau)} +4\frac{\theta_1^{\prime}(v_2,\tau)}{\theta_1(v_2,\tau)} \frac{\theta_1^{\prime\prime}(2v_2,\tau)}{\theta_1(2v_2,\tau)}\right]
\\ \hphantom{g^{12}=}
{}+\big(t^2\big)^2 \left[ - 4\frac{\theta_1^{\prime}(v_2,\tau)}{\theta_1(v_2,\tau)} \frac{\theta_1^{\prime^2}(2v_2,\tau)}{\theta_1^2(2v_2,\tau)}\right].
\end{gather*}
Applying~\eqref{lemmaa2}, we get
\begin{gather*}
g^{12}=-2\pi {\rm i}\big(t^2\big)^2\left[\frac{\partial^2}{\partial v_2\partial \tau}\left(\!\log\left(\frac{\theta_1^{\prime}(0,\tau)}{\theta_1(2v_2,\tau)}\right)\right)\right].
\end{gather*}
\textbf{Computing $\boldsymbol {g^{11}}$:}
\begin{gather*}
g^{11}=\frac{1}{2}\left(\frac{\partial t^1}{\partial v_0}\right)^2-\frac{1}{2}\left(\frac{\partial t^1}{\partial v_2}\right)^2+2\frac{\partial t^1}{\partial u}\frac{\partial t^1}{\partial \tau}
=\frac{1}{2}\left(\frac{\partial t^1}{\partial v_0}\right)^2-\frac{1}{2}\left(\frac{\partial t^1}{\partial v_2}\right)^2-4\pi {\rm i}t^1\frac{\partial t^1}{\partial \tau}.
\end{gather*}
Computing $\frac{1}{2}\big(\frac{\partial t^1}{\partial v_0}\big)^2$, $\frac{1}{2}\big(\frac{\partial t^1}{\partial v_2}\big)^2$ and $-4\pi {\rm i}t^1\frac{\partial t^1}{\partial \tau}$:

To simplify the computation let us define
\begin{gather*}
A:=\frac{\theta_1^2(v_0,\tau)}{\theta_1^2(v_2,\tau)}{\rm e}^{-2\pi {\rm i} u}.
\end{gather*}
Then,
\begin{gather*}
\frac{1}{2}\left(\frac{\partial t^1}{\partial v_0}\right)^2 =2\frac{\theta_1^{\prime^2}(v_0,\tau)}{\theta_1^2(v_0,\tau)}A^2 +4A\frac{\theta_1^{\prime}(v_0,\tau)}{\theta_1(v_0,\tau)}
\frac{\partial t^2}{\partial v_0}\frac{\theta_1^{\prime}(v_2,\tau)}{\theta_1(v_2,\tau)} +2\left(\frac{\partial t^2}{\partial v_0}\right)^2\frac{\theta_1^{\prime^2}(v_2,\tau)}{\theta_1^2(v_2,\tau)},
\\
-\frac{1}{2}\left(\frac{\partial t^1}{\partial v_2}\right)^2\! =-2\frac{\theta_1^{\prime^2}(v_2\tau)}{\theta_1^2(v_2,\tau)}A^2
+2A\frac{\theta_1^{\prime}(v_2,\tau)}{\theta_1(v_2,\tau)}\! \left[2\frac{\partial t^2}{\partial v_2}\frac{\theta_1^{\prime}(v_2,\tau)}{\theta_1(v_2,\tau)} +2t^2\!\left[\frac{\theta_1^{\prime\prime}(v_2,\tau)}{\theta_1(v_2,\tau)} -\frac{\theta_1^{\prime^2}(v_2,\tau)}{\theta_1^2(v_2,\tau)}\right]\right]
\\ \hphantom{\frac{1}{2}\left(\frac{\partial t^1}{\partial v_0}\right)^2 =}
-2\left(\frac{\partial t^2}{\partial v_2}\right)^2 \frac{\theta_1^{\prime^2}(v_2,\tau)}{\theta_1^2(v_2,\tau)} -4t^2\frac{\partial t^2}{\partial v_2}\frac{\theta_1^{\prime}(v_2,\tau)}{\theta_1(v_2,\tau)} \left[\frac{\theta_1^{\prime\prime}(v_2,\tau)}{\theta_1(v_2,\tau)} -\frac{\theta_1^{\prime^2}(v_2,\tau)}{\theta_1^2(v_2,\tau)}\right]
\\ \hphantom{\frac{1}{2}\left(\frac{\partial t^1}{\partial v_0}\right)^2 =}
-2\big(t^2\big)^2\left[\frac{\theta_1^{\prime\prime}(v_2,\tau)}{\theta_1(v_2,\tau)} -\frac{\theta_1^{\prime^2}(v_2,\tau)}{\theta_1^2(v_2,\tau)}\right]^2,
\\
-4\pi {\rm i}t^1\frac{\partial t^1}{\partial \tau}=-8\pi {\rm i}A^2\left[\frac{\partial_{\tau}\theta_1(v_0,\tau)}{\theta_1(v_0,\tau)} -\frac{\partial_{\tau}\theta_1(v_2,\tau)}{\theta_1(v_2,\tau)}\right]
-8\pi {\rm i}A\frac{\partial t^2}{\partial \tau}\frac{\theta_1^{\prime}(v_2,\tau)}{\theta_1(v_2,\tau)}
\\ \hphantom{-4\pi {\rm i}t^1\frac{\partial t^1}{\partial \tau}=}
{}-8\pi {\rm i}At^2\frac{\partial}{\partial \tau}\left(\frac{\theta_1^{\prime}(v_2,\tau)}{\theta_1(v_2,\tau)}\right)-16\pi {\rm i}A t^2\frac{\theta_1^{\prime}(v_2,\tau)}{\theta_1(v_2,\tau)} \left[\frac{\partial_{\tau}\theta_1(v_0,\tau)}{\theta_1(v_0,\tau)} {}-\frac{\partial_{\tau}\theta_1(v_2,\tau)}{\theta_1(v_2,\tau)}\right]
\\ \hphantom{-4\pi {\rm i}t^1\frac{\partial t^1}{\partial \tau}=}
-16\pi {\rm i}t^2\frac{\partial t^2}{\partial \tau}\frac{\theta_1^{\prime^2}(v_2,\tau)}{\theta_1^2(v_2,\tau)}-16\pi {\rm i}\big(t^2\big)^2\frac{\theta_1^{\prime}(v_2,\tau)}{\theta_1(v_2,\tau)}\frac{\partial}{\partial \tau}\left(\frac{\theta_1^{\prime}(v_2,\tau)}{\theta_1(v_2,\tau)}\right).
\end{gather*}
Then, we have
\begin{gather*}
g^{11}=(1)+(2)+(3)+(4)+(5),
\end{gather*}
where
\begin{gather*}
(1)=A^2\left[ 2\frac{\theta_1^{\prime^2}(v_0,\tau)}{\theta_1^2(v_0,\tau)} -2\frac{\theta_1^{\prime^2}(v_2\tau)}{\theta_1^2(v_2,\tau)} -8\pi {\rm i}\left[\frac{\partial_{\tau}\theta_1(v_0,\tau)}{\theta_1(v_0,\tau)} -\frac{\partial_{\tau}\theta_1(v_2,\tau)}{\theta_1(v_2,\tau)}\right] \right]
\\ \hphantom{(1)}
=A^2\left[ 2\frac{\theta_1^{\prime^2}(v_0,\tau)}{\theta_1^2(v_0,\tau)} -2\frac{\theta_1^{\prime^2}(v_2\tau)}{\theta_1^2(v_2,\tau)} -2\frac{\theta_1^{\prime\prime}(v_0,\tau)}{\theta_1(v_0,\tau)} +2\frac{\theta_1^{\prime\prime}(v_2,\tau)}{\theta_1(v_2,\tau)}\right]
\\\hphantom{(1)}
=2A^2[\wp(v_0)-\wp(v_2)]=2\frac{16\omega^4}{[\wp(v_0)-\wp(v_2)]^2} [\wp(v_0)-\wp(v_2)]=32\frac{\omega^4}{\wp(v_0)-\wp(v_2)},
\\
(2)=-8\pi {\rm i}t^2A\frac{\partial}{\partial \tau}\! \left(\frac{\theta_1^{\prime}(v_2,\tau)}{\theta_1(v_2,\tau)}\right)
\!+2At^2\frac{\theta_1^{\prime^2}(v_2,\tau)}{\theta_1^2(v_2,\tau)} \!\left[2\frac{\theta_1^{\prime}(v_0,\tau)}{\theta_1(v_0,\tau)} \!\left[ \frac{\theta_1^{\prime}(v_0\!-\!v_2,\tau)}{\theta_1(v_0\!-\!v_2,\tau)} \! +\!\frac{\theta_1^{\prime}(v_0\!+\!v_2,\tau)}{\theta_1(v_0\!+\!v_2,\tau)}\right]\right]
\\ \hphantom{(2)=}
+2At^2\frac{\theta_1^{\prime^2}(v_2,\tau)}{\theta_1^2(v_2,\tau)} \left[2\frac{\theta_1^{\prime}(v_2,\tau)}{\theta_1(v_2,\tau)}\left[ \frac{-\theta_1^{\prime}(v_0-v_2,\tau)}{\theta_1(v_0-v_2,\tau)} +\frac{\theta_1^{\prime}(v_0+v_2,\tau)}{\theta_1(v_0+v_2,\tau)} -2\frac{\theta_1^{\prime}(2v_2,\tau)}{\theta_1(2v_2,\tau)}\right]\right]
\\ \hphantom{(2)=}
+2At^2\frac{\theta_1^{\prime^2}(v_2,\tau)}{\theta_1^2(v_2,\tau)} \left[2\left[\frac{\theta_1^{\prime\prime}(v_2,\tau)}{\theta_1(v_2,\tau)} -\frac{\theta_1^{\prime^2}(v_2,\tau)}{\theta_1^2(v_2,\tau)}\right]-8\pi {\rm i}\left[\frac{\partial_{\tau}\theta_1(v_0,\tau)}{\theta_1(v_0,\tau)} -\frac{\partial_{\tau}\theta_1(v_2,\tau)}{\theta_1(v_2,\tau)}\right]\right]
\\ \hphantom{(2)=}
+\!2At^2\frac{\theta_1^{\prime^2}(v_2,\tau)}{\theta_1^2(v_2,\tau)}\!\!\left[\!-4\pi {\rm i}\!\left[\frac{\partial_{\tau}\theta_1(v_0\!+\!v_2,\tau)}{\theta_1(v_0\!+\!v_2,\tau)} \!+\!\frac{\partial_{\tau}\theta_1(v_0\!-\!v_2,\tau)}{\theta_1(v_0\!-\!v_2,\tau)} \!-\!\frac{\partial_{\tau}\theta_1(2v_2,\tau)}{\theta_1(2v_2,\tau)} \!-\!\frac{\partial_{\tau}\theta_1^{\prime}(0,\tau)}{\theta_1^{\prime}(0,\tau)}\!\right]\! \right]\!.
\end{gather*}
Using~\eqref{lemmaa1},
\begin{gather*}
2\frac{\theta_1^{\prime\prime}(v_0,\tau)}{\theta_1(v_0,\tau)} +2\frac{\theta_1^{\prime\prime}(v_2,\tau)}{\theta_1(v_2,\tau)}-4\pi {\rm i} \frac{\partial}{\partial \tau}\left( \log\left(\frac{\theta_1^{\prime}(0,\tau)}{\theta(v_0-v_2,\tau)}\right)\right)-4\pi {\rm i} \frac{\partial}{\partial \tau}\left( \log\left(\frac{\theta_1^{\prime}(0,\tau)} {\theta(v_0+v_2,\tau)}\right)\right)
\\ \qquad
{}=2\frac{\theta_1^{\prime}(v_0,\tau)}{\theta_1(v_0,\tau)} \left(\frac{\theta_1^{\prime}(v_0+v_2,\tau)}{\theta_1(v_0+v_2,\tau)}
+\frac{\theta_1^{\prime}(v_0-v_2,\tau)}{\theta_1(v_0-v_2,\tau)}\right)
\\ \qquad\phantom{=} +2\frac{\theta_1^{\prime}(v_2,\tau)}{\theta_1(v_2,\tau)} \left( - \frac{\theta_1^{\prime}(v_0-v_2,\tau)}{\theta_1(v_0-v_2,\tau)} +\frac{\theta_1^{\prime}(v_0+v_2,\tau)}{\theta_1(v_0+v_2,\tau)}\right),
\\
(2)=-8\pi {\rm i}t^2A\frac{\partial}{\partial \tau}\left(\frac{\theta_1^{\prime}(v_2,\tau)}{\theta_1(v_2,\tau)}\right) +2At^2\frac{\theta_1^{\prime^2}(v_2,\tau)}{\theta_1^2(v_2,\tau)} \left[ - 4\frac{\theta_1^{\prime}(v_2,\tau)}{\theta_1(v_2,\tau)} \frac{\theta_1^{\prime}(2v_2,\tau)}{\theta_1(2v_2,\tau)}\right]
\\ \phantom{(2)=}
+2At^2\frac{\theta_1^{\prime^2}(v_2,\tau)}{\theta_1^2(v_2,\tau)} \left[2\left[\frac{\theta_1^{\prime\prime}(v_2,\tau)}{\theta_1(v_2,\tau)} -\frac{\theta_1^{\prime^2}(v_2,\tau)}{\theta_1^2(v_2,\tau)}\right]+4 \frac{\theta_1^{\prime\prime}(v_2,\tau)}{\theta_1(v_2,\tau)} \right]
\\ \hphantom{(2)=}
+2At^2\frac{\theta_1^{\prime^2}(v_2,\tau)}{\theta_1^2(v_2,\tau)}\left[-4\pi {\rm i}\frac{\partial_{\tau}\theta_1^{\prime}(0,\tau)}{\theta_1^{\prime}(0,\tau)}+4\pi {\rm i}\frac{\partial_{\tau}\theta_1(2v_2,\tau)}{\theta_1(2v_2,\tau)} \right].
\end{gather*}
Using again~\eqref{lemmaa1},
\begin{gather*}
(2)=-8\pi {\rm i}t^2A\frac{\partial}{\partial \tau}\left(\frac{\theta_1^{\prime}(v_2,\tau)}{\theta_1(v_2,\tau)}\right) +8At^2\frac{\theta_1^{\prime^2}(v_2,\tau)}{\theta_1^2(v_2,\tau)} \left[\frac{\theta_1^{\prime\prime}(v_2,\tau)}{\theta_1(v_2,\tau)} -\frac{\theta_1^{\prime^2}(v_2,\tau)}{\theta_1^2(v_2,\tau)}\right],
\\
(3)=4\frac{\theta_1^{\prime^2}(v_2,\tau)}{\theta_1^2(v_2,\tau)}\left[ \frac{1}{2}\left(\frac{\partial t^2}{\partial v_0}\right)^2- \frac{1}{2}\left(\frac{\partial t^2}{\partial v_2}\right)^2-4\pi {\rm i}t^2\frac{\partial t^2}{\partial \tau}\right]
\\ \hphantom{(3)}
=8\frac{\theta_1^{\prime^2}(v_2,\tau)}{\theta_1^2(v_2,\tau)}\big(t^2\big)^2 \left[\frac{\theta_1^{\prime\prime}(2v_2,\tau)}{\theta_1(2v_2,\tau)} -\frac{\theta_1^{\prime^2}(2v_2,\tau)}{\theta_1^2(2v_2,\tau)}\right],
\\
(4)=-2\big(t^2\big)^2\left[\frac{\partial}{\partial v_2}\left(\frac{\theta_1^{\prime}(v_2,\tau)}{\theta_1(v_2,\tau)}\right)\right]^2-16\pi {\rm i}\big(t^2\big)^2\frac{\theta_1^{\prime}(v_2,\tau)}{\theta_1(v_2,\tau)}\frac{\partial}{\partial \tau}\left(\frac{\theta_1^{\prime}(v_2,\tau)}{\theta_1(v_2,\tau)}\right),
\\
(5)=-4\big(t^2\big)\frac{\partial t^2} {\partial v_2}\frac{\theta_1^{\prime}(v_2,\tau)}{\theta_1(v_2,\tau)}\frac{\partial}{\partial v_2}\left(\frac{\theta_1^{\prime}(v_2,\tau)}{\theta_1(v_2,\tau)}\right)
\\ \hphantom{(5)}
=-4\big(t^2\big)^2\frac{\theta_1^{\prime}(v_2,\tau)}{\theta_1(v_2,\tau)}\frac{\partial}{\partial v_2}\left(\frac{\theta_1^{\prime}(v_2,\tau)}{\theta_1(v_2,\tau)}\right)\left[ \frac{-\theta_1^{\prime}(v_0-v_2,\tau)}{\theta_1(v_0-v_2,\tau)} +\frac{\theta_1^{\prime}(v_0+v_2,\tau)}{\theta_1(v_0+v_2,\tau)} -2\frac{\theta_1^{\prime}(2v_2,\tau)}{\theta_1(2v_2,\tau)}\right]
\\ \hphantom{(5)}
=-4\big(t^2\big)^2\frac{\theta_1^{\prime}(v_2,\tau)}{\theta_1(v_2,\tau)}\frac{\partial}{\partial v_2}\left(\frac{\theta_1^{\prime}(v_2,\tau)}{\theta_1(v_2,\tau)}\right)\left[ \frac{-\theta_1^{\prime}(v_0-v_2,\tau)}{\theta_1(v_0-v_2,\tau)} +\frac{\theta_1^{\prime}(v_0+v_2,\tau)}{\theta_1(v_0+v_2,\tau)} -2\frac{\theta_1^{\prime}(v_2,\tau)}{\theta_1(v_2,\tau)}\right]
\\ \hphantom{(5)=}
-4\big(t^2\big)^2\frac{\theta_1^{\prime}(v_2,\tau)}{\theta_1(v_2,\tau)}\frac{\partial}{\partial v_2}\left(\frac{\theta_1^{\prime}(v_2,\tau)}{\theta_1(v_2,\tau)}\right)\left[ 2\frac{\theta_1^{\prime}(v_2,\tau)}{\theta_1(v_2,\tau)} -2\frac{\theta_1^{\prime}(2v_2,\tau)}{\theta_1(2v_2,\tau)}\right]
\\ \hphantom{(5)}
=-4\big(t^2\big)A\frac{\theta_1^{\prime}(v_2,\tau)}{\theta_1(v_2,\tau)}\frac{\partial}{\partial v_2}\left(\frac{\theta_1^{\prime}(v_2,\tau)}{\theta_1(v_2,\tau)}\right)
\\ \hphantom{(5)=}
-4\big(t^2\big)^2\frac{\theta_1^{\prime}(v_2,\tau)}{\theta_1(v_2,\tau)}\frac{\partial}{\partial v_2}\left(\frac{\theta_1^{\prime}(v_2,\tau)}{\theta_1(v_2,\tau)}\right)\left[ 2\frac{\theta_1^{\prime}(v_2,\tau)}{\theta_1(v_2,\tau)} -2\frac{\theta_1^{\prime}(2v_2,\tau)}{\theta_1(2v_2,\tau)}\right].
\end{gather*}
Summing $(2)$ and $(5)$
\begin{gather*}
(2)+(5)=-8\pi {\rm i}t^2A\frac{\partial}{\partial \tau}\left(\frac{\theta_1^{\prime}(v_2,\tau)}{\theta_1(v_2,\tau)}\right) +8At^2\frac{\theta_1^{\prime^2}(v_2,\tau)}{\theta_1^2(v_2,\tau)} \left[\frac{\theta_1^{\prime\prime}(v_2,\tau)}{\theta_1(v_2,\tau)} -\frac{\theta_1^{\prime^2}(v_2,\tau)}{\theta_1^2(v_2,\tau)}\right]
\\ \hphantom{(2)+(5)=}
-4\big(t^2\big)A\frac{\theta_1^{\prime}(v_2,\tau)}{\theta_1(v_2,\tau)}\frac{\partial}{\partial v_2}\left(\frac{\theta_1^{\prime}(v_2,\tau)}{\theta_1(v_2,\tau)}\right)
\\ \hphantom{(2)+(5)=}
-4\big(t^2\big)^2\frac{\theta_1^{\prime}(v_2,\tau)}{\theta_1(v_2,\tau)}\frac{\partial}{\partial v_2}\left(\frac{\theta_1^{\prime}(v_2,\tau)}{\theta_1(v_2,\tau)}\right)\left[ 2\frac{\theta_1^{\prime}(v_2,\tau)}{\theta_1(v_2,\tau)} -2\frac{\theta_1^{\prime}(2v_2,\tau)}{\theta_1(2v_2,\tau)}\right]
\\ \hphantom{(2)+(5)}
=-4\big(t^2\big)^2\frac{\theta_1^{\prime}(v_2,\tau)}{\theta_1(v_2,\tau)}\frac{\partial}{\partial v_2}\left(\frac{\theta_1^{\prime}(v_2,\tau)}{\theta_1(v_2,\tau)}\right)\left[ 2\frac{\theta_1^{\prime}(v_2,\tau)}{\theta_1(v_2,\tau)} -2\frac{\theta_1^{\prime}(2v_2,\tau)}{\theta_1(2v_2,\tau)}\right]
\\ \hphantom{(2)+(5)=}
+At^2\left[{-}\,2\frac{\theta_1^{\prime\prime\prime}(v_2,\tau)}{\theta_1(v_2,\tau)} +6\frac{\theta_1^{\prime}(v_2,\tau)\theta_1^{\prime\prime}(v_2,\tau)}{\theta_1(v_2,\tau)} -4\frac{\theta_1^{\prime^3}(v_2,\tau)}{\theta^3_1(v_2,\tau)}\right]
\\ \hphantom{(2)+(5)}
=-4\big(t^2\big)^2\frac{\theta_1^{\prime}(v_2,\tau)}{\theta_1(v_2,\tau)}\frac{\partial}{\partial v_2}\left(\frac{\theta_1^{\prime}(v_2,\tau)}{\theta_1(v_2,\tau)}\right)\left[ 2\frac{\theta_1^{\prime}(v_2,\tau)}{\theta_1(v_2,\tau)} -2\frac{\theta_1^{\prime}(2v_2,\tau)}{\theta_1(2v_2,\tau)}\right]
+2At^2\wp^{\prime}(v_2)
\\ \hphantom{(2)+(5)}
=-4\big(t^2\big)^2\frac{\theta_1^{\prime}(v_2,\tau)}{\theta_1(v_2,\tau)}\frac{\partial}{\partial v_2}\!\left(\frac{\theta_1^{\prime}(v_2,\tau)}{\theta_1(v_2,\tau)}\right)\!\left[ 2\frac{\theta_1^{\prime}(v_2,\tau)}{\theta_1(v_2,\tau)} -2\frac{\theta_1^{\prime}(2v_2,\tau)}{\theta_1(2v_2,\tau)}\right]\!
-32\frac{\omega^4}{\wp(v_0)\!-\!\wp(v_2)}.
\end{gather*}
Summing $(1)$ and $(2)+(5)$
\begin{gather*}
(1)+(2)+(5)=-4\big(t^2\big)^2\frac{\theta_1^{\prime}(v_2,\tau)}{\theta_1(v_2,\tau)}\frac{\partial}{\partial v_2}\left(\frac{\theta_1^{\prime}(v_2,\tau)}{\theta_1(v_2,\tau)}\right)\left[ 2\frac{\theta_1^{\prime}(v_2,\tau)}{\theta_1(v_2,\tau)} -2\frac{\theta_1^{\prime}(2v_2,\tau)}{\theta_1(2v_2,\tau)}\right].
\end{gather*}
 From the above results, we find
\begin{align*}
g^{11}=&(1)+(2)+(5)+(3)+(4)\\
=&-4\big(t^2\big)^2\frac{\theta_1^{\prime}(v_2,\tau)}{\theta_1(v_2,\tau)}\frac{\partial}{\partial v_2}\left(\frac{\theta_1^{\prime}(v_2,\tau)}{\theta_1(v_2,\tau)}\right)\left[ 2\frac{\theta_1^{\prime}(v_2,\tau)}{\theta_1(v_2,\tau)} -2\frac{\theta_1^{\prime}(2v_2,\tau)}{\theta_1(2v_2,\tau)}\right]\\
&+8\frac{\theta_1^{\prime^2}(v_2,\tau)}{\theta_1^2(v_2,\tau)}\big(t^2\big)^2\left[\frac{\theta_1^{\prime\prime}(2v_2,\tau)}{\theta_1(2v_2,\tau)}-\frac{\theta_1^{\prime^2}(2v_2,\tau)}{\theta_1^2(2v_2,\tau)}\right]\\
&-2\big(t^2\big)^2\left[\frac{\partial}{\partial v_2}\left(\frac{\theta_1^{\prime}(v_2,\tau)}{\theta_1(v_2,\tau)}\right)\right]^2-16\pi {\rm i}\big(t^2\big)^2\frac{\theta_1^{\prime}(v_2,\tau)}{\theta_1(v_2,\tau)}\frac{\partial}{\partial \tau}\left(\frac{\theta_1^{\prime}(v_2,\tau)}{\theta_1(v_2,\tau)}\right).
\end{align*}
Summarizing, we have proved the identities~\eqref{g23} and \eqref{g11}.

\subsection*{Acknowledgements}
I am grateful to Professor Boris Dubrovin for proposing this problem, for his remarkable advi\-ses and guidance. I would like also to thanks Professors Davide Guzzetti and Marco Bertola for~helpful discussions, and guidance of this paper. In addition, I~thank the anonymous referees for~their valuable comments and remarks, that have helped to improve the paper.

\pdfbookmark[1]{References}{ref}
\LastPageEnding

\end{document}